\documentclass[11pt]{article}

\usepackage[a4paper,left=20mm,top=20mm,right=20mm,bottom=25mm]{geometry}

\usepackage{amsmath, amsfonts, amssymb, amsthm}
\usepackage[noend]{algorithmic}
\usepackage[linesnumbered,ruled,vlined]{algorithm2e}

\usepackage{mathtools}
\usepackage[mathscr]{euscript}
\usepackage{wasysym}

\usepackage{graphicx}
\usepackage{xcolor}
\usepackage{enumitem}
\usepackage{authblk}
\usepackage{float}

\usepackage{natbib}
\setcitestyle{authoryear,open={(},close={)}}
\let\cite\citep

\usepackage{color}

\usepackage[colorlinks]{hyperref}
\hypersetup{
  citecolor=blue,
  linkcolor=blue,
}

\usepackage[font=footnotesize,width=.85\textwidth,labelfont=bf]{caption}

\newtheorem{theorem}{Theorem}
\newtheorem{lemma}{Lemma}
\newtheorem{corollary}{Corollary}
\newtheorem{definition}{Definition}
\newtheorem{proposition}{Proposition}
\newtheorem{fact}{Observation}
\newtheorem{remark}{Remark}

\def\arrowedvec{\mathaccent"017E}
\DeclareMathOperator{\lca}{lca}
\DeclareMathOperator{\aug}{\mathcal{A}}
\newcommand{\hourglass}{\mathrel{\text{\ooalign{$\searrow$\cr$\nearrow$}}}}
\newcommand{\fp}{\emph{fp}\xspace}
\newcommand{\ufp}{\emph{u-fp}\xspace}
\newcommand{\Scap}{\mathcal{S}^{\cap}}
\newcommand{\CIG}{\mathfrak{C}}
\newcommand{\G}{\arrowedvec{G}}  
\newcommand{\ROOT}{\circledcirc}
\newcommand{\LEAF}{\odot}
\newcommand{\SPEC}{\newmoon}

\newcommand{\DUPL}{\square}
\newcommand{\child}{\mathsf{child}}

\newcommand{\tT}{{\widehat{t}_T}}
\newcommand{\NH}{\ensuremath{\mathbb{NH}}}
\newcommand{\tTp}{{\widehat{t}_{\aug(T)}}}
\newcommand{\tTps}{{\widehat{t}_{\aug(T^*)}}}
\newcommand{\wt}{{\widehat{t}}}
\newcommand{\AX}[1]{\textnormal{#1}}

\newtheorem{innercustomEthm}{}
\newenvironment{emptyTHM}[1]
{\renewcommand\theinnercustomEthm{\hspace{-0.1cm}#1}\innercustomEthm}
{\endinnercustomEthm}

\newtheorem{innercustomthm}{Theorem}
\newenvironment{ctheorem}[1]
{\renewcommand\theinnercustomthm{#1}\innercustomthm}
{\endinnercustomthm}

\newtheorem{innercustomlem}{Lemma}
\newenvironment{clemma}[1]
{\renewcommand\theinnercustomlem{#1}\innercustomlem}
{\endinnercustomlem}

\newtheorem{innercustomprop}{Proposition}
\newenvironment{cproposition}[1]
{\renewcommand\theinnercustomprop{#1}\innercustomprop}
{\endinnercustomprop}

\newtheorem{innercustomdef}{Definition}
\newenvironment{cdefinition}[1]
{\renewcommand\theinnercustomdef{#1}\innercustomdef}
{\endinnercustomdef}

\newtheorem{innercustomcor}{Corollary}
\newenvironment{ccorollary}[1]
{\renewcommand\theinnercustomcor{#1}\innercustomcor}
{\endinnercustomcor}

\newtheorem{innercustomfact}{Observation}
\newenvironment{cfact}[1]
{\renewcommand\theinnercustomfact{#1}\innercustomfact}
{\endinnercustomfact}

\providecommand{\keywords}[1]{\textbf{\textit{Keywords: }} #1}

\title{Complete Characterization of Incorrect Orthology Assignments in
  Best Match Graphs}

\author[1,2]{David Schaller}

\author[3]{Manuela Gei{\ss}} 
\author[1,2,4-7]{Peter F. Stadler}    

\author[8,*]{Marc Hellmuth} 

\affil[1]{Max Planck Institute for Mathematics in the Sciences,
  Inselstra{\ss}e 22, D-04103 Leipzig, Germany}

\affil[2]{Bioinformatics Group, Department of Computer Science \&
  Interdisciplinary Center for Bioinformatics, Universit{\"a}t Leipzig,
  H{\"a}rtelstra{\ss}e~16--18, D-04107 Leipzig, Germany.}

\affil[3]{Software Competence Center Hagenberg GmbH, Softwarepark 21, A-4232 
Hagenberg, Austria}

\affil[4]{German Centre for Integrative Biodiversity Research
  (iDiv) Halle-Jena-Leipzig, Competence Center for Scalable Data Services
  and Solutions Dresden-Leipzig, Leipzig Research Center for Civilization
  Diseases, and Centre for Biotechnology and Biomedicine at Leipzig
  University at Universit{\"a}t Leipzig}

\affil[5]{Institute for Theoretical Chemistry, University of Vienna,
  W{\"a}hringerstrasse 17, A-1090 Wien, Austria}

\affil[6]{Facultad de Ciencias, Universidad National de Colombia, Sede
  Bogot{\'a}, Colombia}

\affil[7]{Santa Fe Insitute, 1399 Hyde Park Rd., Santa Fe NM 87501,
  USA}

\affil[8]{School of Computing, University of Leeds, EC Stoner
  Building, Leeds LS2 9JT, UK \newline \texttt{mhellmuth@mailbox.org}}

\affil[*]{corresponding author}

\date{\ }

\begin{document}

\maketitle 

\abstract{  
  Genome-scale orthology assignments are usually based on reciprocal best
  matches.  In the absence of horizontal gene transfer (HGT), every pair of
  orthologs forms a reciprocal best match.  Incorrect orthology assignments
  therefore are always false positives in the reciprocal best match graph.
  We consider duplication/loss scenarios and characterize unambiguous
  false-positive (\ufp) orthology assignments, that is, edges in the best
  match graphs (BMGs) that cannot correspond to orthologs for any gene tree
  that explains the BMG. Moreover, we provide a polynomial-time algorithm
  to identify all \ufp orthology assignments in a BMG.  Simulations show
  that at least $75\%$ of all incorrect orthology assignments can be
  detected in this manner. All results rely only on the structure of the
  BMGs and \emph{not} on any \emph{a priori} knowledge about underlying
  gene or species trees.
}

\bigskip
\noindent
\keywords{
  orthology detection,
  best matches,
  unambiguous orthologs,
  colored graphs,
  cograph,
  tree reconciliation,
  polynomial-time algorithm}

\sloppy

\section{Introduction}
\label{sec:intro}

Orthology is one of the key concepts in evolutionary biology: Two genes are
orthologs if their last common ancestor was a speciation event
\cite{Fitch:70}. Distinguishing orthologs from paralogs (originating from
gene duplications) or xenologs (i.e., genes that have undergone horizontal
gene transfer) is of considerable practical importance for functional
genome annotation and thus for a wide array of methods in bioinformatics
and computational biology that rely on gene annotation data.  In
particular, according to the ``ortholog conjecture'', orthologous genes in
different species are expected to have essentially the same biological and
molecular functions, whereas paralogs and xenologs tend to have similar,
but distinct functions. Albeit controversial
\cite{Nehrt:11,Stamboulian:20}, this assumption is widely made in the
computational prediction of gene functions
\cite{Nehrt:11,Gabaldon:13,Soria:14,Zallot:16}. Moreover, the distinction
of orthologs and paralogs is crucial in phylogenomics \cite{Delsuc:05}.
Most of the commonly used tools for large-scale orthology identification
compute reciprocal best hits as a first step followed by some filtering and
refinement steps to improve the results
\cite{tatusov2000cog,Roth:08,Lechner:11a,
  LTPL:11,sonnhammer2015inparanoid,TGG+17,HSH+18}, see also
\cite{Nichio:17,Setubal:18a,Galperin:19} for reviews and
\cite{Altenhoff:16} for benchmarking results.

Orthology identification has also received increasing attention from a
mathematical perspective starting from the concept of an \emph{evolutionary
  scenario} comprising a gene tree $T$ and a species tree $S$ together with
a \emph{reconciliation map $\mu$} from $T$ to $S$. The map $\mu$ identifies
the locations in the species tree at which evolutionary events, represented
by the vertices of the gene tree, took place.  \emph{In this contribution,
  we consider exclusively duplication/loss scenarios, i.e., we explicitly
  exclude horizontal gene transfer.} Characterizations of reconciliation
maps are given e.g.\ in
\cite{Gorecki:06,Vernot:08,Doyon:11,Rusin:14}. While every gene tree can be
reconciled with any species tree \cite{Guigo:96,Page:97}, this is no longer
true if event-labels are prescribed in the gene tree~$T$
\cite{HernandezRosales:12a,Lafond:14,Hellmuth:17}.

The orthology relation itself has been characterized as a cograph (i.e.,
graphs that do not contain induced paths $P_4$ on four vertices) by
\citet{Hellmuth:13a} based on earlier work by \citet{Boecker:98}.  This
line of research has led to the idea of editing reciprocal best hit data to
conform to the required cograph structure \cite{Hellmuth:15}. There are,
however, two distinct sources of errors in an orthology assignment pipeline
based on best matches:
\begin{description}
  \item[(i)] inaccuracies in the assignment of best matches from sequence
  similarity data \cite{Stadler:20a}, and
  \item[(ii)] limits in the reconstruction of the ``true'' orthology relation
  from best match graphs \cite{Geiss:20a}.
\end{description}
We consider best matches as an evolutionary concept: A gene $y$ in species
$s$ is a best match of a gene $x$ from species $r\ne s$ if $s$ contains no
gene $y'$ that is more closely related to $x$. That is, best matches
capture the idea of phylogenetically most closely related genes. Maybe
surprisingly, the combinatorial structure of best matches has become a
focus only very recently \cite{Geiss:19a}. Best match graphs (BMGs) have
several appealing properties: They have several alternative
characterizations providing polynomial-time recognition algorithms
\cite{BMG-corrigendum,Schaller:20d} and they are ``explained'' by a unique
least resolved tree \cite{Geiss:19a}. These properties will be introduced
formally in the next section and play an important role in our
discussion. The reciprocal best match graphs (RBMGs) are the symmetric
parts of BMGs and conceptually correspond to the reciprocal best hits used
in orthology detection. In contrast to BMGs, RBMGs are much more difficult
to handle and are not associated with unique trees \cite{Geiss:19b}.  An
example for an evolutionary scenario with corresponding BMG and RBMG is
given Fig.~\ref{fig:bmg_example}.

\begin{figure}[t]
  \begin{center}
    \includegraphics[width=0.85\linewidth]{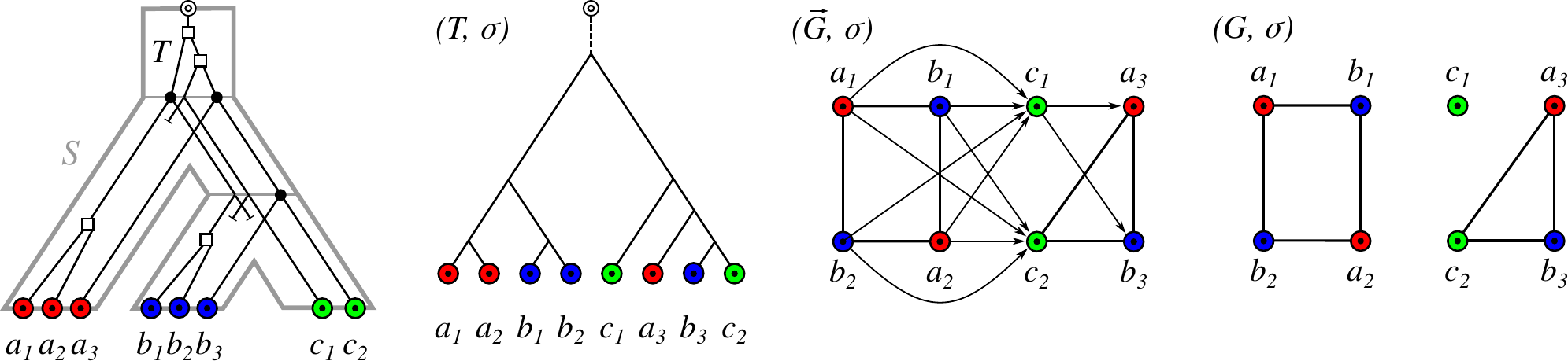}
  \end{center}
  \caption[]{An evolutionary scenario (left) consists of a gene tree
    $(T,\sigma)$ (whose observable part is shown in the second panel)
    together with an embedding into a species tree $S$.  The coloring
    $\sigma$ of the leaves of $T$ represents the species in which the genes
    reside.  Speciation vertices ($\SPEC$) of the gene tree coincide with
    the vertices of the species tree, whereas gene duplications ($\DUPL$)
    are mapped to the edges of $S$. The reciprocal best match graph (RBMG)
    $(G,\sigma)$ on the right corresponds to the undirected graph
    underlying the symmetric part of the best match graph (BMG)
    $(\G,\sigma)$ (third panel).}
  \label{fig:bmg_example}
\end{figure}

In this contribution, we are only concerned with the second source of
errors, i.e., with the limits in the reconstruction of the true orthology
relation from best matches.  We therefore assume throughout that a
``correct'' BMG (cf.\ Def.~\ref{def:BestMatchGraph}) is given. \emph{We do
  not assume, however, that we have any \emph{a priori} knowledge about the
  underlying gene or species tree}. The problem we aim to solve is to
determine the orthology relation that is best supported by the given BMG.

Of course, the \emph{true} orthology relation is not known. Nevertheless,
we start our mathematical analysis with the following definition: A pair of
genes $x$ and $y$ that are not true orthologs but reciprocal best matches
are false-positive orthologs. If they are orthologs but not reciprocal best
matches, they are false-negative orthologs. \citet{Geiss:20a} showed that,
for evolutionary scenarios that involve only speciations, gene
duplications, and gene losses, there are no false-negative orthology
assignments (see also Thm.~\ref{thm:extrem-ortho} below).  Our task
therefore reduces to understanding the false-positive orthology
assignments. Being a false positive is a property of the edge $xy$ in an
RBMG, and equivalently of the symmetric pair $(x,y)$ and $(y,x)$ in the
BMG. Here, we aim to identify false-positive edges from the structure of
the BMG itself.

\begin{figure}[t]
  \begin{center}
    \includegraphics[width=0.85\linewidth]{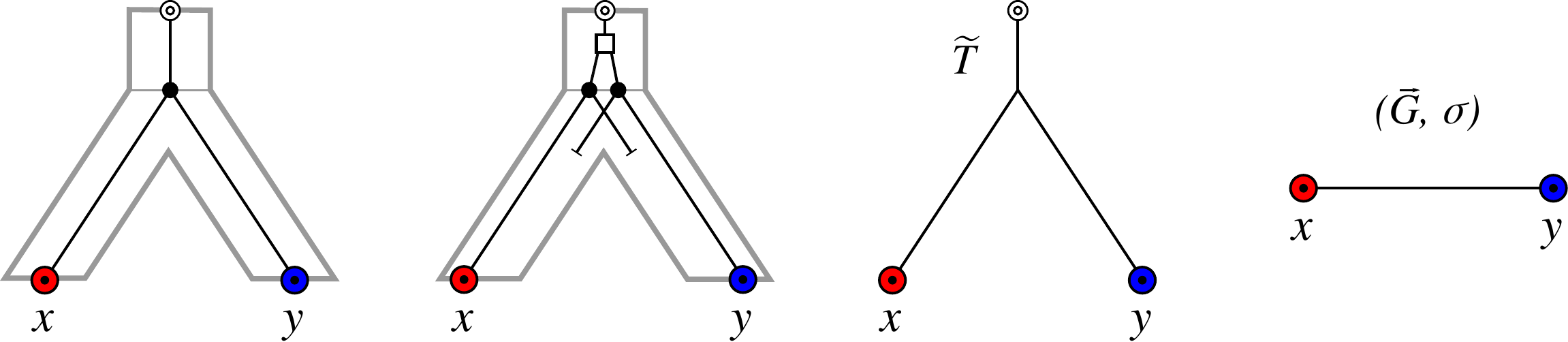}
  \end{center}
  \caption[]{Two scenarios (1st and 2nd panel to the left) for the
    evolution of a gene family embedded into a species tree (shown in
    gray), where $\SPEC$ represents speciation and $\DUPL$ duplication
    events.  The second scenario is the simplest example for a
    complementary gene loss that is not witnessed by any other species. In
    particular, the two different true histories result in the same
    topology $\widetilde{T}$ of the true (loss-free) gene tree, and thus
    explain the same BMG $(\G,\sigma)$.  However, only for the leftmost
    scenario the edge $xy$ in $(\G,\sigma)$ describes correct orthologs.}
  \label{fig:compl_loss}
\end{figure}

We first note that false positives cannot be avoided altogether, i.e., not
all false positives can be identified from a BMG alone.  The simplest
example, Fig.~\ref{fig:compl_loss} (second scenario), comprises a gene
duplication and a subsequent speciation and complementary gene losses in
the descendant lineages such that each paralog survives only in one of
them. In this situation, $xy$ is a reciprocal best match. If there are no
other descendants that harbor genes witnessing the duplication event, then
the framework of best matches provides no information to recognize $xy$ as
a false-positive assignment.

On the other hand, RBMGs and thus BMGs contain at least some information on
false positives. Since the orthology relation forms a cograph but RBMGs are
not cographs in general \cite{Geiss:19b}, incorrect orthology assignments
are associated with induced $P_4$s, the forbidden subgraphs that
characterize cographs. $P_4$s arise for instance as a consequence of the
complete loss of different paralogous groups in disjoint lineages.
\citet{Dessimoz:2006} noted that such false-positive orthology assignments
can be identified under certain circumstances, in particular, if there is
some species in which both paralogs have survived. The corresponding motif
in BMGs, the ``good quartets'', was investigated in some detail by
\citet{Geiss:19b}.  The removal of such false-positive orthologs already
leads to a substantial improvement of the orthology assignments in
simulated data \cite{Geiss:20a}.  Here, we extend the results of
\citet{Geiss:20a} to a complete characterization of false-positive
orthology assignments for a given BMG.

Good quartets cannot be defined on RBMGs because information on
non-reciprocal best matches is also needed explicitly.  This suggests to
consider BMGs rather than RBMGs as the first step in graph-based orthology
detection methods. In practice, best matches are approximated by sequence
similarity and thus are subject to noise and biases \cite{Stadler:20a}. The
empirically determined best match relation thus will usually need to be
corrected to conform to the formal definition (cf.\
Def.~\ref{def:BestMatchGraph} below) of BMGs. This naturally leads to a
graph editing problem that was recently shown to be NP-complete
\cite{Schaller:20d,Hellmuth:20a}.

Sec.~\ref{sec:prelim} establishes the notation and summarizes properties of
BMGs that are needed throughout this contribution.
Sec.~\ref{sec:false-positives} formalizes the notion of \emph{unambiguous
  false-positive} (\ufp) edges, i.e., reciprocal best matches that cannot
be orthologs w.r.t.\ to \emph{any} gene tree explaining the BMG.
Sec.~\ref{sec:main-results} contains the main mathematical contributions of
this work:
\begin{enumerate}
  \item We provide a full characterization of unambiguous false-positive
  orthology assignments in BMGs.
  \item We provide a polynomial-time algorithm to determine all unambiguous
  false-positive orthology assignments in BMGs.
\end{enumerate}
In Sec.~\ref{sec:simulations}, we complement the mathematical results with
a computational analysis of simulated scenarios and observe that at least
three quarters of all false positives fall into this class. The remaining
cases are not recognizable from best matches alone and correspond to
complementary losses without surviving witnesses, i.e., cases that cannot
be corrected without additional knowledge on the gene tree and/or the
species tree.

Since the material is extensive and very technical, we subdivide our
presentation into a main narrative part
(Secs.~\ref{sec:intro}--\ref{sec:summary}) and a technical part
(Secs.~\ref{APP:subsect:rbmg}--\ref{APP:sect:leftovers}) that contains all
proofs and additional material in full detail. Together with the
definitions and preliminaries in Sec.~\ref{sec:prelim}, the technical part
is self-contained. Definitions and results appearing in the narrative part
are therefore restated. The order of the material in the two parts is
slightly different.

\section{Preliminaries}
\label{sec:prelim}

\subsection{Graphs and trees}

We consider finite, directed graphs $\G=(V,E)$, for brevity just called
graphs throughout, with arc set
$E\subseteq V\times V\setminus\{(v,v)\mid v\in V\}$. We say that $xy$ is an
\emph{edge} in $\G$ if and only if both $(x,y)\in E(\G)$ and
$(y,x)\in E(\G)$.  If all arcs of $\G$ in a graph form edges, we call $\G$
\emph{undirected}.  A graph $H=(W,F)$ is a \emph{subgraph} of $G=(V,E)$, in
symbols $H\subseteq G$, if $W\subseteq V$ and $F\subseteq E$.  The
underlying \emph{symmetric part} of a directed graph $\G=(V,E)$ is the
subgraph $G=(V,F)$ that contains all edges of $\G$.  A subgraph $H=(W,F)$
(of $\G$) is called \emph{induced}, denoted by $\G[W]$, if for all
$u,v\in W$ it holds that $(u,v) \in E$ implies $(u,v) \in F$.  In addition,
we consider \emph{vertex-colored} graphs $(\G,\sigma)$ with vertex-coloring
$\sigma\colon V\to M$ into some set $M$ of colors. A vertex-coloring is
called \emph{proper} if $\sigma(x)\neq \sigma(y)$ for every arc $(x,y)$ in
$\G$. We write $\sigma(W) = \{\sigma(w) \mid w\in W\}$ for subsets
$W\subseteq V$ and $\sigma_{|W}$ to denote the restriction of the map
$\sigma$ to $W\subseteq V$.  In particular, $(\G[W],\sigma_{|W})$ is an
induced vertex-colored subgraph of $(\G,\sigma)$.

A \emph{path (of length $\ell$)} in a directed graph $\G$ or an undirected
graph $G$ is a subgraph induced by a nonempty sequence of pairwise distinct
vertices $P(x_0,x_{\ell}) \coloneqq (x_0, x_1, \dots, x_{\ell})$ such that
$(x_i, x_{i+1}) \in E(\G)$ or $x_ix_{i+1} \in E(G)$, resp., for
$0 \leq i \leq \ell-1$. We use the notation $P(x_0,x_{\ell})$ both for the
sequence of vertices and the subgraph they induce.

All \emph{trees} $T=(V,E)$ considered here are \emph{undirected},
\emph{planted} and \emph{phylogenetic}, that is, they satisfy (i) the root
$0_T$ has degree $1$ and (ii) all inner vertices have degree
$\deg_T(u)\ge 3$. We write $L(T)$ for the leaves (not including $0_T$) and
$V^0=V(T)\setminus(L(T)\cup\{0_T\})$ for the inner vertices (also not
including $0_T$). To avoid trivial cases, we will always assume
$|L(T)|\geq 2$. An edge $uv$ in $T$ is an inner edge if $u,v\in V^0(T)$ are
inner vertices.  The \emph{conventional root} $\rho_T$ of $T$ is the unique
neighbor of $0_T$. The main reason for using planted phylogenetic trees
instead of modeling phylogenetic trees simply as rooted trees, which is the
much more common practice in the field, is that we will often need to refer
to the time before the first branching event, i.e., the edge $0_T\rho_T$.

We define the \emph{ancestor order} on a given tree $T$ as follows: if $y$
is a vertex of the unique path connecting $x$ with the root $0_T$, we write
$x\preceq_T y$, in which case $y$ is called an ancestor of $x$ and $x$ is
called a descendant of $y$. We use $x \prec_T y$ for $x \preceq_{T} y$ and
$x \neq y$. If $x \preceq_{T} y$ or $y \preceq_{T} x$ the vertices $x$ and
$y$ are \emph{comparable} and, otherwise, \emph{incomparable}. If $xy$ is
an edge in $T$, such that $y \prec_{T} x$, then $x$ is the \emph{parent} of
$y$ and $y$ the \emph{child} of $x$. We denote by $\child_T(x)$ the set of
all children of $x$. It will be convenient for the discussion below to
extend the ancestor relation $\preceq_T$ to the union of the edge and
vertex sets of $T$. More precisely, for a vertex $x\in V(T)$ and an edge
$e=uv\in E(T)$ with $v\prec_T u$ we write $x \prec_T e$ if and only if
$x\preceq_T v$ and $e \prec_T x$ if and only if $u\preceq_T x$. For edges
$e=uv$ with $v\prec_T u$ and $f=ab$ with $b\prec_T a$ in $T$ we put
$e\preceq_T f$ if and only if $v \preceq_T b$.

For a non-empty subset $A\subseteq V\cup E$, we define $\lca_T(A)$, the
\emph{last common ancestor of $A$}, to be the unique $\preceq_T$-minimal
vertex of $T$ that is an ancestor of every vertex or edge in $A$.  For
simplicity we drop the brackets and write
$\lca_T(x_1,\dots,x_k)\coloneqq\lca_T(\{x_1,\dots,x_k\})$ whenever we
specify a set of vertices or edges explicitly.

A vertex $v\in V(T)$ is \emph{binary} if $\deg_T(v)=3$, i.e., if $v$ has
exactly two children. A tree is \emph{binary}, if all vertices $v\in V^0$
are binary.  For $v\in V(T)$ we denote by $T(v)$ the subtree of $T$ rooted
in $v$. The set of \emph{clusters} of a tree $T$ is
$\mathscr{C}(T) = \{L(T(v))\mid v\in V(T)\}$. It is well-known that
$\mathscr{C}(T)$ uniquely determines $T$ \cite{sem-ste-03a}.  We say that a
tree $T$ is a \emph{refinement} of some tree $T'$ if
$\mathscr{C}(T')\subseteq \mathscr{C}(T)$.  A tree $T'$ is \emph{displayed}
by a tree $T$, in symbols $T'\le T$, if $T'$ can be obtained from a subtree
of $T$ by contraction of edges \cite{Semple:03}, where the contraction of
an edge $e = uv$ in a tree $T = (V ,E)$ refers to the removal of $e$ and
identification of $u$ and $v$. It is easy to verify that every refinement
$T$ of $T'$ also displays $T'$. However, the converse is not always true
since $L(T')\subsetneq L(T)$ and thus,
$\mathscr{C}(T')\not\subseteq \mathscr{C}(T)$ may be possible.

\subsection{(Reciprocal) best matches}

We consider a pair $T=(V,E)$ and $S=(W,F)$ of planted phylogenetic trees
together with a map $\sigma\colon L(T)\to L(S)$.  We interpret $T$ as a
\emph{gene tree} and $S$ as a \emph{species tree}; the map $\sigma$
describes, for each gene $x\in L(T)$, in the genome of which species
$\sigma(x)\in L(S)$ it resides.  W.l.o.g.\ we assume that the
``gene-species-association'' $\sigma$ is a surjective map to avoid trivial
cases. Since $\sigma$ can be viewed as a coloring of the leaves of $T$, we
call $(T,\sigma)$ a \emph{leaf-colored tree}. For $s\in L(S)$ we write
$L[s]:=\{x\in L(T)|\sigma(x)=s\}$.

\begin{definition}
  Let $(T,\sigma)$ be a leaf-colored tree. A leaf $y\in L(T)$ is a
  \emph{best match} of the leaf $x\in L(T)$ if $\sigma(x)\neq\sigma(y)$ and
  $\lca(x,y)\preceq_T \lca(x,y')$ holds for all leaves $y'$ from species
  $\sigma(y')=\sigma(y)$.  The leaves $x,y\in L(T)$ are \emph{reciprocal
    best matches} if $y$ is a best match for $x$ and $x$ is a best match
  for $y$.
\end{definition}
Neither best matches nor reciprocal best matches are unique. That is, a
gene $x$ may have two or more (reciprocal) best matches of the same color
$r\neq \sigma(x)$. Some orthology detection tools, such as
\texttt{ProteinOrtho} \cite{Lechner:11a}, explicitly attempt to extract all
reciprocal best matches from the sequence data.  Moreover, neither of the
two relations is transitive. These two properties are at odds e.g.\ with
the \emph{clusters of orthologous groups} (COGs) concept
\cite[cf.][]{Tatusov1997,tatusov2000cog,Roth:08}, which at least
conceptually presupposes unique reciprocal best matches.

The graph $\G(T,\sigma) = (V,E)$ with vertex set $V=L(T)$, vertex coloring
$\sigma$, and with arcs $(x,y)\in E$ if and only if $y$ is a best match of
$x$ w.r.t.\ $(T,\sigma)$ is known as the (colored) \emph{best match
  graph} of $(T,\sigma)$ \cite{Geiss:19a}. The symmetric part $G(T,\sigma)$
of $\G(T,\sigma)$ obtained by retaining the edges of $\G(T,\sigma)$ is the
(colored) \emph{reciprocal best match graph} \cite{Geiss:19b}.

\begin{definition}\label{def:BestMatchGraph}
  An arbitrary vertex-colored graph $(\G,\sigma)$ is a \emph{best match
    graph (BMG)} if there exists a leaf-colored tree $(T,\sigma)$ such that
  $(\G,\sigma) = \G(T,\sigma)$. In this case, we say that $(T,\sigma)$
  \emph{explains} $(\G,\sigma)$.  An arbitrary undirected vertex-colored
  graph $(G,\sigma)$ is a \emph{reciprocal best match graph (RBMG)} if it
  is the symmetric part of a BMG $(\G,\sigma)$.
\end{definition}

For the symmetric part of the BMG $(\G,\sigma)$, i.e., the RBMG
$(G,\sigma)$, we have $xy\in E(G)$ if and only if $x$ and $y$ are
reciprocal best matches in $(T,\sigma)$. In this sense, $(T,\sigma)$ also
explains $(G,\sigma)$.  We note, furthermore, that RBMGs are not associated
with a unique least resolved tree \cite{Geiss:19b}.

\subsection{Reconciliation maps, event-labeling, and orthology relations}

An \emph{evolutionary scenario} extends the map $\sigma\colon L(T)\to L(S)$
to an embedding of the gene tree into the species tree. It (implicitly)
describes different types of evolutionary events: speciations, gene
duplications, and gene losses. In this contribution we do not consider
other types of events such as horizontal gene transfer. Gene losses do not
appear explicitly since $L(T)$ only contains extant genes.  Inner vertices
in the gene tree $T$ that designate speciations have their correspondence
in inner vertices of the species tree. In contrast, gene duplications occur
independently of speciations and thus belong to edges of the species tree.
The embedding of $T$ into $S$ is formalized by
\begin{definition}[Reconciliation Map]
  Let $S=(W,F)$ and $T=(V,E)$ be two planted phylogenetic trees and let
  $\sigma\colon L(T) \to L(S)$ be a surjective map. A reconciliation from
  $(T,\sigma)$ to $S$ is a map $\mu\colon V \to W \cup F$ satisfying
  \begin{description}
    [itemsep=0.2ex, parsep=0cm, topsep=0.7ex,]
    \item[\emph{(R0)}] \emph{Root Constraint.} $\mu(x) = 0_S$ if and only if
    $x=0_T$.
    \item[\emph{(R1)}] \emph{Leaf Constraint.} If $x \in L(T)$, then
    $\mu(x)=\sigma(x)$.
    \item[\emph{(R2)}] \emph{Ancestor Preservation.} If $x \prec_T y$, then
    $\mu(x) \preceq_S \mu(y)$.
    \item[\emph{(R3)}] \emph{Speciation Constraints.} Suppose
    $\mu(x) \in W^0$ for some $x\in V$. Then
    \begin{enumerate}[label=(\roman*), itemsep=0.2ex, topsep=0.2ex,
      parsep=0cm]
      \item $\mu(x)=\lca_S(\mu(v'),\mu(v''))$ for at least two distinct
      children $v',v''$ of $x$ in $T$.
      \item $\mu(v')$ and $\mu(v'')$ are incomparable in $S$ for any two
      distinct children $v'$ and $v''$ of $x$ in $T$.
    \end{enumerate}
  \end{description}
  \label{def:reconc_map}
\end{definition}

Several alternative definitions of reconciliation maps for duplication/loss
scenarios have been proposed in the literature, many of which have been
shown to be equivalent. This type of reconciliation map has been
established in \cite{Geiss:20a}. Moreover, it has been shown in
\cite{Geiss:20a} that the axiom set used here is equivalent to axioms that
are commonly used in the literature, see e.g.\
\cite{Gorecki:06,Vernot:08,Doyon:11,Rusin:14,Hellmuth:17,Nojgaard:18a}, and
the references therein. Without any further constraints,
Def.~\ref{def:reconc_map} gives rise to a well-known result:
\begin{lemma}{\cite[Lemma~3]{Geiss:20a}}
  For every tree $(T, \sigma)$ there is a reconciliation map $\mu$ to any
  species tree $S$ with leaf set $L(S) = \sigma (L(T ))$.
  \label{lem:reconAll}
\end{lemma}

The reconciliation map $\mu$ from $(T,\sigma)$ to $S$ determines 
the types of evolutionary events in $T$. This can be formalized by
associating an event labeling with the vertices of $T$. We use the notation
introduced in \cite{Geiss:20a}:
\begin{definition}
  Given a reconciliation map $\mu$ from $(T,\sigma)$ to $S$, the
  \emph{event labeling on $T$ (determined by $\mu$)} is the map
  $t_\mu:V(T)\to \{\ROOT,\LEAF,\SPEC,\DUPL\}$ given by:
  \begin{equation*}
  t_\mu(u) = \begin{cases}
  \ROOT & \, \text{if } u=0_T \text{, i.e., } \mu(u)=0_S  \text{ (root)}\\
  \LEAF & \, \text{if } u\in L(T) \text{, i.e., } \mu(u)\in L(S)
  \text{ (leaf)}\\
  \SPEC & \, \text{if } \mu(u)\in V^0(S)  \text{ (speciation)}\\
  \DUPL & \, \text{else, i.e., } \mu(u)\in E(S) \text{ (duplication)}\\
  \end{cases}
  \end{equation*}
  \label{def:event-rbmg}
\end{definition}

The following result is a simple but useful consequence of combining the
axioms of the reconciliation map with the event labeling of
Def.~\ref{def:event-rbmg}.
\begin{lemma}{\cite[Lemma~3]{Geiss:20a}}
  Let $\mu$ be a reconciliation map from $(T,\sigma)$ to a tree $S$ and
  suppose that $u\in V(T)$ is a vertex with $\mu(u)\in V^0(S)$ and thus,
  $t(\mu(u))=\SPEC$.  Then,
  $\sigma(L(T(v_1)))\cap \sigma(L(T(v_2))) = \emptyset$ for any two
  distinct $v_1,v_2\in \child(u)$.
  \label{lem:duplication_witness}
\end{lemma}
We will regularly make use of the observation that, by contraposition of
Lemma~\ref{lem:duplication_witness},
$\sigma(L(T(v)))\cap \sigma(L(T(v'))) \ne \emptyset$ for two distinct
$v_1,v_2\in \child(u)$ implies that $\mu(u)\in E(S)$, and thus
$t_{\mu}(u)=\DUPL$.

Lemma~\ref{lem:duplication_witness} suggests to define \emph{event-labeled
  trees} as trees $(T,t)$ endowed with a map
$t: V(T)\to \{\ROOT,\LEAF,\SPEC,\DUPL\}$ such that $t(0_T)=\ROOT$ and
$t(u)=\LEAF$ for all $u\in L(T)$.  In \cite{Geiss:20a},
Lemma~\ref{lem:duplication_witness} also served as a motivation for
\begin{definition}
  Let $(T,\sigma)$ be a leaf-colored tree.  The \emph{extremal
    event labeling} of $T$ is the map
  $\tT:V(T)\to\{\ROOT,\LEAF,\SPEC,\DUPL\}$ defined for $u\in V(T)$ by
  \begin{equation*}
  \tT(u) = \begin{cases}
  \ROOT & \, \text{if } u=0_{T} \\
  \LEAF & \, \text{if } u\in L(T) \\
  \DUPL & \, \text{if there are two children } v_1,v_2\in \child(u)
  \text{ such that}\\
  & \qquad \sigma(L(T(v_1)))\cap \sigma(L(T(v_2)))\neq\emptyset\\
  \SPEC & \, \text{otherwise} \\
  \end{cases}
  \end{equation*}
  \label{def:extremal_labeling}
\end{definition} 
An example of an extremal event labeling is shown in
Fig.~\ref{fig:contradictory_triples} (rightmost tree).  The extremal event
labeling is closely related to the concept of apparent duplication (AD)
vertices often found in the literature
\cite[e.g.][]{Swenson:12,Lafond:14b}.  For a (binary) gene tree $T$ and a
reconciliation of $T$ with a species tree $S$, a duplication vertex of $T$
is an AD vertex if its two subtrees have at least one color in common.  In
contrast, it is a non-apparent duplication (NAD) vertex if the color sets
of its subtrees are disjoint.  This notion is useful for a variety of
parsimony problems that usually aim to avoid or minimize the number of NAD
vertices \cite{Swenson:12,Lafond:14b}.  However, the extremal event
labeling $\tT$ is completely defined by $(T,\sigma)$. That is, in contrast
to both the event labeling in Def.~\ref{def:event-rbmg} and the concept of
AD and NAD vertices, $\tT$ does not depend on a specific reconciliation
map. On the other hand, there is no guarantee that there always exists a
reconciliation map $\mu$ from $(T,\sigma)$ to some species tree $S$ such
that $t_{\mu} = \tT$, cf.\ \cite[Fig.~2]{Geiss:20a} and
Fig.~\ref{fig:contradictory_triples} in Sec.~\ref{ssect:algorithms} for
counterexamples. Nevertheless, we shall see below that the extremal
labeling is a key step towards identifying false-positive orthology
assignments.

The event labeling on $T$ defines the orthology graph. 
\begin{definition}
  The \emph{orthology graph} $\Theta(T,t)$ of an event-labeled tree $(T,t)$
  has vertex set $L(T)$ and edges $uv\in E(\Theta)$ if and only if
  $t(\lca(u,v))=\SPEC$.
  \label{def:ortho-graph}
\end{definition}
The orthology graph is often referred to as the orthology relation.
Orthology graphs coincide with a well-known graph class:
\begin{theorem}{\cite[Cor.~4]{Hellmuth:13a}}
  A graph $G$ is an orthology graph for some event-labeled tree $(T,t)$,
  i.e.\ $G=\Theta(T,t)$, if and only if $G$ is a cograph.
  \label{thm:ortho-cograph}
\end{theorem}
One of many equivalent characterizations of cographs identifies them with
the graphs that do not contain an induced path $P_4$ on four vertices
\cite{CORNEIL:81}.

The orthology graph is a subgraph of the RBMG (and thus also of the BMG)
for any given reconciliation map connecting a gene with a species tree.
\begin{theorem}{\cite[Lemma~4~\&~5]{Geiss:20a}}
  Let $(T,\sigma)$ be a leaf-colored tree and $\mu$ a reconciliation map
  from $(T,\sigma)$ to some species tree $S$. Then
  $\Theta(T,t_{\mu}) \subseteq \Theta(T,\tT)\subseteq G(T,\sigma) \subseteq
  \G(T,\sigma)$.
  \label{thm:extrem-ortho}
\end{theorem}
In particular, $t_{\mu}(v) =\SPEC$ implies $\tT(v) =\SPEC$ for any
reconciliation map. By contraposition, therefore, if $\tT(v) =\DUPL$ then
$t_{\mu}(v) =\DUPL$ for all possible reconciliation maps $\mu$ from
$(T,\sigma)$ to any species tree $S$.  A crucial implication of
Thm.~\ref{thm:extrem-ortho} is that edges in a BMG $\G(T,\sigma)$ always
correspond to either correct orthologous pairs of genes or false-positive
orthology assignments.  Hence, $\G(T,\sigma)$ never contains false-negative
orthology assignments.

\section{False-positive orthology assignments}
\label{sec:false-positives}

As discussed in the introduction, we are not concerned here with the errors
that arise in the reconstruction of best matches from sequence similarity
data. We therefore assume that we are given a BMG $(\G,\sigma)$ as
specified in Def.~\ref{def:BestMatchGraph}. More precisely, we assume that
$(\G,\sigma)$ derives from a duplication/loss scenario that is unknown to
us. Denote by $(\widetilde{T},\widetilde{t},\sigma)$ the corresponding true
leaf-colored and event-labeled gene tree.  An edge $xy$ of $(\G,\sigma)$,
or equivalently of the corresponding RBMG $(G,\sigma)$, is a false-positive
orthology assignment if $xy\in E(G)$ but
$xy\notin E(\Theta(\widetilde{T},\widetilde{t}))$.  By
Thm.~\ref{thm:extrem-ortho}, $(G,\sigma)$ cannot contain false-negative
orthology assignments, i.e., there is no
$xy\in E(\Theta(\widetilde{T},\widetilde{t}))$ with $xy\notin E(G)$. We
assume no additional information about the gene tree or the species tree,
i.e., the only data about the evolutionary scenario that is available to us
is the BMG $(\G,\sigma)$.

In order to study false-positive orthology assignments, we first consider a
tree $(T,\sigma)$ that explains the BMG $(\G,\sigma)$. We neither make the
assumption that $(T,\sigma)$ is least resolved nor that $(T,\sigma)$
reflects the true history, i.e., that $(T,\sigma)$ is related to the true
gene tree $(\widetilde{T},\sigma)$.
\begin{cdefinition}{\ref{def:Ts-fp}}[$\mathbf{(T,\sigma)}$-false-positive]
  Let $(T,\sigma)$ be a tree explaining the BMG $(\G,\sigma)$.  An edge
  $xy$ in $\G$ is called \emph{$(T,\sigma)$-false-positive}, or
  $(T,\sigma)$-\fp for short, if for every reconciliation map $\mu$ from
  $(T,\sigma)$ to any species tree $S$ we have $t_\mu(\lca_T(x,y))=\DUPL$,
  i.e., $\mu(\lca_T(x,y))\in E(S)$,
\end{cdefinition}
In other words, $xy$ is called $(T,\sigma)$-\fp whenever $x$ and $y$ cannot
be orthologous w.r.t.\ any possible reconciliation $\mu$ from $(T,\sigma)$
to any species tree. Interestingly, $(T,\sigma)$-\fp{}s can be identified
without considering reconciliation maps explicitly.

\begin{clemma}{\ref{lem:T-fp-no-mu}}
  Let $(\G,\sigma)$ be a BMG, $xy$ be an edge in $\G$ and $(T,\sigma)$ be a
  tree that explains $(\G,\sigma)$. Then, the following statements are
  equivalent:
  \begin{enumerate}[itemsep=0.2ex, topsep=0.2ex, parsep=0cm]
    \item The edge $xy$ is $(T,\sigma)$-\fp.
    \item There are two children $v_1$ and $v_2$ of $\lca_T(x,y)$ such that
    $\sigma(L(T(v_1)))\cap \sigma(L(T(v_2)))\neq\emptyset$.
    \item For the extremal labeling $\tT$ of $(T,\sigma)$ it holds that
    $\tT(\lca_T(x,y)) = \DUPL$.
  \end{enumerate}
\end{clemma}

Lemma~\ref{lem:T-fp-no-mu} implies that $(T,\sigma)$-\fp can be verified in
polynomial time for any given gene tree $(T,\sigma)$.  By contraposition of
Lemma~\ref{lem:duplication_witness}, inner vertices with two distinct
children $v_1$ and $v_2$ satisfying
$\sigma(L(T(v_1)))\cap \sigma(L(T(v_2)))\neq\emptyset$ are duplication
vertices for every possible reconciliation map to every possible species
tree.  Therefore, the property of being an AD vertex only depends on
$(T,\sigma)$. In particular, $(T,\sigma)$-\fp edges coincide with the edges
$xy$ in $(\G,\sigma)$ for which $\lca_{T}(x,y)$ is an AD vertex.

\begin{figure}[t]
  \begin{center}
    \includegraphics[width=0.75\textwidth]{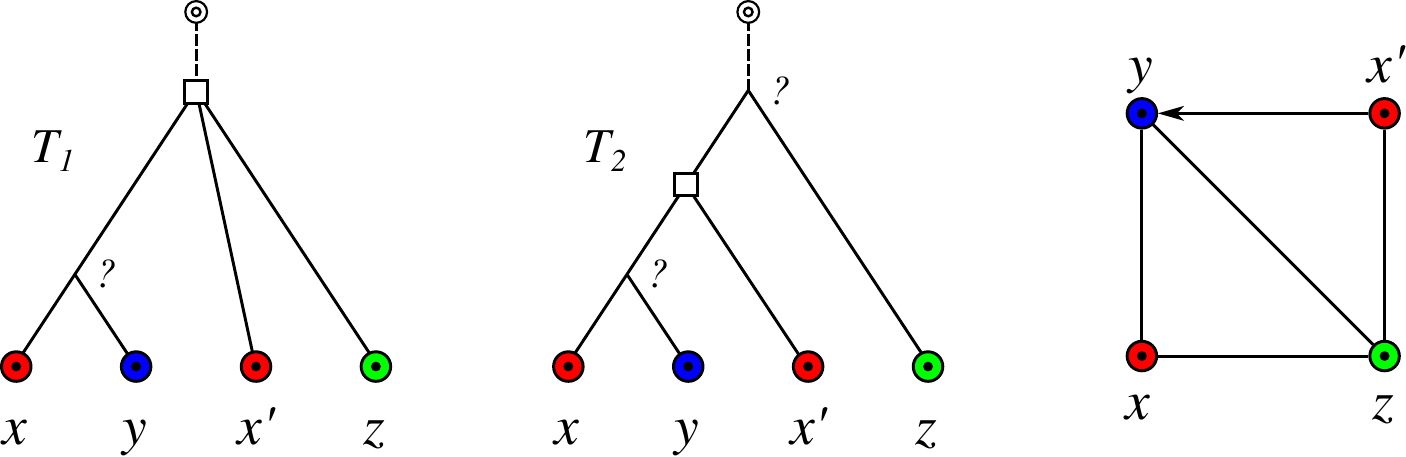}
  \end{center}
  \caption{The BMG $(\G,\sigma)$ shown on the right is explained by both
    $(T_1,\sigma)$, which is the unique least resolved tree for
    $(\G,\sigma)$, and $(T_2,\sigma)$. The vertices labeled $\DUPL$ must be
    duplications due to Lemma~\ref{lem:duplication_witness}, whereas the
    vertices labeled ``?'' could be both duplications or speciations. The
    edges $xz$, $x'z$ and $yz$ are $(T_1,\sigma)$-\fp but not
    $(T_2,\sigma)$-\fp (cf.\ Lemma~\ref{lem:T-fp-no-mu}).  Thus, neither of
    the edges $xz$, $x'z$ and $yz$ is \ufp.}
  \label{fig:non_binary_tree_example}
\end{figure}

As shown in Fig.~\ref{fig:non_binary_tree_example}, there are trees
$(T_1,\sigma)$ and $(T_2,\sigma)$ that explain the same BMG for which,
however, the edges $xz$, $x'z$, and $yz$ are $(T_1,\sigma)$-\fp but not
$(T_2,\sigma)$-\fp. Since we assume that no information on $(T,\sigma)$ is
available \emph{a priori}, it is natural to consider the set of edges that
are false positives for all trees explaining a given BMG.

\begin{cdefinition}{\ref{def:ufp}}[Unambiguous false-positive]
  Let $(\G,\sigma)$ be a BMG. An edge $xy$ in $\G$ is called
  \emph{unambiguous false-positive (\ufp)} if for all trees $(T,\sigma)$
  that explain $(\G,\sigma)$ the edge $xy$ is $(T,\sigma)$-\fp.
\end{cdefinition}

Hence, if an edge $xy$ in $\G$ is \ufp, then it is in particular
$(T,\sigma)$-\fp in the true history that explains $(\G,\sigma)$.  Thus,
\ufp edges are always correctly identified as false positives. Not all
``correct'' false-positive edges are \ufp, however. It is possible that,
for an edge $xy$ in $\G$, we have $t_\mu(\lca_T(x,y))=\DUPL$ for the true
gene tree and the true species tree, but $xy$ is not $(T',\sigma)$-\fp for
some gene tree $(T',\sigma)$ possibly different from $(T,\sigma)$. One of
the simplest examples is shown in Fig.~\ref{fig:compl_loss}, assuming that
$(\G,\sigma)$ is the ``true'' BMG. Since
$t_\mu(\lca_{\widetilde{T}}(x,y))=\SPEC$ may be possible
(Fig.~\ref{fig:compl_loss}, leftmost scenario, the edge $xy$ is not
$(\widetilde{T},\sigma)$-\fp and therefore not \ufp.

\section{Main results}
\label{sec:main-results}

\subsection{Characterization of \ufp edges}
\label{ssec:charac-ufp}

In order to adapt the concept of AD vertices for our purposes, we introduce
the color-intersection $\Scap$ associated with a gene tree
$(T,\sigma)$. For a pair of distinct leaves $x,y\in L(T)$ we denote by
$v_x, v_y \in \child_T(\lca_T(x,y))$ the unique children of the last common
ancestor of $x$ and $y$ for which $x\preceq_T v_x$ and $y\preceq_T
v_y$. That is, $T(v_x)$ and $T(v_y)$ are the subtrees of $T$ rooted in the
children of $\lca_T(x,y)$ with $x\in L(T(v_x))$ and $y\in L(T(v_y))$. The
set
\begin{equation*}
\mathcal{S}_T^{\cap}(x,y)\coloneqq\sigma(L(T(v_x)))\cap\sigma(L(T(v_y)))
\end{equation*} 
contains the colors, i.e.\ species, that are common to both subtrees.  The
existence of common colors, $\mathcal{S}_T^{\cap}(x,y)\ne\emptyset$,
determines whether or not the inner vertex $\lca_T(x,y)$ is
AD. Lemma~\ref{lem:Scap} (Sec.~\ref{APP:ssec:CI}) shows that the
color-intersection $\mathcal{S}_T^{\cap}(x,y)$ of an edge in a BMG
$(\G,\sigma)$ is independent of the corresponding tree. Hence, it suffices
to consider the color-intersection for the unique least resolved tree
$(T^*,\sigma)$ explaining $(\G,\sigma)$. From here on, we drop the explicit
reference to the tree and simply write $\Scap(x,y)$; see also
Remark~\ref{rem:Scap} in Sec.~\ref{APP:ssec:CI}. The color-intersection
provides a sufficient condition for \ufp edges in a BMG.
\begin{emptyTHM}{Prop.~\ref{prop:color_intersection_dupl} and
    Cor.~\ref{cor:Tfp-Scap-noneqi}}
  Every edge $xy$ in a BMG $(\G,\sigma)$ with $\Scap(x,y)\ne\emptyset$ is
  $(T,\sigma)$-\fp for every tree $(T,\sigma)$ that explains $(\G,\sigma)$,
  and thus \ufp.
\end{emptyTHM}

As we shall see below, the converse of
Prop.~\ref{prop:color_intersection_dupl} and Cor.~\ref{cor:Tfp-Scap-noneqi}
is not true in general. It does hold for the special case of binary trees,
however:
\begin{ctheorem}{\ref{thm:ufp-binary}}
  Let $(\G,\sigma)$ be a BMG that is explained by a binary tree
  $(T,\sigma)$. Then, for every edge $xy$ in $(\G,\sigma)$, the following
  three statements are equivalent:
  \begin{enumerate}[itemsep=0.2ex, topsep=0.2ex, parsep=0cm]
    \item The edge $xy$ is $(T,\sigma)$-\fp.
    \item $\Scap(x,y)\ne\emptyset$.
    \item The edge $xy$ is \ufp.
  \end{enumerate}
\end{ctheorem}
Prop.~\ref{prop:binary-iff-hourglass-free} in Sec.~\ref{ssect:quart}
provides a characterization of BMGs that can be explained by binary trees;
a property that can be tested in polynomial time (cf.\
Cor.~\ref{cor:binary-polytime}). However, not every BMG can be explained by
a binary tree as shown by the simple example in
Fig.~\ref{fig:hourglasses}(A).  This BMG can only be explained by the
unique non-binary tree as shown in Fig.~\ref{fig:hourglasses}(B).

Since every orthology graph is a cograph (Thm.~\ref{thm:ortho-cograph}) and
thus free of induced $P_4$s, every induced $P_4$ in the RBMG necessarily
contains a false-positive orthology assignments. The subgraphs of the BMG
spanned by a $P_4$ in its symmetric part (i.e., the RBMG) are known as
quartets. The quartets on three colors of a BMG $(\G,\sigma)$ fall into
three distinct classes depending on the coloring and the additional,
non-symmetric edges (cf.\ \cite[Lemma~32]{Geiss:19b}). We write
$\langle abcd \rangle$ or, equivalently, $\langle dcba \rangle$ for an
induced $P_4$ with edges $ab$, $bc$, and $cd$.

\begin{cdefinition}{\ref{def:GoodBadUgly}}[Good, bad, and ugly quartets]
  Let $(\G,\sigma)$ be a BMG with symmetric part $(G,\sigma)$ and vertex
  set $L$, and let $Q\coloneqq \{x,y,z,z'\} \subseteq L$ with $x\in L[r]$,
  $y\in L[s]$, and $z,z'\in L[t]$. The set $Q$, resp., the induced subgraph
  $(\G[Q],\sigma_{|Q})$ is
  \begin{itemize}[itemsep=1.2ex, topsep=0.2ex, parsep=0cm]
    \item[] a \emph{good quartet} if (i) $\langle zxyz'\rangle$ is an induced
    $P_4$ in $(G,\sigma)$ and (ii) $(z,y),(z',x)\in E(\G)$ and
    $(y,z),(x,z')\notin E(\G)$,
    \item[] a \emph{bad quartet} if (i) $\langle zxyz'\rangle$ is an induced
    $P_4$ in $(G,\sigma)$ and (ii) $(y,z),(x,z')\in E(\G)$ and
    $(z,y),(z',x)\notin E(\G)$,
    \item[] an \emph{ugly quartet} if $\langle zxz'y\rangle$ is an induced
    $P_4$ in $(G,\sigma)$.
  \end{itemize}
  \noindent The edge $xy$ in a good quartet $\langle zxyz'\rangle$ is its
  \emph{middle} edge. The edge $zx$ of an ugly quartet
  $\langle zxz'y\rangle$ or a bad quartet $\langle zxyz'\rangle$ is called
  its \emph{first} edge.  First edges in ugly quartets are uniquely
  determined due to the colors.  In bad quartets, this is not the case and
  therefore, the edge $yz'$ in $\langle zxyz'\rangle$ is a first edge as
  well.
\end{cdefinition}

\begin{figure}[t]
  \begin{center}
    \includegraphics[width=0.7\textwidth]{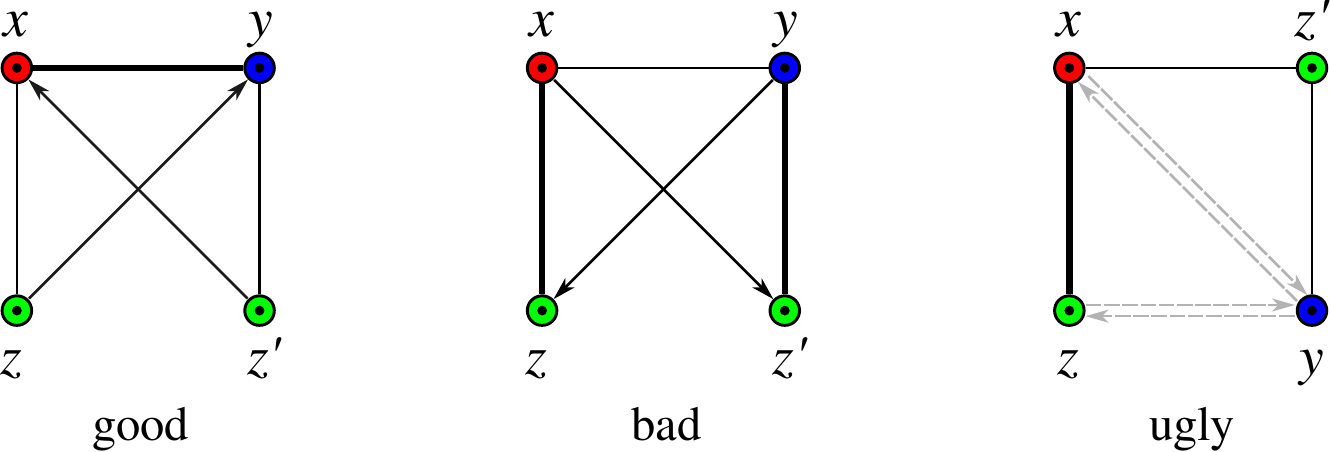}
  \end{center}
  \caption[]{The three types of quartets in BMGs. Ugly quartets may or may
    not contain either of the two (dashed) arcs between $x$ and $y$, and
    $y$ and $z$, respectively. Bold edges highlight the middle and first
    edges of the respective quartets as specified in
    Def.~\ref{def:GoodBadUgly}.}
  \label{fig:good_bad_ugly}
\end{figure}

The three different types of quartets are shown in
Fig.~\ref{fig:good_bad_ugly}. RBMGs never contain induced $P_4$s on two
colors \cite[Obs.~5]{Geiss:19b}.  This, in particular, implies that for the
induced $P_4$s in Def.~\ref{def:GoodBadUgly} the colors $r$, $s$, and $t$
must be pairwise distinct. Note that (R)BMGs may also contain induced
$P_4$s on four colors. These are investigated in some more detail in
Secs.~\ref{ssect:quart} and~\ref{APP:ssec-4colP4}.

Good quartets are characteristic of a complementary gene loss (as shown in
Fig.~\ref{fig:compl_loss}) that is ``witnessed'' by a third species in
which both child branches of the problematic duplication event
survive. That is, good quartets appear if there is a pair of genes $z$ and
$z'$ with $\sigma(z)=\sigma(z')$ and $\lca(z,z')=\lca(x,y)$ in the true
gene tree. We remark that previous work also noted that complementary gene
loss can be resolved successfully under certain circumstances
\cite{Dessimoz:2006} such as this one. An in-depth analysis of quartets
shows that they can be used to identify many of the \ufp edges. We collect
here the main results of Sec.~\ref{APP:ssec:quartets}:
\begin{emptyTHM}{Prop.~\ref{prop:good_quartet_middle_edge}, \ref{prop:bad} and 
    \ref{prop:ugly_quartet}}
  Let $\mathcal{Q} = \langle xyzw \rangle$ be a quartet in a BMG
  $(\G,\sigma)$.
  \begin{description}[nolistsep]
    \item[(i)] If $\mathcal{Q}$ is good, then its middle edge $yz$ is \ufp.
    \item[(ii)] If $\mathcal{Q}$ is ugly, then its first edge $xy$ and its
    middle edge $yz$ are \ufp.
    \item[(iii)] If $\mathcal{Q}$ is bad, then its first edges $xy$ and $zw$
    are \ufp.
  \end{description}
\end{emptyTHM}
Not surprisingly, quartets are intimately linked to color-intersections:
\begin{ccorollary}{\ref{cor:ufp-quartets}}
  Let $(\G,\sigma)$ be a BMG that contains the edge $xy$. Then,
  $\Scap(x,y)\ne\emptyset$ implies that $xy$ is either the middle edge of
  some good quartet or the first edge of some ugly quartet, which in turn
  implies that $xy$ is \ufp.
\end{ccorollary}
All \ufp edges $xy$ with $\Scap(x,y)\ne\emptyset$ in~$(\G,\sigma)$ are
therefore completely determined by the middle edges of good quartets and
the first edges of ugly quartets.  In particular, not all such edges are
the middle edge of a good quartet as the example in
Fig.~\ref{fig:ugly_quartet} shows.  Therein, the edge $xy$ must be \ufp
since $\Scap(x,y)=\{\sigma(z)\}\ne\emptyset$ (cf.\
Prop.~\ref{prop:color_intersection_dupl}). The only good quartet is
$\langle zx'yz'\rangle$ identifying $x'y$ as \ufp.  Moreover, $(\G,\sigma)$
does not contain any bad quartet.  The edge $xy$, on the other hand, is the
first edge of the ugly quartet $\langle xyx'z\rangle$.

\begin{figure}[t]
  \begin{center}
    \includegraphics[width=0.85\textwidth]{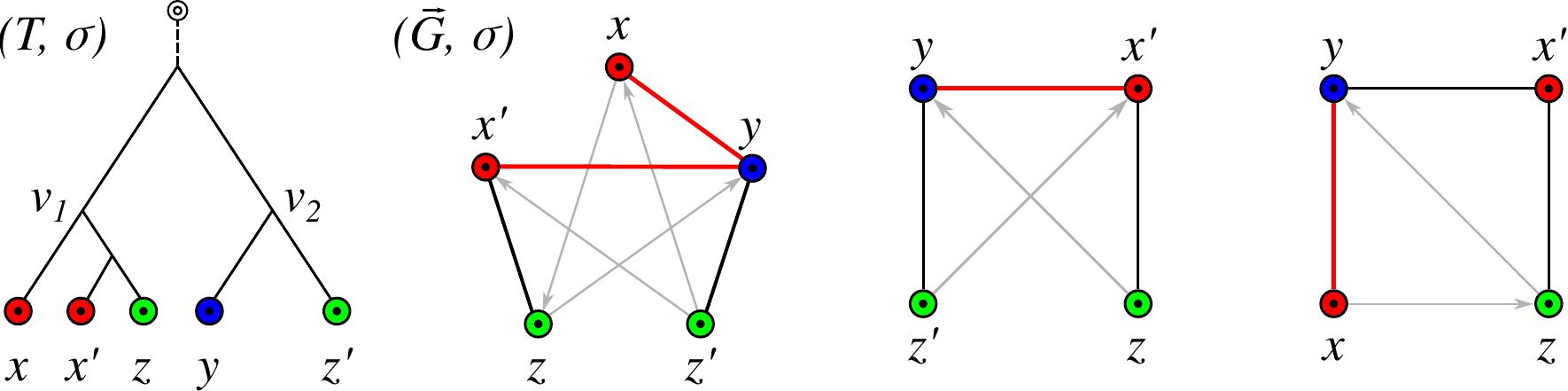}
  \end{center}
  \caption{Example for a $(T,\sigma)$-\fp edge $xy$ in $(\G,\sigma)$ which
    is not the middle edge of a good quartet, but the first edge in an ugly
    quartet (right). Note, $(\G,\sigma)$ does not contain bad quartets.}
  \label{fig:ugly_quartet}
\end{figure}

Furthermore, if an edge $xy$ is the middle edge of a good quartet, then
$\Scap(x,y)\ne\emptyset$. Therefore, only ugly quartets may provide
additional information about \ufp edges that are not identified with the
help of the color-intersection $\Scap$ (see
Fig.~\ref{fig:Scap_empty_ugly_quartet} in Sec.~\ref{APP:ssec:quartets} for
an example). Ugly quartets, however, do not convey all the missing
information on \ufp edges. The edge $xy$ in the BMG shown in
Fig.~\ref{fig:hourglasses}(A) is \ufp, but it is not contained in a good,
bad, or ugly quartet.

In order to characterize the \ufp edges that are not identified by
quartets, we first introduce an additional motif that may occur in
vertex-colored graphs.
\begin{cdefinition}{\ref{def:hourglass}}[Hourglass]
  An \emph{hourglass} in a proper vertex-colored graph $(\G,\sigma)$,
  denoted by $[xy \hourglass x'y']$, is a subgraph $(\G[Q],\sigma_{|Q})$
  induced by a set of four pairwise distinct vertices
  $Q=\{x, x', y, y'\}\subseteq V(\G)$ such that (i)
  $\sigma(x)=\sigma(x')\ne\sigma(y)=\sigma(y')$, (ii) $xy$ and $x'y'$ are
  edges in $\G$, (iii) $(x,y'),(y,x')\in E(\G)$, and (iv)
  $(y',x),(x',y)\notin E(\G)$.
\end{cdefinition}
Note that Condition (i) rules out arcs between $x,x'$ and $y,y'$,
respectively, i.e., the only arcs in an hourglass are the ones specified by
Conditions (ii) and (iii).  An example is shown in
Fig.~\ref{fig:hourglasses}(A).
\begin{cfact}{\ref{obs:hourbmg}}
  Every hourglass is a BMG since it can be explained by a tree as shown in
  Fig.~\ref{fig:hourglasses}(B).
\end{cfact}
Hourglasses are not necessarily part of an induced $P_4$. In particular, an
hourglass does not contain an induced $P_4$ (see
Fig.~\ref{fig:hourglasses}(A)).

\begin{figure}[t]
  \begin{center}
    \includegraphics[width=0.85\textwidth]{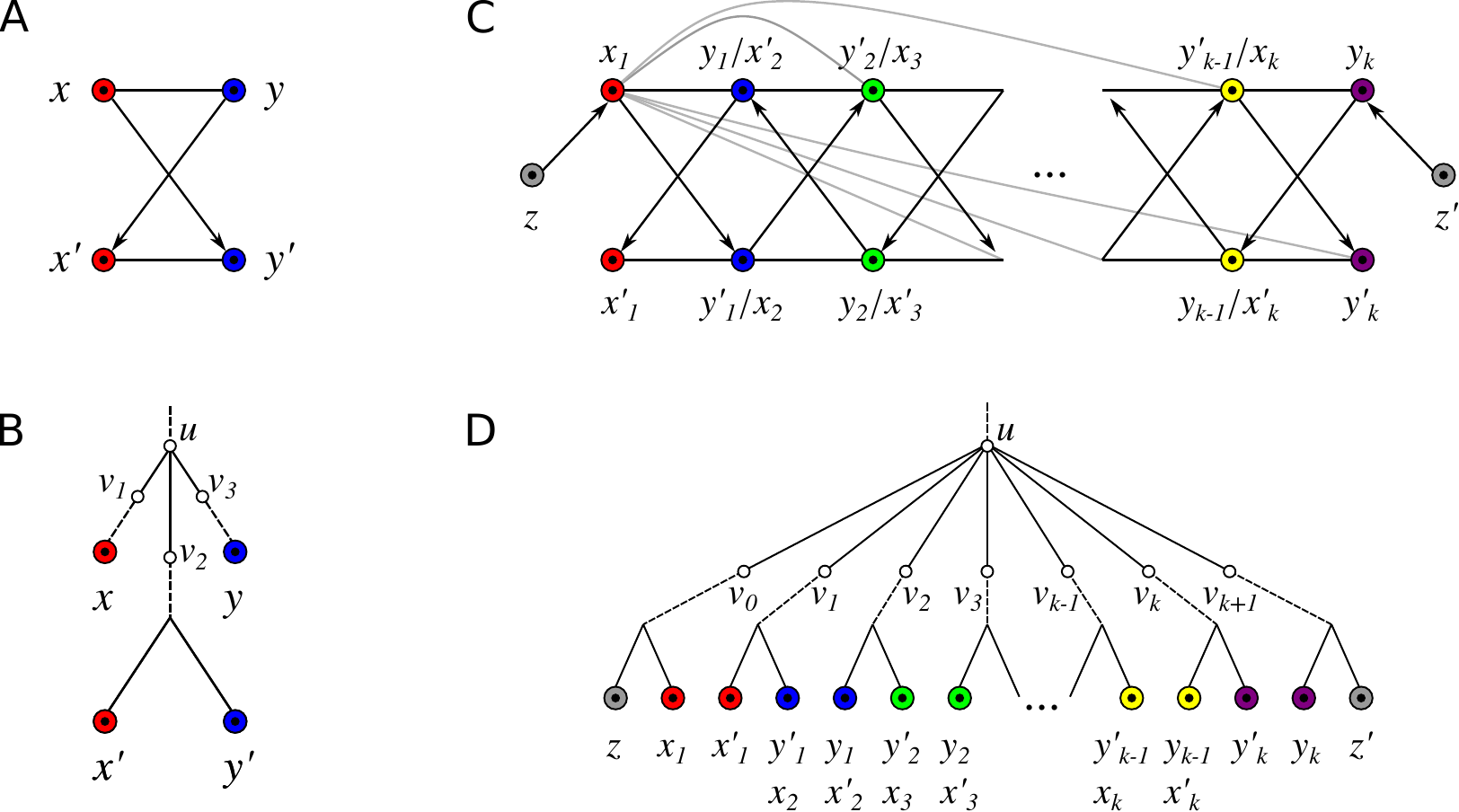}
  \end{center}
  \caption{A: Hourglass.  B: Visualization of Lemma~\ref{lem:hourglass}. C:
    Hourglass chain with left tail $z$ and right tail $z'$ for an odd
    number of hourglasses in the chain. Edges of the form
    $x_i y'_j\in E(G)$ are only shown for $x_1$, the others are omitted. An
    hourglass chain $\mathfrak{H}$ is a subgraph but not necessarily
    induced and thus additional arcs may exist.  In particular, the
    elements $e\in\{x_1y_k, zy_k, x_1z', zz'\}$ are not necessarily edges
    in an hourglass chain. However, whenever they exist, they are \ufp
    (cf.\ Lemma~\ref{lem:hourglass_chain_dupl}).  Moreover, each single
    hourglass in $\mathfrak{H}$ is an induced subgraph of the BMG; by
    definition, therefore, there are no arcs $(z,x'_1)$ or
    $(z',y'_k)$. Note, $\sigma(z)\neq \sigma(z')$ is possible.  D:
    Visualization of Lemmas~\ref{lem:hourglass_chain}
    and~\ref{lem:hourglass_chain_tails}.}
  \label{fig:hourglasses}
\end{figure}

Hourglasses $[xy \hourglass x'y']$ can be used to identify false-positive
edges $xy$ with $\Scap(x,y)=\emptyset$.  More precisely, we have
\begin{cproposition}{\ref{prop:singeHG-ufp}}
  If a BMG $(\G,\sigma)$ contains an hourglass $[xy \hourglass x'y']$, then
  the edge $xy$ is \ufp.
\end{cproposition}

Prop.~\ref{prop:singeHG-ufp} implies that there are \ufp edges that are not
contained in a quartet, see Fig.~\ref{fig:hourglasses}(A). In this example,
we have $\Scap(x,y)=\emptyset$ and no induced $P_4$. However, as shown in
Fig.~\ref{fig:hourglasses}(B), the subtree $T(v_2)$ contains both colors
$\sigma(x)$ and $\sigma(y)$ and thus, ``bridges'' the color sets of the
subtrees $T(v_1)$ and $T(v_3)$.  Similarly, in the tree $(T,\sigma)$ in
Fig.~\ref{fig:hourglasses}(D), each subtree $T(v_i)$, $1\leq i \leq k$
``bridges'' the color sets of the subtrees $T(v_{i-1})$ and
$T(v_{i+1})$. This observation suggests the concept of hourglass chains, a
generalization of hourglasses.
\begin{cdefinition}{\ref{def:hc}}[Hourglass chain]
  An \emph{hourglass chain} $\mathfrak{H}$ in a graph $(\G,\sigma)$ is a
  sequence of $k\ge 1$ hourglasses
  $[x_1 y_1 \hourglass x'_1 y'_1],\dots,[x_k y_k \hourglass x'_k y'_k]$
  such that the following two conditions are satisfied for all
  $i\in\{1,\dots,k-1\}$:
  \begin{description}[itemsep=0.2ex, topsep=0.2ex, parsep=0cm]
    \item[\emph{(H1)}] $y_i=x'_{i+1}$ and $y'_i=x_{i+1}$, and
    \item[\emph{(H2)}] $x_i y'_j$ is an edge in $\G$ for all
    $j\in\{i+1,\dots,k\}$
  \end{description}
  A vertex $z$ is called a \emph{left} (resp., \emph{right}) \emph{tail} of
  the hourglass chain $\mathfrak{H}$ if it holds that $(z,x_1)\in E(\G)$
  and $(z,x'_1)\notin E(\G)$ (resp., $(z,y_k)\in E(\G)$ and
  $(z,y'_k)\notin E(\G)$).  We call $\mathfrak{H}$ \emph{tailed} if it has
  a left or right tail.
\end{cdefinition}
In contrast to the quartets and the hourglass, an hourglass chain in
$(\G,\sigma)$ is not necessarily an induced subgraph.  Hourglass chains are
``overlapping'' hourglasses.  The additional condition that
$x_i y'_j\in E(G)$ for all $1\le i<j\le k$ ensures that the two pairs
$x'_k,y'_k$ and $x'_l,y'_l$ with $k\ne l$ cannot lie in the same subtree
below the last common ancestor $u$ which is common to all hourglasses in
the chain (cf.\ Lemma~\ref{lem:hourglass_chain}
and~\ref{lem:hourglass_chain_tails} in Sec.~\ref{APP:ssect:hourglass}).

\begin{cdefinition}{\ref{def:hug-edge}}
  An edge $xy$ in a vertex-colored graph $(\G,\sigma)$ is a
  \emph{hug-edge} if it satisfies at least one of the following
  conditions:
  \begin{description}[itemsep=0.2ex, topsep=0.2ex, parsep=0cm]
    \item[\emph{(C1)}] $xy$ is the middle edge of a good quartet in
    $(\G,\sigma)$;
    \item[\emph{(C2)}] $xy$ is the first edge of an ugly quartet in
    $(\G,\sigma)$; or
    \item[\emph{(C3)}] there is an hourglass chain
    $\mathfrak{H}=[x_1 y_1 \hourglass x'_1 y'_1],\dots,[x_k y_k \hourglass
    x'_k y'_k]$ in $(\G,\sigma)$, and one of the following cases holds:
    \begin{enumerate}[noitemsep, topsep=0.2ex, parsep=0cm, nolistsep]
      \item $x_1=x$ and $y_k=y$;
      \item $y_k=y$ and $z\coloneqq x$ is a left tail of $\mathfrak{H}$;
      \item $x_1=x$ and $z'\coloneqq y$ is a right tail of $\mathfrak{H}$; or
      \item $z\coloneqq x$ is a left tail and $z'\coloneqq y$ is a right tail
      of $\mathfrak{H}$.
    \end{enumerate}
  \end{description}
\end{cdefinition}  
The term \textbf{hug}-edge refers to the fact that $xy$ is a particular
edge of an \textbf{h}ourglass-chain, an \textbf{u}gly quartet, or a
\textbf{g}ood quartet. In
Sec.~\ref{APP::ssec-sec:augmented-extremal-labeling}, we show that
hug-edges coincide with the \ufp edges.

\begin{ctheorem}{\ref{thm:ufp-iff-hug}}
  An edge $xy$ in a BMG $(\G,\sigma)$ is \ufp if and only if $xy$ is a
  hug-edge of $(\G,\sigma)$.
\end{ctheorem}
Interestingly, bad quartets turn out to be redundant for the identification
of \ufp edges in the sense that every \ufp edge in a bad quartet appears as
a \ufp edge in a good quartet, an ugly quartet, or an hourglass chain.  At
present, we do not know whether hourglass chains in a colored graph
$(\G,\sigma)$ can be found efficiently. We shall see in the following
section, however, that the identification of \ufp edges does not require
the explicit enumeration of hourglass chains.

The fact that all hug-edges are \ufp by Thm.~\ref{thm:ufp-iff-hug} suggests
to consider the subgraph of a BMG that is left after removing all these
unambiguously recognizable false-positive orthology assignments.
\begin{cdefinition}{\ref{def:non-hug-graph}}
  Let $(\G,\sigma)$ be a BMG with symmetric part $G$ and let $F$ be the set
  of its hug-edges. The \emph{no-hug\footnote{a good advice in the time of
      COVID-19}} graph $\NH(\G,\sigma)$ is the subgraph of $G$ with
  vertex set $V(\G)$, coloring $\sigma$ and edge set $E(G)\setminus F$.
\end{cdefinition}
By Thm.~\ref{thm:ufp-iff-hug}, $\NH(\G,\sigma)$ is therefore the subgraph
of the underlying RBMG of $(\G,\sigma)$ that does not contain any \ufp
edge. Importantly, it contains the orthology graph for every reconciliation
map $\mu$ as well as the orthology graph induced by the extremal event
labeling as subgraphs:
\begin{ccorollary}{\ref{cor:NH}}
  Let $(T,\sigma)$ be a leaf-colored tree and $\mu$ a reconciliation map
  from $(T,\sigma)$ to some species tree $S$. Then,
  \begin{equation*}
  \Theta(T,t_{\mu}) \subseteq \Theta(T, \tT ) \subseteq
  \NH(\G(T,\sigma)) \subseteq \G(T,\sigma).
  \end{equation*}
\end{ccorollary}

The no-hug graph still may contain false-positive orthology assignments,
i.e., $\NH(\G(T,\sigma))=\Theta(T,\tT)$ does not hold in general.  As an
example, consider the BMG $\G(T_1,\sigma)$ in
Fig.~\ref{fig:non_binary_tree_example}. Here, none of the edges $xz$, $x'z$
and $yz$ are \ufp and thus, by Thm.~\ref{thm:ufp-iff-hug} also not
hug-edges.  Hence, they still remain in $\NH(\G(T_1,\sigma))$. However,
these edges are not contained in $\Theta(T_1,\tT)$, since
$\tT(\lca_{T_1}(x,x',y,z)) = \DUPL$ and thus,
$\Theta(T_1,\tT) \subsetneq \NH(\G(T_1,\sigma))$.

\subsection{Algorithms}
\label{ssect:algorithms}

In this section, we provide a polynomial-time algorithm to identify all
\ufp edges in a given BMG. To this end, we take a closer look at hourglass
chains and the trees that explain them. In Fig.~\ref{fig:hourglasses}(D),
each subtree $T(v_i)$, $1\leq i \leq k$, ``bridges'' the color sets of the
subtrees $T(v_{i-1})$ and $T(v_{i+1})$.  That is,
$\sigma(L(T(v_{i-1})))\cap\sigma(L(T(v_i)))$ and
$\sigma(L(T(v_i)))\cap\sigma(L(T(v_{i+1})))$ are non-empty.  This suggests
to consider the children of a vertex $u$ as the vertices of a ``color-set
intersection graph'' with edges connecting children with non-empty
color-set intersection:
\begin{definition}
  The \emph{color-set intersection graph} $\CIG_T(u)$ of an inner vertex
  $u$ of a leaf-colored gene tree $(T,\sigma)$ is the undirected graph with
  vertex set $V\coloneqq\child_T(u)$ and edge set
  \begin{equation*}
  E\coloneqq\{ v_1v_2 \mid v_1,v_2\in V \textrm{, }v_1\ne v_2
  \textrm{ and } \sigma(L(T(v_1)))\cap\sigma(L(T(v_2)))\ne\emptyset \}.
  \end{equation*}
\end{definition}
This construction is similar to the definition of intersection graphs e.g.\
used in \cite{McKee:1999}. $\CIG_T(u)$ can be viewed as a natural
generalization of $\Scap(x,y)$ in the following sense: if $u=\lca_T(x,y)$
is a binary vertex, then $\CIG_T(u)=K_2$ \emph{iff}
$\Scap(x,y)\ne\emptyset$ and therefore, $\CIG_T(u)=K_1\cup K_1$ \emph{iff}
$\Scap(x,y)=\emptyset$. In the non-binary case, there is an edge $v_1v_2$
\emph{iff} $\Scap(x,y)\ne\emptyset$ for some $x\in L(T(v_1))$ and
$y\in L(T(v_2))$. 

Every BMG $(\G,\sigma)$ contains all information necessary to determine the
trees $(T,\sigma)$ by which it is explained. Since \ufp edges are defined
in terms of the explaining trees, every BMG $(\G,\sigma)$ also contains --
at least implicitly -- all information needed to identify its \ufp edges.
Since $(\G,\sigma)$ is determined by its unique least resolved tree
$(T^*,\sigma)$, the \ufp edges must also be determined by $(T^*,\sigma)$.
It is not sufficient for this purpose, however, to find an event labeling
$t$ of the vertices of $T^*$.

To see this, consider for example the ``true'' history
$(\widetilde{T},\widetilde{t},\sigma)$ of the BMG
$\G(\widetilde{T},\sigma)$ as shown in Fig.~\ref{fig:messy_vertices-1}.
The unique least resolved tree $(T^*,\sigma)$ for
$\G(\widetilde{T},\sigma)$ is obtained by merging the two vertices $v_1$
and $v_2$ of $\widetilde{T}$ resulting in the vertex $v$ of $T^*$. We have
$\widetilde{t}(v_1)= \SPEC\neq\DUPL=\widetilde{t}(v_2)$. For vertex $v$ and
every reconciliation map $\mu$ from $(T^*,\sigma)$ to any species tree $S$,
it must hold that $\mu(v)\in E(S)$ and thus $t^*_{\mu}(v)=\DUPL$, since $v$
has two children with overlapping color sets and by
Lemma~\ref{lem:duplication_witness}. Thus, the edges $cx$ with
$x\in \{a_1,a_2,b_1,b_2\}$ are $(T^*,\sigma)$-\fp although they are not
false positives at all. Since speciation and duplication vertices may be
merged into the same vertex $v$ of $T^*$, the least resolved tree $T^*$ in
general cannot simply inherit the event labeling from the true gene
history, and thus there may not be a ``correct'' labeling $t^*$ of $T^*$
that provides evidence for all \ufp edges.

\begin{figure}[t]
  \begin{center}
    \includegraphics[width=0.85\linewidth]{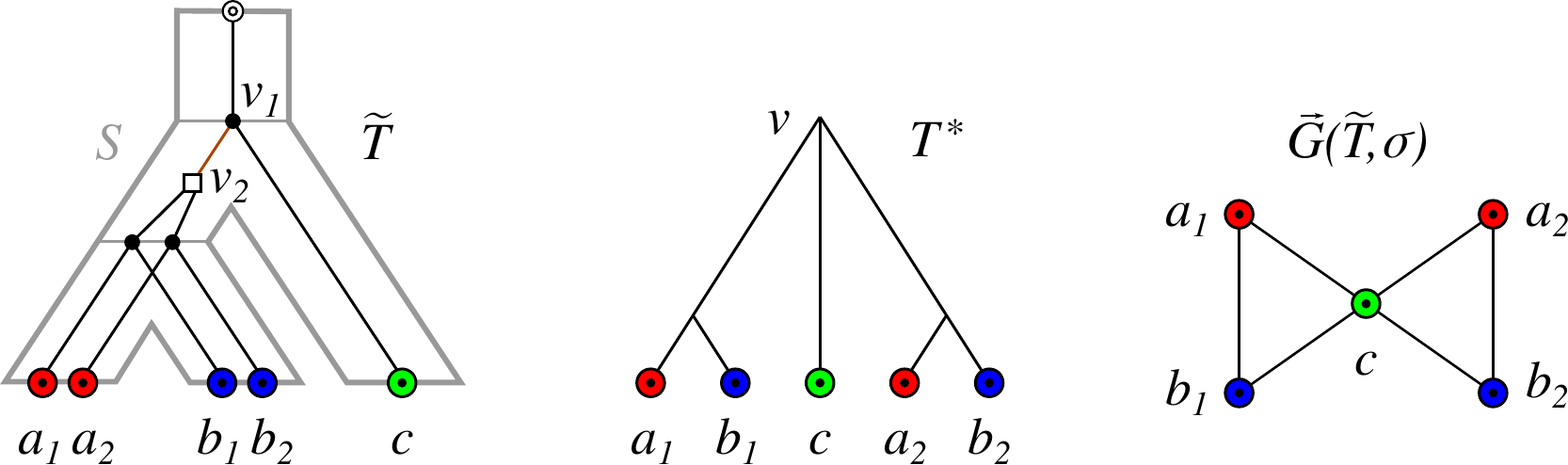}
  \end{center}
  \caption{The evolutionary scenario (left) shows the event-labeled gene
    tree $(\widetilde{T},\widetilde{t},\sigma)$ embedded into a species
    tree $S$. In the least resolved tree $(T^*,\sigma)$ of
    $\G(\widetilde{T},\sigma)$, the edge $v_1v_2$ of $\widetilde{T}$ has
    been contracted into vertex $v$. The BMG $\G(\widetilde{T},\sigma)$
    does not contain any \ufp edge.  \emph{See text for further
      explanations.}}
  \label{fig:messy_vertices-1}
\end{figure}

The example in Fig.~\ref{fig:messy_vertices-1} shows that the least
resolved tree $T^*$ simply may not be ``resolved enough''. In the
following, we therefore describe how the unique least resolved tree can be
resolved further to provide more evidence about \ufp edges. Eventually,
this will lead us to a characterization of the \ufp edges.  To this end, we
need to gain more insights into the structure of redundant edges, i.e.,
those edges $e$ in $T$ for which $(T_e,\sigma)$ still explains
$\G(T,\sigma)$.

Since the color sets of distinct subtrees below a speciation vertex cannot
overlap by Lemma~\ref{lem:duplication_witness},
Cor.~\ref{cor:edge_redundant} (Sec.~\ref{APP:subsect:rbmg}) implies that all 
edges below a 
speciation
vertex are redundant and thus can be contracted. More precisely, we have
\begin{cfact}{\ref{obs:speciations_merged}}
  Let $\mu$ be a reconciliation map from $(T,\sigma)$ to $S$ and assume
  that there is a vertex $u\in V^0(T)$ such that $\mu(u)\in V^0(S)$ and
  thus, $t_{\mu}(u)=\SPEC$.  Then every inner edge $uv$ of $T$ with
  $v\in\child_{T}(u)$ is redundant w.r.t.\ $\G(T,\sigma)$.  Moreover, if an
  inner edge $uv$ with $v\in\child_{T}(u)$ is non-redundant, then $u$ must
  have two children with overlapping color sets, and hence,
  $t_{\mu}(u)=\DUPL$.
\end{cfact}

Our goal is to identify those vertices in $(T^*,\sigma)$ that can be
expanded to yield a tree that still explains $\G(T^*,\sigma)$. To this end,
we need to introduce a particular way of ``augmenting'' a leaf-colored
tree.
\begin{cdefinition}{\ref{def:augmenting}}
  Let $(T,\sigma)$ be a leaf-colored tree, $u$ be an inner vertex of $T$,
  $\CIG_T(u)$ the corresponding color-set intersection graph, and
  $\mathcal{C}$ the set of connected components of $\CIG_T(u)$. Then
  the tree $T_u$ \emph{augmented at vertex $u$} is obtained by
  applying the following editing steps to $T$:
  \begin{itemize}[itemsep=0.2ex, topsep=0.2ex, parsep=0cm]
    \item If $\CIG_T(u)$ is connected, do nothing.
    \item Otherwise, for each $C\in\mathcal{C}$ with $|C|>1$
    \begin{itemize}[noitemsep,nolistsep]
      \item introduce a vertex $w$ and attach it as a child of $u$, i.e., add
      the edge $uw$,
      \item for every element $v_i\in C$, substitute the edge $uv_i$ by the
      edge $wv_i$.
    \end{itemize}
  \end{itemize}
  The augmentation step is \emph{trivial} if $T_u=T$, in which case
  we say that \emph{no edit step was performed}.
\end{cdefinition}
An example of an augmentation is shown in
Fig.~\ref{fig:augmenting_labeling_algo}.  The tree $T_u$ obtained by an
augmentation of a phylogenetic tree $T$ is again a phylogenetic tree.

\begin{figure}[t]
  \begin{center}
    \includegraphics[width=0.85\textwidth]{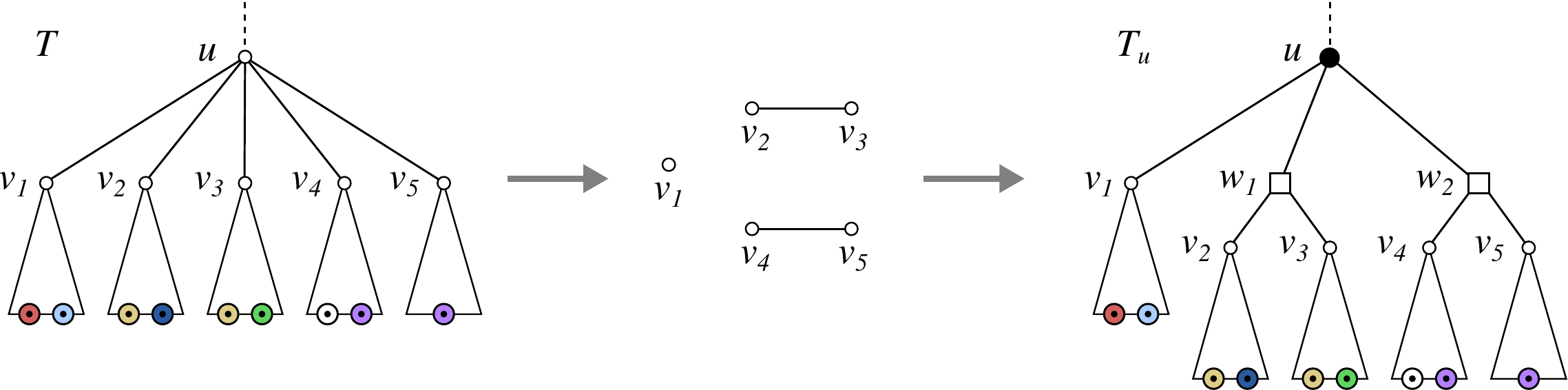}
  \end{center}
  \caption{Left, a (part of a) leaf-colored tree $(T,\sigma)$.  The tree
    $(T_u,\sigma)$ on the right is obtained from $(T,\sigma)$ by augmenting
    $T$ at vertex $u$. The color-set intersection graph $\CIG_T(u)$ (shown
    in the middle) has more than one connected component and there are
    connected components consisting of more than two vertices
    $v_i\in\child_T(u)$. According to Lemma~\ref{lem:augmenting_color_disjoint},
    $\sigma(L(T_u(v)))\cap\sigma(L(T_u(v')))=\emptyset$ for any two
    distinct vertices $v,v'\in\child_{T_u}(u) = \{v_1,w_1,w_2\}$.
    By Cor.~\ref{cor:edge_redundant} (Sec.~\ref{APP:subsect:rbmg}), the edges 
    $uw_1$ and $uw_2$ 
    are redundant w.r.t.\ $\G(T_u,\sigma)$ and thus, both trees explain the
    same BMG.}
  \label{fig:augmenting_labeling_algo}
\end{figure}

A key property of the procedure in Def.~\ref{def:augmenting} is that
repeated augmentation of the same inner vertex leads to at most one
expansion and that the order of augmenting multiple vertices does not
matter.  More precisely, Lemma~\ref{lem:augment_extremal} in
Sec.~\ref{APP:ssect:augtree} ensures the existence of a unique augmented
tree:
\begin{cdefinition}{\ref{def:aug-tree}}[Augmented tree]
  Let $(T,\sigma)$ be a leaf-colored tree. The \emph{augmented tree of
    $(T,\sigma)$}, denoted by $(\aug(T),\sigma)$, is obtained by augmenting
  all inner vertices of $(T,\sigma)$ (in an arbitrary order).
\end{cdefinition}
In particular, the augmented tree preserves the best match relation:
\begin{cproposition}{\ref{prop:aug-bmg}}
  For every leaf-colored tree $(T,\sigma)$, it holds
  $\G(T,\sigma)=\G(\aug(T),\sigma)$.
\end{cproposition}

We now have everything in place to present the main results of this
section.
\begin{ctheorem}{\ref{thm:MAIN}}
  Let $(\G,\sigma)$ be a BMG, $(T^*,\sigma)$ its unique least resolved
  tree, and $\wt\coloneqq \tTps$ the extremal event labeling of the
  augmented tree $(\aug(T^*),\sigma)$. Then
  $(\Theta(\aug(T^*),\wt),\sigma) = \NH(\G,\sigma)$. 
\end{ctheorem}

Since $(\Theta(\aug(T^*),\wt),\sigma) = \NH(\G,\sigma)$ is the subgraph of
the underlying RBMG of $(\G,\sigma)$ that does not contain any \ufp edges
(cf.\ Def.~\ref{def:non-hug-graph} and Thm.~\ref{thm:ufp-iff-hug}), the set
of all \ufp edges can readily be obtained by comparing the edges of
$(\G,\sigma)$ with the edges in the orthology graph obtained from
$(\aug(T^*),\wt)$.  Since only \ufp edges have been removed to obtain
$(\Theta(\aug(T^*),\wt),\sigma)$ and since $(\aug(T^*),\sigma)$ still
explains $(\G,\sigma)$, the graph $(\Theta(\aug(T^*),\wt),\sigma)$ is, in
the sense of an unambiguous editing, the best estimate of the orthology
relation that we can make by solely utilizing the structural information of
a given BMG $(\G,\sigma)$. Note, Thm.~\ref{thm:ortho-cograph} implies that
$\NH(\G,\sigma)$ must, in particular, be a cograph.

Since $(\Theta(\aug(T^*),\wt),\sigma) = \NH(\G,\sigma)$, the computation of
$\NH(\G,\sigma)$ can be achieved in polynomial time and avoids the need to
find the hourglass chains of $(\G,\sigma)$. In fact, the effort is
dominated by computing the least resolved tree $(T^*,\sigma)$ for a given
BMG.
\begin{ctheorem}{\ref{thm:time}}
  For a given BMG $(\G,\sigma)$, the set of all \ufp edges can be computed
  in $O(|L|^3 |\mathscr{S}|)$ time, where $L=V(\G)$ and
  $\mathscr{S} = \sigma(L(T))$ is the set of species under consideration.
\end{ctheorem}
As argued in \cite[Sec.~5]{Geiss:19a}, the number of genes between
different species will be comparable in practical applications, i.e.,
$O(\ell) = O(|L|/|\mathscr{S}|)$ with
$\ell = \max_{s\in \mathscr{S}} |L[s]|$.  In this case, the running time to
compute $(T^*,\sigma)$ reduces to $O(|L|^3/|\mathscr{S}|)$ and we obtain an
overall running time to compute the set of all \ufp edges of
$O(|L|^3/|\mathscr{S}| + |L|^2 |\mathscr{S}|)$.  Thms.~\ref{thm:MAIN}
and~\ref{thm:time} imply that we do not need to find induced quartets and
hourglasses explicitly, nor do we need to identify the hourglass chains.
Instead, it is more efficient to compute the least resolved tree
$(T^*,\sigma)$, its augmented tree $(\aug(T^*),\sigma)$, and the
corresponding extremal event labeling $\wt$.

Deletion of all \ufp edges is necessary to obtain an orthology relation
without false positives. It is not sufficient, however, since
$\NH(\G,\sigma)$ may contain additional false-positive orthology
assignments. In order to construct an example, we consider for a BMG
$(\G,\sigma)$ the set $\mathfrak{T}$ of all trees $(T,t,\sigma)$ for which
$\NH(\G,\sigma) = (\Theta(T,t),\sigma)$.  The example in
Fig.~\ref{fig:contradictory_triples} shows that it may be the case that
none of the trees $(T,t,\sigma)\in\mathfrak{T}$ admits a reconciliation map
$\mu$ to any species tree such that $t_{\mu} =
t$. Lemma~\ref{lem:reconc-aug-all} in Sec.~\ref{ssec:APP:further-fp} shows
that the augmented tree $(\aug(T^*),\wt,\sigma)$ is sufficient to test in
polynomial time whether or not $\mathfrak{T}$ contains a reconcilable
tree. In the negative case, we have clear evidence that $\NH(\G,\sigma)$
still contains a false-positive edge and thus must be edited further. This
type of false-positive orthology assignments is the topic of ongoing work.

\begin{figure}[t]
  \begin{center}
    \includegraphics[width=0.85\textwidth]{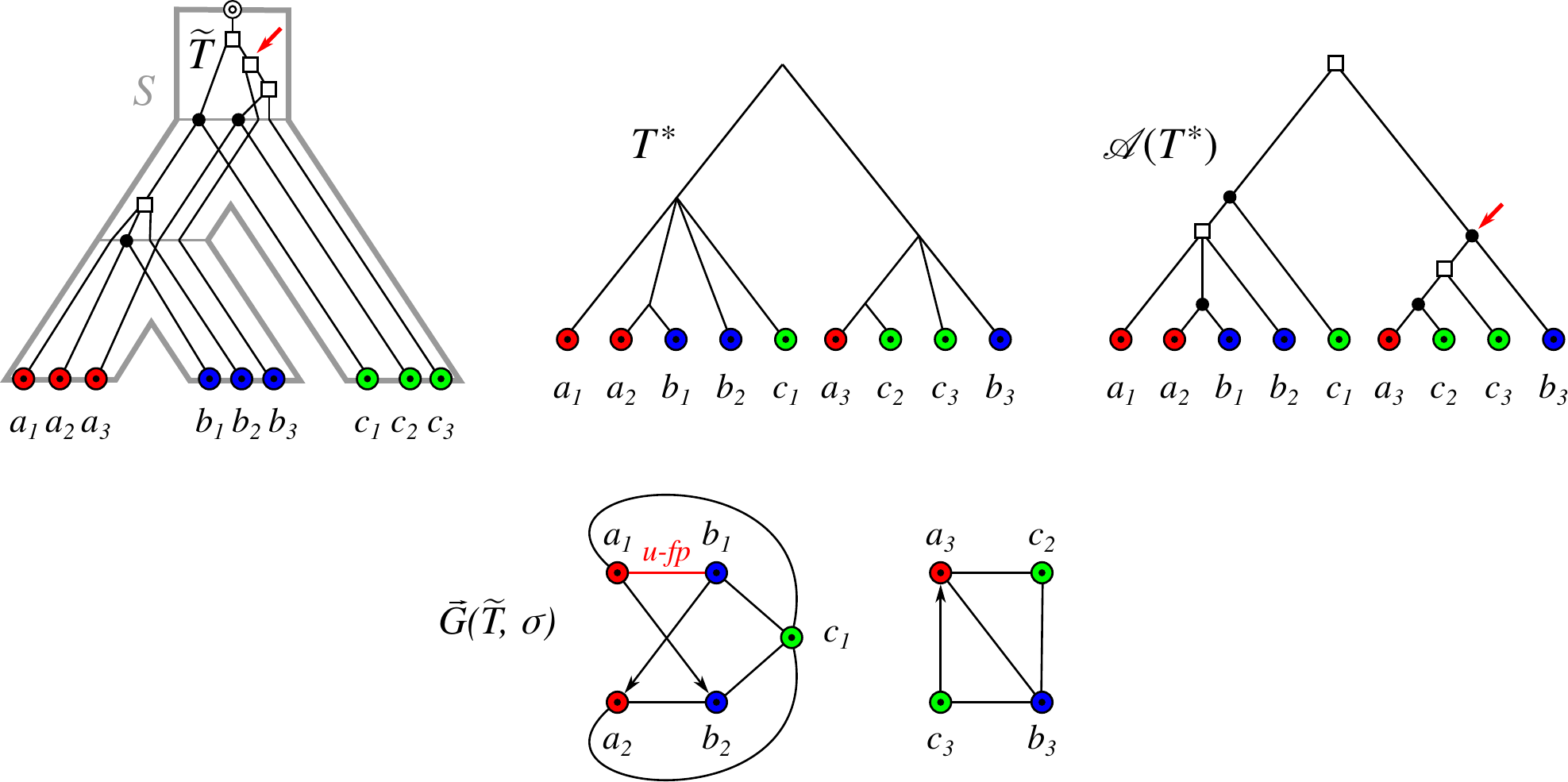}
  \end{center}
  \caption{An evolutionary scenario (left) with a no-hug graph
    $\NH(\G,\sigma)$ that still contains false-positive edges.  Deletion of
    the highlighted \ufp edge $a_1b_1$ for $\G(\widetilde{T},\sigma)$
    yields $\NH(\G,\sigma) = (\Theta(\aug(T^*),\wt),\sigma)$ and thus, an
    orthology graph. However, none of its cotrees can be reconciled with
    any species tree since each of them contains the contradictory species
    triples $\sigma(a_1)\sigma(b_1)|\sigma(c_1)$ and
    $\sigma(a_1)\sigma(c_1)|\sigma(b_1)$ (see e.g.\
    \cite{HernandezRosales:12a,Hellmuth:17}).  Note, the trees
    $(\widetilde{T},\widetilde{t})$ and $(\aug(T^*),\wt)$ differ in the
    event label marked by the arrows, resulting in the three additional \fp
    edges $a_3b_3$, $c_2b_3$ and $c_3b_3$ in $\NH(\G,\sigma)$.  }
  \label{fig:contradictory_triples}
\end{figure}

In contrast to the LRT of a BMG, its augmented tree is not necessarily
displayed by the true gene tree of the underlying evolutionary
scenario. Hence, we advocate the augmented tree endowed with the
corresponding extremal event labeling $(\aug(T^*),\wt,\sigma)$ primarily as
convenient tool to identify false-positive orthology assignments. Whether
or not $(\aug(T^*),\wt,\sigma)$ is a plausible representation of the gene
phylogeny depends on whether it admits a reconciliation of the
(phylogenetically correct) species tree. As discussed above, this is not
always the case. The following result, however, shows that
$(\aug(T^*),\wt,\sigma)$ is informative in an important special case.

\begin{clemma}{\ref{lem:T-displays-aug-tree}}
  Let $(T,t,\sigma)$ be an event-labeled tree explaining the BMG
  $(\G,\sigma)$, and let $(T^*,\sigma)$ be the least resolved tree of
  $(\G,\sigma)$.  If $(\Theta(T,t),\sigma) = \NH(\G,\sigma)$, then
  $\aug(T^*)$ is displayed by $T$.
\end{clemma}

Lemma~\ref{lem:T-displays-aug-tree} guarantees that $\aug(T^*)$ is
displayed by the true gene tree $\widetilde{T}$ whenever $\NH(\G,\sigma)$
equals the true orthology relation.  In a practical workflow, it can be
checked efficiently whether there is evidence for additional false-positive
edges because $\mathfrak{T}$ contains no reconcilable tree. If this is not
the case, then it is likely that $\NH(\G,\sigma)$ equals the true orthology
relation. In this case, $\widetilde{T}$ also displays the unique
discriminating cotree of $\NH(\G,\sigma)$.

One has to keep in mind, however, that it is not possible to find a
mathematical guarantee for $\NH(\G,\sigma)$ to be the true orthology
relation, because it cannot be ruled out that the true scenario contains
unwitnessed duplications that are compensated by additional gene losses. In
the extreme case, it is logically possible for every BMG that, in the true
scenario, all inner vertices of the gene tree predate the root of the
species tree, resulting in a true orthology graph without any edges
\cite{Guigo:96,Page:97,Geiss:20a}. Of course, this is extremely unlikely
for real data.

\subsection{Quartets, hourglasses, and the structure of reciprocal best
  match graph}
\label{ssect:quart}

The characterization of \ufp edges is in a way surprising when compared to
previous results on the structure of RBMGs \cite{Geiss:20a,Geiss:19b},
which were focused on $P_4$s and quartets. The expected connection between
good and ugly quartets and \ufp edges is captured by
Cor.~\ref{cor:ufp-quartets}.  However, Prop.~\ref{prop:singeHG-ufp} implies
that there are also \ufp edges entirely unrelated to quartets and thus
induced $P_4$s. In this section, we aim to close this gap in our
understanding.

\paragraph{Hourglass-free BMGs.}
We start with an important special case for which quartets are sufficient.
\begin{cdefinition}{\ref{def:hourglass-free}}
  A BMG $(\G,\sigma)$ is \emph{hourglass-free} if it does not contain an
  hourglass as an induced subgraph.
\end{cdefinition}
In particular, an hourglass-free BMG does not contain an hourglass chain.
It turns out that hourglasses are the forbidden induced subgraph
characterizing BMGs that can be explained by binary trees.
\begin{emptyTHM}{Prop.~\ref{prop:binary-iff-hourglass-free} and 
    Cor.~\ref{cor:binary-polytime}.}
  A BMG $(\G,\sigma)$ can be explained by a binary tree if and only if it is
  hourglass-free.  In particular, it can be decided in polynomial time
  whether $(\G,\sigma)$ can be explained by a binary tree.
\end{emptyTHM}

The RBMGs that are already cographs are called \emph{co-RBMGs}. 
As shown in Sec.~\ref{APP:ssec:hourglass-free}, we  obtain
\begin{ccorollary}{\ref{cor:hourglass-free}}
  Let $(\G,\sigma)$ be an hourglass-free BMG.  Then its symmetric part
  $(G,\sigma)$ is either a co-RBMG or it contains an induced $P_4$ on three
  colors whose endpoints have the same color, but no induced cycle $C_n$ on
  $n\geq 5$ vertices.
\end{ccorollary}
As outlined in Sec.~\ref{APP:ssec:hourglass-free}, all \ufp edges in an
hourglass-free BMG are identified by the good and ugly quartets, which are
3-colored by construction. In hourglass-free BMGs, it is indeed sufficient
to consider only the 3-colored $P_4$s to identify all \ufp edges and thus,
to obtain an orthology graph, even though the BMG may also contain
4-colored $P_4$s.  Since hourglasses can only appear in BMGs that require
multifurcations for their explanation (cf.\ Lemma~\ref{lem:hourglass}), the
case of hourglass-free BMGs is the most relevant for practical
applications.

Since all \ufp edges in an hourglass-free BMG are contained in quartets, it
is also easy to identify the hourglass-free BMGs that are already orthology
graphs.
\begin{ccorollary}{\ref{cor:hourglass-free-coRBMG}}
  Let $(\G,\sigma)$ be an hourglass-free BMG. Then, its symmetric part
  $(G,\sigma)$ is a co-RBMG if and only if there are no \ufp edges in
  $(\G,\sigma)$.
\end{ccorollary}

\paragraph{\ufp Edges in Hourglass Chains.}
The situation is much more complicated in the presence of hourglasses. We
start by providing sufficient conditions for \ufp edges that are identified
by hourglass chains.
\begin{cproposition}{\ref{prop:hourglass-ufp}} 	
  Let
  $\mathfrak{H}=[x_1 y_1 \hourglass x'_1 y'_1],\dots,[x_k y_k \hourglass
  x'_k y'_k]$ be an hourglass chain in $(\G,\sigma)$, possibly with a left
  tail $z$ or a right tail $z'$.  Then, an edge in $\G$ is \ufp if it is
  contained in the set
  \begin{align*}
  F =
  & \{x_iy_j\mid 1\leq i \leq j \leq k\}
  \cup\{zz'\}
  \cup\{zy_{i}, x_iz', zy'_{i}, x'_{i}z' \mid 1 \leq i \leq k \}\\
  & \cup\{ x_{i}x_{j+1} \mid 1\le i < j < k \} 
  \cup \{ y_{i}y_{j+1} \mid 1\le i < j < k \} \\
  & \cup\{x'_1 y'_i, x'_1 y_i \mid 2 \leq i \leq k \}
  \cup\{x_i y'_k, x'_i y'_k \mid 1 \leq i \leq k-1 \} \\
  & \cup\{x'_1 z, x'_1 z', y'_k z, y'_k z'\}
  \end{align*} 
\end{cproposition}

As outlined in Sec.~\ref{APP:ssec:hchain}, hourglass chains identify
false-positive edges that are not associated with quartets in the BMG and,
in particular, false-positive edges that are not even part of an induced
$P_4$.  This observation limits the use of cograph editing in the context
of orthology detection, at least in the case of gene trees with polytomies:
On one hand, an RBMG can be a cograph and still contain \ufp edges and, on
the other hand, there are examples where deletion of the \ufp edge
identified by quartets (and thus, by induced $P_4$s) is not sufficient to
arrive at a cograph (cf.\ Sec.~\ref{APP:ssec:hchain}).

\paragraph{Four-colored $P_4$s}
\citet[Thm.~8]{Geiss:19b} established that the RBMG $(G,\sigma)$ is a
co-RBMG, i.e., a cograph, if and only if every subgraph induced on three
colors is a cograph. Therefore, if $(G,\sigma)$ contains an induced
4-colored $P_4$, it also contains an induced 3-colored $P_4$. For
hourglass-free BMGs $(\G,\sigma)$ it is clear that a 4-colored $P_4$ always
overlaps with a 3-colored $P_4$: In this case $\NH(\G,\sigma)$ is obtained
by deleting middle edges of good quartets and first edges of ugly quartets.
Since $\NH(\G,\sigma)$ is a cograph, there is no $P_4$ left, and thus at
least one edge of any 4-colored $P_4$ was among the deleted edges. It is
natural to ask whether this is true for BMGs in general. However, as shown
in Sec.~\ref{APP:ssec-4colP4}, good and ugly quartets are not sufficient on
their own and there are examples with 4-colored $P_4$s that do not overlap
with the middle edge of a good quartet or the first edge of an ugly
quartet.

Still, in the context of cograph-editing approaches it is of interest
whether the 3-colored $P_4$s are sufficient. In the following we provide
an affirmative answer.
\begin{clemma}{\ref{lem:4col-P4}}
  Let $(\G,\sigma)$ be a BMG and $\mathscr{P}$ a 4-colored induced
  $P_4$ in the symmetric part of $(\G,\sigma)$. Then at least one of
  the edges of $\mathscr{P}$ is either the middle edge of some good
  quartet or the first edge of a bad or ugly quartet in $(\G,\sigma)$.
\end{clemma}

It is important to recall in this context, however, that the deletion of
all \ufp-edges identified by quartets does not necessarily lead to a
cograph (see Fig.~\ref{fig:evenhg}(C) in Sec.~\ref{APP:ssec-4colP4} for an
example). Hence, the quartets alone therefore cannot provide a complete
algorithm for correcting an RBMG to an orthology graph.

\section{Simulation results}
\label{sec:simulations}

We illustrate the potential impact of our mathematical results discussed in
the previous sections with the help of simulated data. To this end, we
focus on the accuracy of the inferred orthology graph \emph{assuming} that
the best matches are accurate. Of course, this is only one of several
components in complete orthology detection pipeline, which would also need
to consider the genome annotation, pairwise alignments of genes or
predicted protein sequences, and the conversion of sequence similarities
into best match data. The latter step has been investigated in considerable
detail by \citet{Stadler:20a}. Here, we start from simulated evolutionary
scenarios and extract the BMG directly from the ground truth using the
simulation library \texttt{AsymmeTree} \cite{Stadler:20a}.

In brief, \texttt{AsymmeTree} generates realistic evolutionary scenarios in
four steps. (1)~A planted species tree $S$ is generated using the
Innovation Model \cite{Keller:2012}, which models observed phylogenies
well. (2)~A dating map $\tau$ assigns time points to all vertices of $S$
and thus branch lengths to the edges of $S$. (3)~On~$S$, we use a variant of
the well-known constant-rate birth-death process with a given age
\citep[see e.g.][]{Kendall:1948,Hagen:2018} to simulate an event-labeled
gene tree $(T,t,\sigma)$ containing duplication and loss events.
Speciations are included as additional branching events that generate
copies of all genes present at a speciation vertex in all descendant
lineages. The simulated gene trees are constrained to have at least one
surviving gene in each species to avoid trivial cases. (4)~The observable
part of the gene tree is extracted by recursively removing leaves that
correspond to loss events and suppressing inner vertices with a single
child.  \texttt{AsymmeTree} can also assign rates to edges of
$(T,t,\sigma)$ to convert evolutionary time differences into general
additive distances; however, this is not relevant here since the rates do
not affect evolutionary relatedness and thus the BMG.

\begin{figure}[t]
  \begin{center}
    \includegraphics[width=0.85\textwidth]{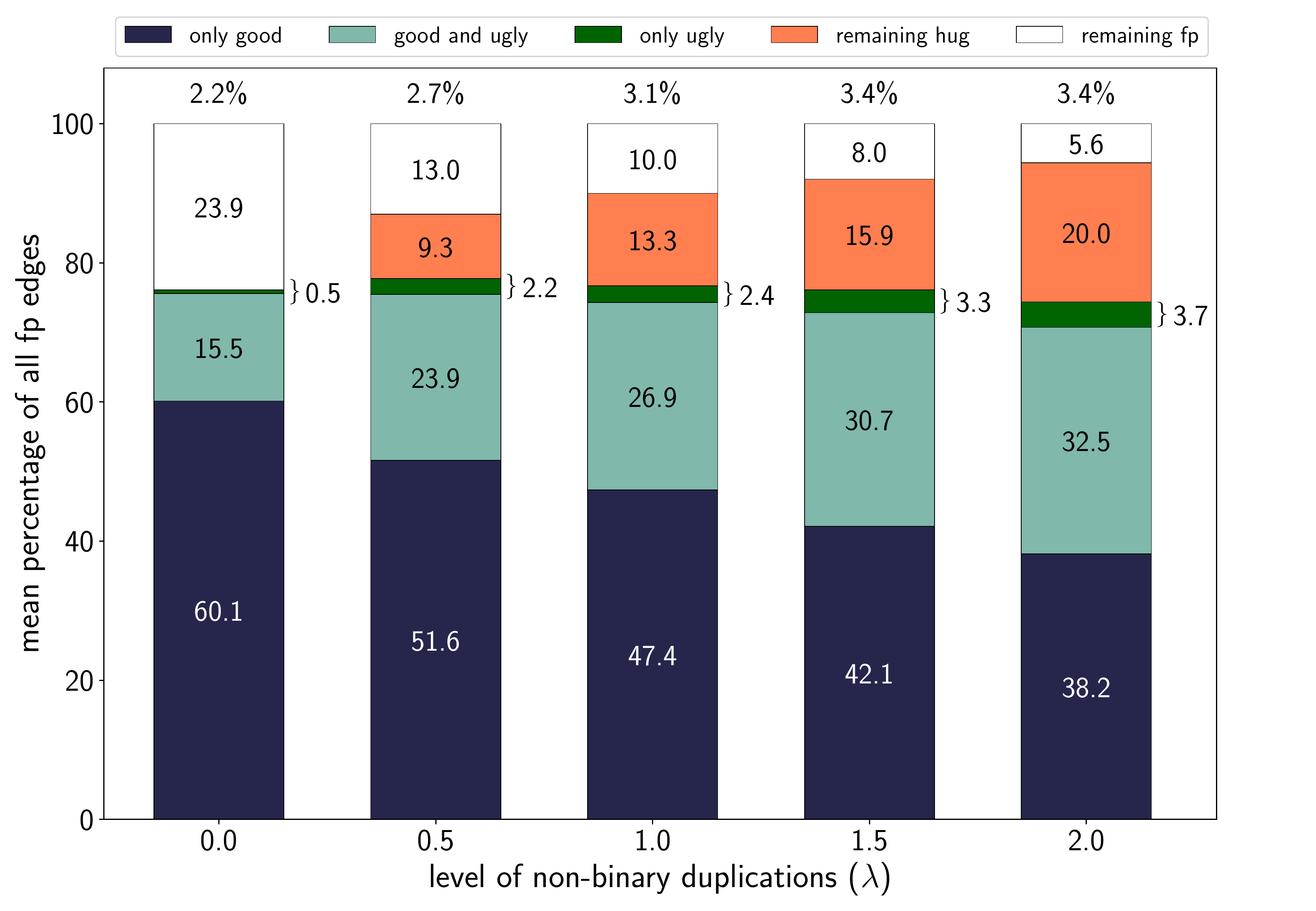}
  \end{center}
  \caption{Average relative abundance of the different types of hug-edges
    and undetectable false positives in the BMGs of simulated evolutionary
    scenarios. We distinguish hug-edges in good and ugly quartets as well
    as hug-edges appearing only in hourglass chains (orange).  In the
    simulations, the fraction of \ufp edges that are first edges of bad
    quartets is too small too be visible and therefore not shown here.  The
    undetectable false positives correspond to complementary gene losses
    without surviving witnesses of the duplication event.  Species trees
    are binary, while gene trees contain multifurcations. The number of
    offsprings is modeled as $2+k$, where $k$ is drawn from a Poisson
    distribution with parameter $\lambda$. For $\lambda=0$, the gene trees
    are binary.  In the experiments, we observed that on average 62.4\% of
    the $25000$ simulated BMGs do not contain any false-positive edge (cf.\
    Fig.~\ref{fig:fdr_heatmap}). Those instances are included in the
    computation of the fraction $|\mathfrak{F}|/|E(G)|$ (percentage above
    the bars).  However, for the computation of all other values only
    scenarios that contain false-positives are considered. }
  \label{fig:hug_percentage}
\end{figure}

Extending the simulations used in \cite{Geiss:20a,Stadler:20a}, we also
consider non-binary gene trees. This is important here since, by
Lemma~\ref{lem:hourglass}, hourglasses cannot appear in BMGs that are
explained by a binary tree. There is an ongoing discussion to what extent
polytomies in phylogenetic trees are biological reality as opposed to an
artifact of insufficient resolution. At the level of species trees, the
assumption that cladogenesis occurs by a series of bifurcations
\citep[e.g.][]{Maddison:89,DeSalle:94} seems to be prevailing, several
authors have argued quite convincingly that there is evidence for a least
some \textit{bona fide} multifurcations of species
\cite{Kliman:00,Takahashi:01,Sayyari:18}. In the simulation, polytomies in
species trees are introduced after the first step by edge contraction with
a user-defined probability $p$.

The reality of polytomies is less clear for gene trees. One reason is the
abundance of tandem duplications. Although the majority of tandem arrays
comprises only a pair of genes, larger clusters are not at all rare
\cite{Pan:08}. Although one may argue that mechanistically they likely
arise by stepwise duplications, such arrangements are often subject to gene
conversion and non-homologous recombination that keeps the sequences nearly
identical for some time before they eventually escape from concerted
evolution and diverge functionally \cite{Liao:99,Hanada:18}. As a
consequence, duplications in tandem arrays may not be resolvable unless
witnesses of different stages of an ongoing duplication process have
survived. To model polytomies in the gene tree, we modify step (3) of the
simulation procedure by replacing a simple duplication by the generation of
$2+k$ offspring genes. The number $k$ of additional copies is drawn from a
Poisson distribution with parameter $\lambda>0$.

The simulated data set of evolutionary scenarios comprises species trees
with 10 to 30 species (drawn uniformly). The time difference between the
planted root and the leaves of $S$ is set to unity. The duplication and
loss rates in the gene trees are drawn i.i.d.\ from the uniform
distribution on the interval $[0.5,1.5)$. Multifurcating gene trees were
produced for $\lambda=\{0.0, 0.5, 1.0, 1.5, 2.0\}$. In total, we generated
5000 scenarios for each choice of $p$ and $\lambda$. Since the true
scenarios, and thus the true gene tree $T$, the true BMG $\G$, and the
corresponding RBMG $G$ are known, we can also determine the set
\begin{equation}
\mathfrak{F}\coloneqq\left\{xy \;\mid\; xy\in E(G) \;\;\textrm{and}\;\;
t(\lca_T(x,y))=\DUPL\right\}. \
\end{equation}
of false-positive edges. From the BMG, we compute the set $\mathfrak{U}$ of
\ufp edges as well as the subsets $\mathfrak{U}_M$ and $\mathfrak{U}_U$ of
\ufp edges that are middle edges of a good or first edges of an ugly quartet,
respectively. Note that in general we have
$\mathfrak{U}_M\cap \mathfrak{U}_U\ne\emptyset$.  We only discuss the
results for binary species trees in some detail, since species trees with
polytomies yield qualitatively similar results. We observe that the
relative abundance of \ufp edges in good and ugly quartets increases
moderately for larger $p$.

\begin{figure}[t]
  \begin{center}
    \includegraphics[width=0.85\textwidth]{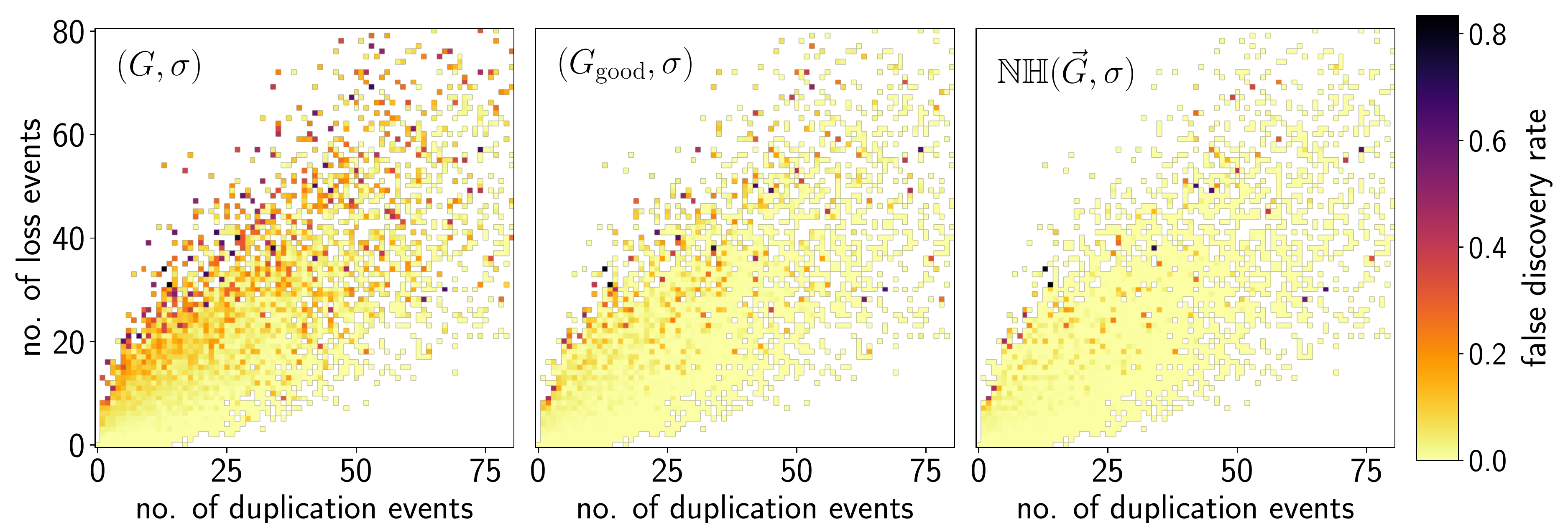}
  \end{center}
  \caption{False discovery rates computed as proportion of \fp among all
    edges averaged over all scenarios with given number of duplications and
    losses. \emph{Left:} RBMGs $(G,\sigma)$, i.e., $|\mathfrak{F}|/|E(G)|$.
    \emph{Middle:} edited RBMG $(G_\textrm{good},\sigma)$ with all middle
    edges of good quartets removed, i.e.,
    $|\mathfrak{F}\setminus\mathfrak{U}_M|/|E(G_\textrm{good})|$. \emph{Right:}
    no-hug graphs $\NH(\G,\sigma)$, i.e.,
    $|\mathfrak{F}\setminus\mathfrak{U}|/|E(\NH)|$.  Scenarios with more
    than 80 duplication/loss events are not shown.}
  \label{fig:fdr_heatmap}
\end{figure}

First, we note that, consistent with \cite{Geiss:20a,Stadler:20a}, the
fraction $|\mathfrak{F}|/|E(G)|$ of false positive orthology assignments is
small in our data set, on the order of $3\%$. This indicates that, in
real-life data, the main source of errors is likely the accurate
determination of best matches from sequence data rather than false-positive
edges contained in the BMG. Considering the fraction
$|\mathfrak{U}|/|\mathfrak{F}|$ of \ufp edges in
Fig.~\ref{fig:hug_percentage}, we find that even in the most adverse case
of all gene trees being binary, the BMG identifies more than three quarters
of $\mathfrak{F}$. It may be surprising at first glance that the problem
becomes easier with increasing $\lambda$ and barely $6\%$ of the false
positives escape discovery. A likely explanation is that multifurcations
increase the likelihood that an inner vertex has two surviving lineages
that serve as witnesses of the event; in addition, multifurcations increase
the vertex degree in the BMG, so that in principle more information is
available to resolve the tree structure.  It is also interesting to note
that $\mathfrak{U}_U\setminus \mathfrak{U}_M$ is small, i.e., there are few
cases of first edges in an ugly quartet that are not also middle edges in a
good quartet. The fraction of \ufp edges that appear only as first edges of
bad quartets is even smaller; only 2-3\% of the \ufp edges associated with
hourglass chains, i.e., less than 0.15\% of all \ufp edges are of this
type. The overwhelming majority of \ufp edges associated with quartets thus
appear (also) as middle edges of good quartets.  This observation provides
an explanation for the excellent performance of removing the
$\mathfrak{U}_M$-edges proposed in \cite{Geiss:20a}. In particular in the
case of binary trees, which was considered by \citet{Geiss:20a}, there is
only a small number of other \ufp edges, which are completely covered by
$\mathfrak{U}_U$.  Fig.~\ref{fig:fdr_heatmap} visualizes the appearance of
false-positive edges depending on the number of duplication and loss
events. Not surprisingly, $\mathfrak{F}$ is enriched in scenarios with a
large number of losses compared to the duplications, and depleted when
losses are rare. In fact, in the absence of losses, the RBMG equals the
orthology graph, i.e., $\mathfrak{F}=\emptyset$
\cite[Thm.~4]{Geiss:20a}. Removal of $\mathfrak{U}_M$, already reduced the
false positives considerably.

\section{Summary and outlook}
\label{sec:summary}

We have shown here how all unambiguously false-positive orthology
assignments can be identified in polynomial time provided that all best
matches are known. In particular, we have provided several
characterizations for \ufp edges in terms of underlying subgraphs and
refinements of trees.  Since the best match graph contains only false
positives, we have obtained a characterization of \emph{all} unambiguously
incorrect orthology assignments. Simulations showed that the majority of
false positives comprises middle edges of good quartets, while \ufp edges
that appear only as first edges of an ugly quartet are rare. Not
surprisingly, the hourglass-related \ufp edges become important in gene
trees with many multifurcations. They do not appear in scenarios derived
from binary gene trees. For the theory developed here, it makes no
difference whether polytomies in the gene tree appear as genuine features,
or whether limited accuracy of the approximation from underlying sequence
data produced the equivalent of a soft polytomy in the BMG.

The augmented tree $(\aug(T^*),\sigma)$ is the least resolved tree that
admits an event labeling such that all inner vertices with child trees
that have overlapping colors are designated as duplications while all inner
vertices with color-disjoint child trees are designated as speciations.
The tree $(\aug(T^*),\sigma)$ therefore does not contain ``non-apparent
duplications'' in the sense of \cite{Lafond:14b}, i.e., duplication
vertices with species-disjoint subtrees. This is an interesting connection
linking the literature concerned with polytomy refinement in given gene
trees \cite{Chang:06,Lafond:14b} with best match graphs. 

The extremal event labeling $\wt$ of $(\aug(T^*),\sigma)$ is the one that
minimizes the necessary number of duplications on $(\aug(T^*),\sigma)$. In
a conceptual sense, therefore, $(\aug(T^*),\wt)$ is a ``most parsimonious''
solution, matching the idea of most parsimonious reconciliations
\cite{Guigo:96,Page:97}. From a technical point of view, however, the
problem we solve here is very different. Instead of considering a given
pair of gene tree $T$ and species tree $S$, we ask here about the
information contained in the BMG $(\G,\sigma)$, i.e., we only consider the
information on the species tree that is already implicitly contained in
$(\G,\sigma)$.  The construction of the event-labeled gene tree
$(\aug(T^*),\wt)$ in fact \emph{implies} a set $\mathfrak{S}$ of
informative triples, namely those $\sigma(x)\sigma(y)|\sigma(z)$ with
$\sigma(x)$, $\sigma(y)$, $\sigma(z)$ pairwise distinct and
$\wt(\lca_{\aug(T^*)}(x,y,z))=\SPEC$, that are displayed by the species
tree $S$ \cite{HernandezRosales:12a,Hellmuth:17}. Nothing in our theory,
however, ensures that $\mathfrak{S}$ is a consistent set of triples, much
less that $\mathfrak{S}$ is consistent with a given species tree $S$. A
lack of consistency, however, implies that the no-hug graph
$\NH(\G,\sigma)$ cannot be the correct orthology relation, and thus,
necessarily contains additional false-positive edges. Consistency, on the
other hand, cannot provide a mathematical proof for biological
correctness. It makes $\NH(\G,\sigma)$ a very likely candidate for the true
orthology relation, however, because alternative scenarios require
additional gene duplications and multiple, strategically placed gene losses
to compensate for them.

Since constraints on reconciliation maps deriving from the species
phylogeny are fully expressed by informative triples, no such constraint
exists in particular for any vertex $u$ of $\aug(T^*)$ that has only leaves
as children. That is, false-positive orthology assignments among the
children of $u$ cannot be identified from the BMG alone because there are
no further descendants to witness $u$ as duplication event. Additional
evidence, such as the assumption of a molecular clock or synteny must be
used to resolve situations such as the complementary loss shown in
Fig.~\ref{fig:compl_loss}.

Every gene tree $T$ can be reconciled with every species tree $S$
\cite{Guigo:96,Page:97,Geiss:20a} at the expense of reassigning events as
duplications.  If $\aug(T^*)$ is already binary, consistency will require
the relabeling of some speciation nodes as duplications. Can one
characterize and efficiently compute the minimal relabelings? In the
general case, a further refinement of $\aug(T^*)$ may be sufficient. Is a
refinement of speciation nodes sufficient, or are there in general
speciation nodes in $(\aug(T^*),\wt)$ that need to be refined into separate
speciation and duplication events?

Since orthology graphs are cographs contained in the RBMG $(G,\sigma)$, it
is of interest to compare the deletion of all \ufp edges in $(G,\sigma)$
with finding a (minimal) edge-deletion set to obtain a cograph.  These two
problems are clearly distinct: The simplest example is the BMG
$(\G,\sigma)$ in Fig.~\ref{fig:hourglasses}(A): its symmetric part $G$ is
already a cograph but $(\G,\sigma)$ contains the hug-edge $xy$, which must
be deleted. Despite its practical use \cite{Hellmuth:15,Lafond:16}, this
observation relegates cograph editing \cite{Liu:12,Hellmuth:20b,Tsur:20} to
the status of a heuristic approximation for the purpose of orthology
detection.

For practical applications, one has to keep in mind that best matches are
inferred from sequence similarity data. Despite efforts to convert best
(blast) hits into evolutionary best matches in a systematic manner
\cite{Stadler:20a}, estimated BMGs will contain errors, which in most cases
will violate the definition of best match graphs.  This begs the question
how an empirical estimate of a BMG can be corrected to a closest
``correct'' BMG that (approximately) fits the data. Not surprisingly, BMG
editing \cite{Schaller:20d} and the analogous RBMG editing problem
\cite{Hellmuth:20a} are NP-hard. Efficient, accurate heuristics are a topic
of ongoing research.

Orthology prediction tools intended for large data sets often do not
attempt to infer the orthology graph, but instead are content with
summarizing the information as \emph{clusters of orthologous groups} (COGs)
in an empirically estimated RBMG \cite{Tatusov1997,Roth:08}. Formally, this
amounts to editing the BMG to a set of disjoint cliques. The example in
Fig.~\ref{fig:messy_vertices-1} shows that this approach can destroy
correct orthology information: the BMG $(\G,\sigma)$ does not contain \ufp
edges and thus, it is the closest orthology graph. However, $(\G,\sigma)$
is not the disjoint union of cliques.

\vspace{5mm}
\textbf{Acknowledgments.} 
We thank Carsten R.\ Seemann for fruitful discussions and his helpful
comments.  Moreover, we thank the anonymous reviewers for their important
and valuable comments that helped to significantly improve the
paper. This work was supported in part by the Austrian Federal Ministries
BMK and BMDW and the Province of Upper Austria in the frame of the COMET
Programme managed by FFG, and by the German Research Foundation (DFG,
grant no.\ STA 850/49-1).

\begin{appendix}
  \section*{TECHNICAL PART}
  
  \section{(Reciprocal) best matches}
  \label{APP:subsect:rbmg}
  
  We start by collecting some useful properties of BMGs and RBMGs that will
  be needed for later reference.
  \begin{lemma}{\cite[Lemma~10]{Geiss:19b}}
    Let $(T,\sigma)$ be a leaf-colored tree on $L$ and let $v\in V(T)$. Then,
    for any two distinct colors $r,s\in \sigma(L(T(v)))$, there is an edge
    $xy$ in $\G(T,\sigma)$ with $x\in L[r]\cap L(T(v))$ and
    $y\in L[s]\cap L(T(v))$. 
    \label{lem:exEdge}
  \end{lemma}
  
  \begin{lemma}\label{lem:edge-xy-lca}
    Let $(\G,\sigma)$ be a BMG explained by a tree $(T,\sigma)$.  Moreover,
    let $x,y \in L(T)$ with $\sigma(x)\neq \sigma(y)$ and
    $v_x,v_y\in\child(\lca_T(x,y))$ with $x\preceq_Tv_x$ and $y\preceq_Tv_y$.
    Then, $\sigma(x)\notin \sigma(L(T(v_y)))$ and
    $\sigma(y)\notin \sigma(L(T(v_x)))$ if and only if $xy$ is an edge in
    $\G$.
  \end{lemma}
  \begin{proof}
    By the definition of best matches, it holds that $xy$ is an edge in $\G$
    if and only if $\lca_T(x,y) \preceq_T \lca_T(x,y')$ for all $y'\in L(T)$
    of color $\sigma(y)$ and $\lca_T(x,y) \preceq_T \lca_T(x',y)$ for all
    $x'\in L(T)$ of color $\sigma(x)$.  Clearly,
    $\lca_T(x,y) \preceq_T \lca_T(x,y')$ for all such $y'$ if and only if
    $\sigma(y)\notin \sigma(L(T(v_x)))$, and
    $\lca_T(x,y) \preceq_T \lca_T(x',y)$ for all such $x'$ if and only if
    $\sigma(x)\notin \sigma(L(T(v_y)))$.  
  \end{proof}
  
  \begin{definition}
    Suppose that $(T,\sigma)$ explains $(\G,\sigma)$. Then we say that
    $(T,\sigma)$ is \emph{least resolved} (w.r.t.\ $(\G,\sigma)$) if no tree
    $(T',\sigma)$ displayed by $(T,\sigma)$ explains $(\G,\sigma)$.
  \end{definition}
  Recall all trees in this contribution are planted, and thus least resolved
  trees (LRTs) are also considered as planted. Strictly speaking, this
  differs from the construction in \cite{Geiss:19a,Geiss:19b,Geiss:20a}, the
  additional (non-contractible) edge $0_T\rho_T$ is a trivial detail that
  does not affect the properties of LRTs.
  \begin{theorem}{\cite[Thm.~8 and Cor.~4]{Geiss:19a}}
    Every BMG $(\G,\sigma )$ is explained by a unique least resolved tree
    $(T^*,\sigma )$.  In particular, every other tree $(T,\sigma)$ explaining
    $(\G,\sigma)$ is a refinement of $(T^*,\sigma)$.  The least resolved tree
    $(T^* , \sigma )$ of a BMG $(\G, \sigma )$ can be constructed in
    polynomial time.
    \label{thm:LRT}
  \end{theorem}
  
  \noindent The following definition of informative triples is equivalent to
  the version given by \cite{Geiss:19a}.
  \begin{definition}\label{def:informative_triples}
    Let $(\G,\sigma)$ be a colored digraph. We say that a triple $ab|b'$ is
    \emph{informative} for $(\G,\sigma)$ if $a$, $b$ and $b'$ are three
    different vertices with $\sigma(a)\neq\sigma(b)=\sigma(b')$ in $\G$ such
    that $(a,b)\in E(\G)$ and $(a,b')\notin E(\G)$.
  \end{definition}
  
  \begin{lemma}
    \label{lem:informative_triples}
    Let $(\G,\sigma)$ be a BMG and $ab|b'$ an informative triple for
    $(\G,\sigma)$.  Then, every tree $T$ that explains $(\G,\sigma)$ displays
    the triple $ab|b'$, i.e. $\lca_T(a,b)\prec_T\lca_T(a,b')=\lca_T(b,b')$.
  \end{lemma}
  \begin{proof}
    The definition of informative triples implies that $(a,b)\in E(\G)$ and
    $(a,b')\notin E(\G)$.  Using $\sigma(b)=\sigma(b')$ and the definition of
    best matches we immediately conclude $\lca_T(a,b)\prec_T\lca_T(a,b')$.
  \end{proof}
  
  \begin{lemma}\label{lem:inf_triples_overlap}
    Let $ab|b'$ and $cb'|b$ be informative triples for a BMG
    $(\G,\sigma)$. Then every tree $(T,\sigma)$ that explains $(\G,\sigma)$
    contains two distinct children $v_1, v_2\in\child_{T}(\lca_{T}(a,c))$
    such that $a,b\prec_T v_1$ and $b',c\prec_T v_2$.
  \end{lemma}
  \begin{proof}
    Let $(T,\sigma)$ be an arbitrary tree that explains $(\G,\sigma)$.  By
    Lemma~\ref{lem:informative_triples}, $T$ displays the informative triples
    $ab|b'$ and $cb'|b$.  Thus we have
    $\lca_{T}(a,b)\prec_T\lca_{T}(a,b')=\lca_{T}(b,b')$ and
    $\lca_{T}(c,b')\prec_T\lca_{T}(c,b)=\lca_{T}(b,b')$.  In particular,
    $\lca_{T}(a,b')=\lca_{T}(b,b')=\lca_{T}(c,b)\eqqcolon u$. Therefore,
    $a\preceq_{T} v_1$ and $b'\preceq_{T} v_2$ for distinct
    $v_1, v_2\in\child_{T}(u)$.  Since $\lca_{T}(a,b)\prec_T u$, we have
    $a,b\prec_T v_1$ and thus $v_1$ is an inner node.  Likewise,
    $\lca_{T}(b',c)\prec_{T}u$ implies $b',c\prec_{T}v_2$.  
  \end{proof}
  
  Given a tree $T$ and an edge $e$, denote by $T_e$ the tree obtained from
  $T$ by contracting the edge $e$. An edge $e\ne0_T\rho_T$ in $(T,\sigma)$ is
  \emph{redundant (w.r.t.\ $(\G,\sigma)$)} if $(T,\sigma)$ explains
  $(\G,\sigma)$ and $\G(T_e,\sigma)=\G(T,\sigma)$. Redundant edges have
  already been characterized in \cite[Lemma~15, Thm.~8]{Geiss:19a} in terms
  of equivalence classes using a more complicated notation. Here we
  give a simpler characterization:
  
  \begin{lemma}
    \label{lem:redundant_edges}
    Let $(\G,\sigma)$ be a BMG explained by a tree $(T,\sigma)$.  The edge 
    $e=uv$
    with $v\prec_T u$ in $(T,\sigma)$ is redundant w.r.t.\ $(\G,\sigma)$ if
    and only if (i) $e$ is an inner edge of $T$ and (ii) there is no arc
    $(a,b)\in E(\G)$ such that $\lca_T(a,b)=v$ and
    $\sigma(b)\in \sigma(L(T(u))\setminus L(T(v)))$.
  \end{lemma}
  \begin{proof}
    Let $w_e$ be the vertex in $T_e$ resulting from the contraction $e=uv$
    with $v\prec_T u$ in $T$. By assumption we have
    $(\G,\sigma)=\G(T,\sigma)$.
    
    First, assume that $e$ is redundant and thus,
    $\G(T_e,\sigma)=\G(T,\sigma)$.  Then $e$ must be an inner edge, since
    otherwise $L(T)\ne L(T_e)$ and, therefore, $(T_e,\sigma)$ does not
    explain $(\G,\sigma)$.  Now assume, for contradiction, that there is an
    arc $(a,b)\in E(\G)$ such that $\lca_T(a,b)=v$ and
    $\sigma(b)\in \sigma( L(T(u))\setminus L(T(v)))$. Then there is a leaf
    $b'\in L(T(u))\setminus L(T(v))$ with $\sigma(b')=\sigma(b)$ and
    $\lca_T(a,b)=v\prec_{T} u = \lca_T(a,b')$.  Thus, $(a,b')\notin
    E(\G)$. After contraction of $e$, we have
    $\lca_T(a,b)=\lca_T(a,b') = w_e$. Hence, by definition of best matches,
    $(a,b)$ is an arc in $\G(T_e,\sigma)$ if and only if $(a,b')$ is an arc
    in $\G(T_e,\sigma)$; a contradiction to the assumption that
    $(T_e,\sigma)$ explains $(\G,\sigma)$.
    
    Conversely, assume that $e=uv$ with $v\prec_T u$ is an inner edge in $T$
    and that there is no arc $(a,b)\in E(\G)$ such that $\lca_T(a,b)=v$ and
    $\sigma(b)\in \sigma(L(T(u))\setminus L(T(v)))$.  In order to show that
    an edge $e$ is redundant, we need to verify that
    $\G(T,\sigma) = \G(T_e,\sigma)$. To this end, consider an arbitrary leaf
    $c\in L(T)$. Then we have either Case (1) $c\in L(T)\setminus L(T(v))$,
    or Case (2) $c\in L(T(v))$.
    
    In Case (1) it is easy to verify that $\lca_{T}(c,d)=\lca_{T_e}(c,d)$ for
    every $d\in L(T)$. In particular, therefore, $(c,d)\in E(\G(T,\sigma))$
    if and only if $(c,d)\in E(\G(T_e,\sigma))$.
    
    In Case (2), i.e.\ $c\in L(T(v))$, consider another, arbitrary, leaf
    $d\in L(T)$. Note, if $\sigma(c)=\sigma(d)$, then $c$ and $d$ never form
    a best match. Thus, we assume $\sigma(c)\neq\sigma(d)$. Now, we consider
    three mutually exclusive Subcases (a) $\lca_T(c,d)\preceq_{T} v$, (b)
    $\lca_T(c,d)=u$ and (c) $\lca_T(c,d)\succ_T u$.
    
    \textit{Case (a).} Since no edge below $v$ is contracted, we have for
    every $d'$ with $\sigma(d')=\sigma(d)$,
    $\lca_{T}(c,d')\prec_{T}\lca_{T}(c,d)\preceq_{T}v$ if and only if
    $\lca_{T_e}(c,d')\prec_{T_e}\lca_{T_e}(c,d)\preceq_{T_e} w_e$.  In
    particular, therefore, $(c,d)\in E(\G(T,\sigma))$ if and only if
    $(c,d)\in E(\G(T_e,\sigma))$.
    
    \textit{Case (b).} $\lca_T(c,d)=u$ and $c\prec_T v$ implies that
    $d\in L(T(u)\setminus L(T(v))$ and thus,
    $\sigma(d)\in\sigma( L(T(u))\setminus L(T(v)) )$. If
    $(c,d)\in E(\G(T,\sigma))$, then $\sigma(d)\notin \sigma(L(T(v)))$ must
    hold. Therefore, $(c,d)$ is still an arc after contraction of $e$. For
    the case $(c,d)\notin E(\G(T,\sigma))$, assume for contradiction
    $(c,d)\in E(\G(T_e,\sigma))$. Then $(c,d)\notin E(\G(T,\sigma))$ implies
    that there must be a vertex $d'$ with $\sigma(d')=\sigma(d)$ and
    $\lca_T(c,d')\preceq_T v\prec_T u=\lca_T(c,d)$. In particular,
    $d'\in L(T(v))$ can be chosen such that $\lca_T(c,d')$ is farthest away
    from $v$ and thus, $(c,d')\in E(\G(T,\sigma))$. Now,
    $\lca_T(c,d')\preceq_T v$ and $(c,d)\in E(\G(T_e,\sigma))$ imply that
    $\lca_{T_e}(c,d')= w_e=\lca_{T_e}(c,d)$, which is only possible if
    $\lca_T(c,d')= v$. In summary, we found an arc
    $(c,d')\in E(\G(T,\sigma))$ with $\lca_T(c,d')= v$ and
    $\sigma(d') \in\sigma( L(T(u))\setminus L(T(v)))$; a contradiction to our
    assumption.  Hence, in Case (b) we have $(c,d)\in E(\G(T,\sigma))$ if and
    only if $(c,d)\in E(\G(T_e,\sigma))$.
    
    \textit{Case (c).} Since $\lca_T(c,d)\succ_T u$, it is again easy to see
    that, for every $d'$ with $\sigma(d')=\sigma(d)$,
    $\lca_T(c,d')\prec_T \lca_T(c,d)$ if and only if
    $\lca_{T_e}(c,d')\prec_{T_e} \lca_{T_e}(c,d)$ and thus,
    $(c,d)\in E(\G(T,\sigma))$ if and only if $(c,d)\in E(\G(T_e,\sigma))$.
    
    In summary, we have $(c,d)\in E(\G(T,\sigma))$ if and only if
    $(c,d)\in E(\G(T_e,\sigma))$ for all $c,d\in L(T)$. Thus, $e$ is
    redundant.  
  \end{proof}
  
  As a consequence of Lemma~\ref{lem:redundant_edges}, we obtain
  \begin{corollary}
    \label{cor:edge_redundant}
    Let $(T,\sigma)$ be a leaf-colored tree explaining $(G,\sigma)$ and $uv$
    an inner edge inner of $T$ with $v\prec_T u$.  If
    $\sigma(L(T(v)))\cap\sigma(L(T(v')))=\emptyset$ for every
    $v'\in\child_{T}(u)\setminus\{v\}$, then $uv$ is redundant in $T$
    (w.r.t.\ $(G,\sigma)$).
  \end{corollary}
  \begin{proof}
    If there is an arc $e=(a,b)\in E(\G)$ with $\lca_T(a,b)=v$ we have
    $\sigma(b)\notin L(T(u))\setminus L(T(v)) = \cup_{v'\in
      \child(u)\setminus \{v\}} L(T(v'))$ because
    $\sigma(L(T(v)))\cap\sigma(L(T(v')))=\emptyset$ for every
    $v'\in\child_{T}(u)\setminus\{v\}$. By Lemma~\ref{lem:redundant_edges},
    the inner edge $uv$ is redundant.  
  \end{proof}
  
  Both Lemma~\ref{lem:redundant_edges} and Cor.~\ref{cor:edge_redundant} are
  illustrated in Fig.~\ref{fig:redundant_edge}: In~(A), $uv$ is a
  non-redundant inner edge since $(a,b)$ is a best match such that $a$ and
  $b$ have $v$ as their last common ancestor and the color of $b$ is present
  in another subtree below vertex $u$.  Contraction of the edge $uv$ would
  result in a tree $T_{uv}$ in which
  $\lca_{T_{uv}}(a,b)=\lca_{T_{uv}}(a,b')$, and thus, introduce the
  additional best match $(a,b')$.  Clearly, this cannot occur whenever the
  other subtrees of $u$ do not share any colors with the subtree $T(v)$, a
  situation that is shown in~(B), i.e., the edge $uv$ is redundant w.r.t.\
  the BMG $\G(T,\sigma)$.
  
  \begin{figure}[t]
    \begin{center}
      \includegraphics[width=0.69\textwidth]{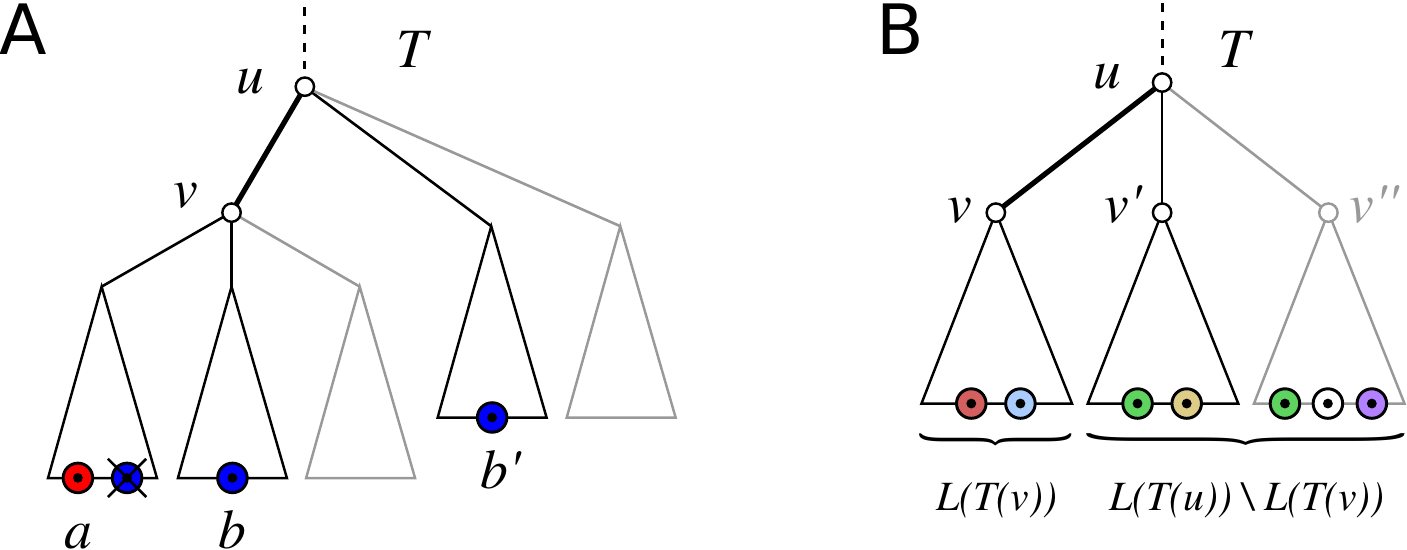}
    \end{center}
    \caption{Visualization of (A) a \emph{non-}redundant edge $uv$ for
      Lemma~\ref{lem:redundant_edges} and (B) a redundant edge $uv$ as in
      Cor.~\ref{cor:edge_redundant}. The gray subtrees may or may not
      exist. In (A), the crossed out leaf indicates that the blue color must
      not be present in this subtree and thus $(a,b)$ is a best match. In
      (B), $\sigma(L(T(v)))$ must not have elements in common with
      $\sigma(L(T(u))\setminus L(T(v)))$.  See text for further details.}
    \label{fig:redundant_edge}
  \end{figure}
  
  Finally, we show that redundant edges can be contracted in arbitrary order,
  similar to \cite[Lemma~6 \& Cor.~2]{Geiss:19a}. To this end, we first prove
  a more general statement.
  \begin{lemma}
    \label{lem:contract-subgraph}
    If $T_A$ is obtained from $T$ by contracting all edges in a subset $A$ of
    inner edges in $T$, then $\G(T,\sigma)\subseteq \G(T_A,\sigma)$.
  \end{lemma}
  \begin{proof}
    First note that $L(T_A)=L(T)$ since $A$ only contains inner edges.  Let
    $(x,y)$ be an arc in $\G(T,\sigma)$. This implies that there is no $y'$
    with $\sigma(y')=\sigma(y)$ such that
    $\lca_T(x,y')\prec_T\lca_T(x,y)$. It is easy to verify that the latter is
    still true after contraction of an arbitrary edge $e$, i.e.\ there is no
    $y'$ with $\sigma(y')=\sigma(y)$ such that
    $\lca_{T_e}(x,y')\prec_{T_e}\lca_{T_e}(x,y)$. Hence, $(x,y)$ is an arc in
    $\G(T_e,\sigma)$.  Now consider the subsets
    $A_1\subset A_2\subset \cdots \subset A_{|A|}=A$ where each $|A_i|=i$,
    $1\leq i\leq |A|$. The argument above implies
    $\G(T,\sigma)\subseteq \G(T_{A_1},\sigma) \subseteq \cdots \subseteq
    \G(T_{A},\sigma)$, which completes the proof.  
  \end{proof}
  
  \begin{lemma}
    \label{lem:redundant_commutative}
    Let $A$ and $B$ be disjoint sets of redundant edges in $(T,\sigma)$
    w.r.t.\ $(\G,\sigma)$ and denote by $T_A$ the tree obtained by
    contraction of all edges in $A$ in arbitrary order.  Then $B$ is a set of
    redundant edges in $T_A$ w.r.t.\ $\G(T_A,\sigma)=\G(T,\sigma)$.
  \end{lemma}
  \begin{proof}
    By Lemma~\ref{lem:contract-subgraph}, contraction of \emph{any} inner
    edge $e=uv\in E(T)$ never leads to a loss of arcs in the BMG
    $(\G,\sigma) = \G(T,\sigma)$.  Furthermore, the redundant edges in $T$
    w.r.t.\ $(G,\sigma)$ are completely characterized by
    Lemma~\ref{lem:redundant_edges}.  Thm.~8 in \cite{Geiss:19a} states
    that by contraction of all redundant edges (in an arbitrary order), one
    obtains the unique least resolved tree $(T^*,\sigma)$ of $(\G,\sigma)$.
    As argued above, no arc of $\G(T,\sigma)$ can be lost in the stepwise
    contraction of redundant edges.  Together with
    $\G(T,\sigma)=\G(T^*,\sigma)=(\G,\sigma)$ this implies
    $\G(T_A,\sigma)=(\G,\sigma)$.  Since by assumption $A\cap B=\emptyset$
    and $A\cup B$ is a set of redundant edges w.r.t.\ $(\G,\sigma)$, we have
    $(T_A)_B=T_{A\cup B}$ and
    $\G(T_A,\sigma)=(\G,\sigma)=\G(T_{A\cup B},\sigma)=\G((T_A)_B,\sigma)$.
    Hence, $B$ is a set of redundant edges in $T_A$ w.r.t.\ $\G(T_A,\sigma)$.
  \end{proof}
  
  \section{False-positive orthology assignments}
  
  \subsection{$(T,\sigma)$-fp and \ufp edges}
  \label{APP:ssect:ufp}
  
  The aim of this contribution is to characterize all those false-positive
  edges in a given BMG $(\G,\sigma)$ that can be identified from the
  structure of the BMG alone, i.e., without any \emph{a priori} knowledge
  about the gene tree, the species tree, or the reconciliation map. In this
  section, we start by considering false-positive edges identifiable with
  respect to a given $(T,\sigma)$ that explains $(\G,\sigma)$ and then
  proceed by considering those edges that are identified by \emph{all} trees
  explaining $(\G,\sigma)$.
  
  \begin{definition}[$\mathbf{(T,\sigma)}$-false-positive]
    Let $(T,\sigma)$ be a tree explaining the BMG $(\G,\sigma)$.  An edge
    $xy$ in $\G$ is called \emph{$(T,\sigma)$-false-positive}, or
    $(T,\sigma)$-\fp for short, if for every reconciliation map $\mu$ from
    $(T,\sigma)$ to any species tree $S$ we have $t_\mu(\lca_T(x,y))=\DUPL$,
    i.e., $\mu(\lca_T(x,y))\in E(S)$.
    \label{def:Ts-fp}
  \end{definition}
  In other words, $xy$ is called
  $(T,\sigma)$-\fp whenever $x$ and $y$ cannot be orthologous w.r.t.\ every
  possible reconciliation $\mu$ from $(T,\sigma)$ to any species
  tree. Interestingly, $(T,\sigma)$-\fp{}s can be identified without
  considering reconciliation maps explicitly.
  
  \begin{lemma}
    \label{lem:T-fp-no-mu}
    Let $(\G,\sigma)$ be a BMG, $xy$ be an edge in $\G$ and $(T,\sigma)$ be a
    tree that explains $(\G,\sigma)$. Then, the following statements are
    equivalent:
    \begin{enumerate}[itemsep=0.2ex, topsep=0.2ex, parsep=0cm]
      \item The edge $xy$ is $(T,\sigma)$-\fp.
      \item There are two children $v_1$ and $v_2$ of $\lca_T(x,y)$ such that
      $\sigma(L(T(v_1)))\cap \sigma(L(T(v_2)))\neq\emptyset$.
      \item For the extremal labeling $\tT$ of $(T,\sigma)$ it holds that
      $\tT(\lca_T(x,y)) = \DUPL$.
    \end{enumerate}
  \end{lemma}
  \begin{proof}
    \par\noindent\textit{(2) implies (1).} Suppose that there are two
    children $v_1$ and $v_2$ of $\lca_T(x,y)$ such that
    $\sigma(L(T(v_1)))\cap \sigma(L(T(v_2)))\neq\emptyset$.  By 
    Lemma~\ref{lem:duplication_witness}, $\mu(\lca_T(x,y))\in E(S)$ and thus,
    $t_{\mu}(\lca_T(x,y))=\DUPL$ for all possible reconciliation maps $\mu$
    from $(T,\sigma)$ to any species tree $S$. Hence, $xy$ is
    $(T,\sigma)$-\fp.
    
    \par\noindent\textit{(1) implies (2).}  By contraposition,
    let $v = \lca_T(x,y)$ and suppose that for all distinct children
    $v_i,v_j\in \child(v)=\{v_1,\dots,v_k\}$, $k\geq 2$ we have
    $\sigma(L(T(v_i)))\cap \sigma(L(T(v_j)))=\emptyset$.  In the following, we
    show that there is a species tree $S$ and a reconciliation map $\mu$ from
    $(T,\sigma)$ to $S$ such that $t_{\mu}(\lca(x,y))=\SPEC$, which implies
    that $xy$ is not $(T,\sigma)$-\fp.
    
    We construct the species tree $S$ as follows: $S$ has root edge
    $0_S\rho_S$. Now add $k$ children $u_1,\dots,u_k$ to $\rho_S$.  For each
    of these children $u_i$ with $|\sigma(L(T(v_i)))|>1$, we add a leaf $t$
    for every color $t\in\sigma(L(T(v_i)))$ and the edge $u_it$. Any other
    $u_i$ is considered to be a leaf in $S$, and we identify $u_i$ with the
    single element in $\sigma(L(T(v_i)))$. Furthermore, add for all
    $t\in \sigma(L(T))\setminus \sigma(L(T(v)))$ a leaf $t$ that is adjacent
    to $\rho_S$. Since the color sets
    $\sigma(L(T))\setminus \sigma(L(T(v))), \sigma(L(T(v_1))), \dots,
    \sigma(L(T(v_k))$ are pairwise distinct, $S$ is well-defined, and, by
    construction, a planted phylogenetic tree. To construct a reconciliation
    map we put (i) $\mu(0_T)= 0_S$; (ii) $\mu(x)=\sigma(x)$ for all
    $x\in L(T)$; (iii) $\mu(v)=\rho_S$; (iv) $\mu(w)= 0_S\rho_S$ for all
    $w\in V^0(T \setminus T(v))$; and (v) $\mu(w)=\rho_Su_i$ for all
    $w\in V^0(T(v_i))$.  By Condition (i) and (ii), the Axioms \AX{(R0)} and
    \AX{(R1)} are satisfied, respectively.  By Condition (v), we have
    $\mu(v_i)=\rho_Su_i$ if $v_i$ is an inner vertex.  Otherwise, $v_i$ is a
    leaf and $|\sigma(L(T(v_i)))|=1$.  Therefore, $\mu(v_i)=\sigma(v_i)=u_i$
    by (ii) and by construction.  It is easy to verify that $\mu$ satisfies
    \AX{(R2)}. A sketch of construction of the species tree $S$ and the
    reconciliation map $\mu$ is provided in
    Fig.~\ref{fig:reconc_construction}.
    
    \begin{figure}[H]
      \begin{center}
        \includegraphics[width=0.65\textwidth]{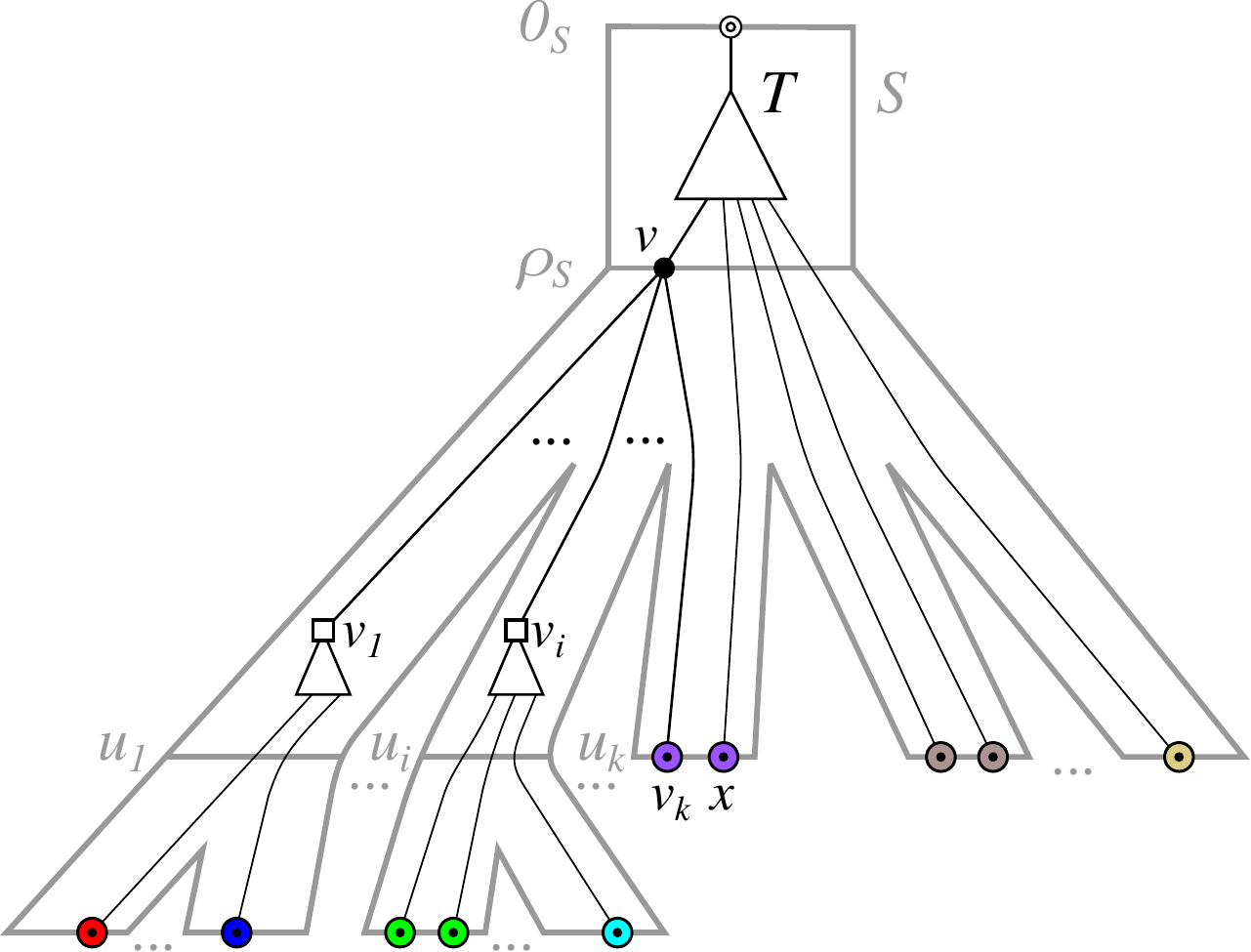}
      \end{center}
      \caption{Visualization of the construction of a species tree $S$ and
        reconciliation map $\mu$ as described in the proof of
        Lemma~\ref{lem:T-fp-no-mu}.  Note that, in the example, $v_k$ is
        already a leaf in the gene tree $T$. Hence, the corresponding $u_k$
        is also a leaf since $|\sigma(L(T(v_k)))|=1$.  Moreover, note that
        for $x\in L(T)\setminus L(T(v))$, it is possible that $\mu(x)=u_j$ or
        $\mu(x)=t$ with $t\in\child_{S}(u_j)$ for some $u_j$.}
      \label{fig:reconc_construction}
    \end{figure}
    
    The only vertex of $T$ that is mapped to a vertex in $S$ is $v$. Hence,
    it remains to show that $\mu(v)=\rho_S\in V^0(S)$ satisfies \AX{(R3)}.
    Note that for every two distinct children $v_i, v_j$ of $v$ we have
    $\mu(v_i)\in\{\rho_Su_i, u_i\}$ and $\mu(v_j)\in\{\rho_Su_j, u_j\}$.  In
    any case, $\mu(v_i)$ and $\mu(v_j)$ are incomparable in $S$.  Hence,
    (R3.ii) is satisfied. In particular,
    $\mu(v) = \rho_S = \lca_S(\mu(v_i),\mu(v_j))$ for all distinct
    $v_i,v_j\in \child(v)$.  Hence, (R3.i) is satisfied. In summary, $\mu$ is
    a reconciliation map from $(T,\sigma)$ to $S$. Since
    $\mu(v)=\rho_S\in V^0(S)$, we have $t_{\mu}(v)=\SPEC$.
    
    Statements (2) and (3) are equivalent by definition of the extremal event
    labeling.  
  \end{proof}
  Lemma~\ref{lem:T-fp-no-mu} implies that $(T,\sigma)$-\fp can be verified in
  polynomial time for any given gene tree $(T,\sigma)$.
  
  \begin{definition}[Unambiguous false-positive]  \label{def:ufp}
    Let $(\G,\sigma)$ be a BMG. An edge $xy$ in $\G$ is called
    \emph{unambiguous false-positive (\ufp)} if for all trees $(T,\sigma)$
    that explain $(\G,\sigma)$ the edge $xy$ is $(T,\sigma)$-\fp.
  \end{definition}
  Hence, if an edge $xy$ in $\G$ is \ufp, then it is in particular
  $(T,\sigma)$-\fp in the true history that explains $(\G,\sigma)$.  Thus, \ufp 
  edges are always ``correct'' false-positives.

  \subsection{The color-intersection $\Scap$}
  \label{APP:ssec:CI}
  
  Given a gene tree $(T,\sigma)$ and a pair of
  distinct leaves $x,y\in L(T)$, we denote by
  $v_x, v_y \in \child_T(\lca_T(x,y))$ the unique children of the last common
  ancestor of $x$ and $y$ for which $x\preceq_T v_x$ and $y\preceq_T
  v_y$. That is, $T(v_x)$ and $T(v_y)$ are the subtrees of $T$ rooted in the
  children of $\lca_T(x,y)$ with $x\in L(T(v_x))$ and $y\in L(T(v_y))$. The
  set
  \begin{equation}
  \mathcal{S}_T^{\cap}(x,y)\coloneqq\sigma(L(T(v_x)))\cap\sigma(L(T(v_y)))
  \end{equation} 
  contains the colors, i.e.\ species, that are common to both subtrees.
  Lemma~\ref{lem:edge-xy-lca} immediately implies
  \begin{corollary}
    \label{cor:sigma-xy-notin-Scap}
    Let $xy$ be an edge in a BMG $(\G,\sigma)$. Then
    $\sigma(\{x,y\})\cap \mathcal{S}_T^{\cap}(x,y)=\emptyset$ for all
    trees $(T,\sigma)$ that explain $(\G,\sigma)$.
  \end{corollary}
  
  The following result shows that the color-intersection of
  a given edge in a BMG $(\G,\sigma)$ in fact does not depend on the tree
  representation of $(\G,\sigma)$.
  \begin{lemma}
    Let $(\G,\sigma)$ be a BMG and $(T^*,\sigma )$ the corresponding unique
    least resolved tree explaining $(\G,\sigma)$.  Then, for each tree
    $(T,\sigma)$ that explains $(\G,\sigma)$, every edge $xy$ in
    $(\G,\sigma)$ satisfies $\Scap_{T^*}(x,y)=\Scap_T(x,y)$.  Thus, in
    particular, $\mathcal{S}_{T^*}^{\cap}(x,y)\neq\emptyset$ if and only if
    $\mathcal{S}_T^{\cap}(x,y) \neq \emptyset$.
    \label{lem:Scap}
  \end{lemma}
  \begin{proof}
    Let $(T,\sigma)$ be an arbitrary tree that explains $(\G,\sigma)$.
    Moreover, let $xy$ be an edge in $\G$ and denote by $v_x$ and $v_y$ be
    the unique children $v_x, v_y \in \child_T(\lca_T(x,y))$ with
    $x\preceq_T v_x$ and $y\preceq_T v_y$. Analogously, $v^*_x$ and $v^*_y$
    are the unique children $v^*_x, v^*_y \in \child_{T^*}(\lca_{T^*}(x,y))$
    with $x\preceq_{T^*} v^*_x$ and $y\preceq_{T^*} v^*_y$.
    
    First, we show that $t\in\Scap_{T^*}(x,y)$ implies
    $t\in\mathcal{S}_T^{\cap}(x,y)$.  Since $(T,\sigma)$ explains
    $(\G, \sigma)$, we apply Thm.~\ref{thm:LRT} to conclude that $T$ is a
    refinement of $T^*$ and thus, $\mathscr{C}(T^*)\subseteq \mathscr{C}(T)$.
    Therefore, $L(T^*(\lca_{T^*}(x,y))$, $L(T^*(v^*_x))$ and $L(T^*(v^*_y))$
    are contained in $\mathscr{C}(T)$.  This implies that there must be
    vertices $u$, $w_x$, and $w_y$ in $T$ with
    $L(T(u))= L(T^*(\lca_{T^*}(x,y))$, $L(T(w_x))=L(T^*(v^*_x))$ and
    $L(T(w_y))=L(T^*(v^*_y))$.  Note that
    $L(T^*(v^*_x))\cap L(T^*(v^*_y))=\emptyset$, and thus
    $L(T(w_x))\cap L(T(w_y))=\emptyset$.  In particular, $w_x$ and $w_y$ are
    incomparable in $T$. Moreover, $u = \lca_T(x,y)=\lca_T(w_x,w_y)$, thus we
    have $w_x\preceq_T v_x$ and $w_y\preceq_T v_y$.  Therefore,
    $L(T^*(v^*_x))\subseteq L(T(v_x))$ and
    $L(T^*(v^*_y))\subseteq L(T(v_y))$.  Therefore,
    $t\in \mathcal{S}_{T^*}^{\cap}(x,y)$ implies
    $t\in \mathcal{S}_T^{\cap}(x,y)$.
    
    Now, we show that $t\in\mathcal{S}_T^{\cap}(x,y)$ implies
    $t\in\Scap_{T^*}(x,y)$.  Let $t\in\Scap_T(x,y)\neq\emptyset$.  In this
    case, $t\in \sigma(L(T(v_x)))$ and we can choose a vertex
    $z_1\in L(T(v_x))$ such that $\sigma(z_1)=t$ and $\lca_T(x,z_1)$ is as
    far away as possible from $v_x$ compared to all $\lca_T(x,z)$ with
    $z\in L[t]$, i.e., $\lca_T(x,z_1) \preceq_T \lca_T(x,z)$ for all
    $z\in L[t]$. Thus, $(x,z_1)\in E(\G)$. An analogous argument ensures that
    there is a vertex $z_2\in L(T(v_y))$ such that $\sigma(z_2)=t$ and
    $(y,z_2)\in E(\G)$.  Clearly,
    $\lca_{T}(x,z_2)=\lca_{T}(x,y)=\lca_{T}(y,z_1)$ and thus
    $\lca_{T}(x,z_1)\preceq_{T}v_x\prec_T\lca_{T}(x,z_2)$, which in turn
    implies that $(x,z_2)\notin E(\G)$.  Since $(x,z_1)\in E(\G)$ and
    $(x,z_2)\notin E(\G)$, we obtain the informative triple $xz_1|z_2$ for
    $(\G,\sigma)$.  Analogously, $yz_2|z_1$ is an informative triple for 
    $(\G,\sigma)$.
    Lemma~\ref{lem:inf_triples_overlap} and the fact that $T^*$ explains
    $(\G,\sigma)$ implies that there are distinct vertices
    $v_1,v_2\in\child_{T^*}(\lca_{T^*}(x,y))$ such that
    $x,z_1\preceq_{T^*}v_1$ and $y,z_2\preceq_{T^*}v_2$.  Since
    $t=\sigma(z_1)=\sigma(z_2)$, we have $t\in\Scap_{T^*}(x,y)$.
    
    Finally, $t\in\Scap_{T^*}(x,y)$ if and only if $t\in\Scap_{T}(x,y)$
    implies both $\Scap_{T^*}(x,y)=\Scap_{T}(x,y)$ and
    $\Scap_{T^*}(x,y)\neq\emptyset$ \emph{if and only if}
    $\Scap_{T}(x,y)\neq\emptyset$.
  \end{proof}
  
  \begin{remark}\label{rem:Scap}
    By Lemma~\ref{lem:Scap}, we have $\Scap_{T}(x,y)=\Scap_{T^*}(x,y)$ for
    every tree $(T,\sigma)$ explaining a BMG $(\G,\sigma)$ with corresponding
    least resolved tree $(T^*,\sigma)$.  Therefore, it is sufficient to
    consider $\Scap_{T^*}(x,y)$. We will therefore drop the explicit
    reference to the tree and simply write $\Scap(x,y)$. We can verify in
    polynomial time whether or not $\Scap(x,y)=\emptyset$ because the least
    resolved tree $(T^*,\sigma)$ explaining $(\G,\sigma)$ can be computed in
    polynomial time.
  \end{remark}
  
  \begin{proposition}\label{prop:color_intersection_dupl}
    Every edge $xy$ in a BMG $(\G,\sigma)$ with $\Scap(x,y)\ne\emptyset$
    is \ufp.
  \end{proposition}
  \begin{proof}
    By Lemma~\ref{lem:Scap} and Remark~\ref{rem:Scap},
    $\Scap(x,y)\ne\emptyset$ if and only if $\Scap_T(x,y)\ne\emptyset$ for
    all trees $(T,\sigma)$ that explain $(\G,\sigma)$.  By 
    Lemma~\ref{lem:duplication_witness}, $\mu(\lca_T(x,y))\in E(S)$ and thus,
    $t_{\mu}(\lca_T(x,y))=\DUPL$ for all trees $(T,\sigma)$ that explain
    $(\G,\sigma)$.  Hence, $xy$ is \ufp.
  \end{proof}
  As we shall see later, the converse of
  Prop.~\ref{prop:color_intersection_dupl} is not always satisfied (cf.\ also
  Fig.~\ref{fig:Scap_empty_ugly_quartet}).
  \noindent An immediate consequence of
  Prop.~\ref{prop:color_intersection_dupl} is:
  \begin{corollary}\label{cor:Tfp-Scap-noneqi}
    An edge $xy$ in a BMG $\G(T,\sigma)$ with $\Scap(x,y)\ne\emptyset$ is 
    $(T,\sigma)$-\fp.
  \end{corollary}
  Although not necessarily true in general, we show next that the converse of
  Prop.~\ref{prop:color_intersection_dupl} and Cor.~\ref{cor:Tfp-Scap-noneqi}
  does hold for the special case of binary trees.
  
  \begin{lemma}
    \label{lem:bin-v-Scap}
    Let $xy$ be an edge in $\G(T,\sigma)$ and suppose $\lca_T(x,y)$ is
    a binary vertex.  Then, the following three statements are equivalent:
    \begin{enumerate}[itemsep=0.2ex, topsep=0.2ex, parsep=0cm]
      \item The edge $xy$ is $(T,\sigma)$-\fp.
      \item $\Scap(x,y)\neq \emptyset$.
      \item The edge $xy$ is \ufp.  
    \end{enumerate}
  \end{lemma}
  \begin{proof}
    \noindent\textit{(1) implies (2).}
    Suppose $xy$ is $(T,\sigma)$-\fp. Since $v$ is binary, it has precisely
    two children $v_1$ and $v_2$. In particular, $v=\lca_T(x,y)$ implies that
    that $x\preceq_T v_i$ and $x\preceq_T v_j$ for $i,j\in \{1,2\}$ being
    distinct.  By Lemma~\ref{lem:T-fp-no-mu}, the two children $v_1$ and
    $v_2$ of $v$ satisfy
    $\sigma(L(T(v_1)))\cap \sigma(L(T(v_2)))\neq\emptyset$.  By 
    Lemma~\ref{lem:Scap} and Remark~\ref{lem:Scap}, we have
    $\Scap(x,y)\ne\emptyset$.
    \par\noindent\textit{(2) implies (3).} 
    If $\Scap(x,y)\ne\emptyset$, we can apply
    Prop.~\ref{prop:color_intersection_dupl} to conclude that $xy$ is \ufp.
    \par\noindent\textit{(3) implies (1).}
    By definition, if $xy$ is \ufp, then it is in particular also
    $(T,\sigma)$-\fp.
  \end{proof}
  
  \begin{theorem}
    Let $(\G,\sigma)$ be a BMG that is explained by a binary tree
    $(T,\sigma)$. Then, for every edge $xy$ in $(\G,\sigma)$, the following
    three statements are equivalent:
    \begin{enumerate}[itemsep=0.2ex, topsep=0.2ex, parsep=0cm]
      \item The edge $xy$ is $(T,\sigma)$-\fp.
      \item $\Scap(x,y)\ne\emptyset$.
      \item The edge $xy$ is \ufp.
    \end{enumerate}
    \label{thm:ufp-binary}
  \end{theorem}
  \begin{proof}
    For every edge $xy$ in $\G$ the last common ancestor $\lca_T(x,y)$ is
    binary. Now apply Lemma~\ref{lem:bin-v-Scap}.
  \end{proof}
  
  Thm.~\ref{thm:ufp-binary} implies that all \ufp edges can be detected in
  a BMG that is explained by a known binary gene tree. However, not all BMGs
  $(\G,\sigma)$ can be explained by a binary tree, as e.g.\ the BMG in
  Fig.~\ref{fig:hourglasses}(A). Thm.~\ref{thm:ufp-binary} does not
  generalize to the non-binary case, and $\Scap(x,y)$ is not sufficient to
  identify all \ufp edges. Furthermore, it is not difficult to find
  non-binary trees in which $(T,\sigma)$-\fp and \ufp edges are not the same:
  As show in Fig.~\ref{fig:non_binary_tree_example}, the edge $xz$ in is
  $(T_1,\sigma)$-\fp but not $(T_2,\sigma)$-\fp according to Lemma
  \ref{lem:T-fp-no-mu}. Since both trees explain the same BMG, the edge $xy$
  is not \ufp.
  
  \subsection{$\Scap(x,y)\ne\emptyset$: quartets}
  \label{APP:ssec:quartets}
  
  Since every orthology graph is a cograph (cf.\
  Thm.~\ref{thm:ortho-cograph}), we know that every induced $P_4$ in the RBMG
  is associated with false-positive edges. The induced subgraphs of the BMG
  spanned by a $P_4$ in its symmetric part (i.e., the RBMG) are called
  quartets.  We write $\langle abcd \rangle$ or, equivalently,
  $\langle dcba \rangle$ for an induced $P_4$ with edges $ab$, $bc$, and
  $cd$. The quartets on three colors fall into three classes:
  \begin{definition}[Good, bad, and ugly quartets]
    Let $(\G,\sigma)$ be a BMG with symmetric part $(G,\sigma)$ and vertex
    set $L$, and let $Q\coloneqq \{x,y,z,z'\} \subseteq L$ with $x\in L[r]$,
    $y\in L[s]$, and $z,z'\in L[t]$. The set $Q$, resp., the induced subgraph
    $(\G[Q],\sigma_{|Q})$ is
    \begin{itemize}[itemsep=1.2ex, topsep=0.2ex, parsep=0cm]
      \item[] a \emph{good quartet} if (i) $\langle zxyz'\rangle$ is an induced
      $P_4$ in $(G,\sigma)$ and (ii) $(z,y),(z',x)\in E(\G)$ and
      $(y,z),(x,z')\notin E(\G)$,
      \item[] a \emph{bad quartet} if (i) $\langle zxyz'\rangle$ is an induced
      $P_4$ in $(G,\sigma)$ and (ii) $(y,z),(x,z')\in E(\G)$ and
      $(z,y),(z',x)\notin E(\G)$,
      \item[] an \emph{ugly quartet} if $\langle zxz'y\rangle$ is an induced
      $P_4$ in $(G,\sigma)$.
    \end{itemize}
    \noindent The edge $xy$ in a good quartet $\langle zxyz'\rangle$ is its
    \emph{middle} edge. The edge $zx$ of an ugly quartet
    $\langle zxz'y\rangle$ or a bad quartet $\langle zxyz'\rangle$ is called
    its \emph{first} edge.  First edges in ugly quartets are uniquely
    determined due to the colors.  In bad quartets, this is not the case and
    therefore, the edge $yz'$ in $\langle zxyz'\rangle$ is a first edge as
    well.
    \label{def:GoodBadUgly}
  \end{definition}
  An RBMG never contains induced $P_4$s on two colors
  \cite[Obs.~5]{Geiss:19b}.  This, in particular, implies that for the
  induced $P_4$s in Def.~\ref{def:GoodBadUgly} the colors $r$, $s$, and $t$
  must be pairwise distinct. Induced $P_4$s on four colors are investigated
  in some more detail in Sec.~\ref{APP:ssec-4colP4} below.
  
  The key property of good quartets is a consequence of
  \cite[Cor.~5]{Geiss:20a}:
  \begin{proposition}
    \label{prop:good_quartet_middle_edge}
    If $\langle zxyz'\rangle$ is a good quartet in the BMG $(\G,\sigma)$,
    then $\Scap(x,y)\neq\emptyset$ and thus, $xy$ is \ufp.
  \end{proposition}
  \begin{proof}
    Let $\langle zxyz' \rangle$ in $(\G,\sigma)$ be a good quartet in
    $(\G,\sigma)$ and let $(T,\sigma)$ be an arbitrary tree explaining
    $(\G,\sigma)$. Then \cite[Lemma~36]{Geiss:19b} implies that
    $v\coloneqq \lca_T(x,y,z,z')$ has two distinct children
    $v_1, v_2\in \child(v)$ such that $x,z \preceq_T v_1$ and
    $y,z'\preceq_T v_2$. Hence, $v=\lca_T(x,y)$.  Since
    $\sigma(z)\in \sigma(L(T(v_1)))\cap \sigma(L(T(v_2)))$, we have
    $\Scap(x,y)\neq\emptyset$ and, by Prop.~\ref{prop:color_intersection_dupl},
    the edge $xy$ is \ufp.  
  \end{proof}
  Prop.~\ref{prop:good_quartet_middle_edge} provides a convenient
  way to identify unambiguous false-positive edges in a BMG.
  
  \begin{lemma}
    \label{lem:good_quartet_existence-2}
    If $xy$ is an edge in a BMG $\G(T,\sigma)$ and $t\in\Scap(x,y)$, then
    there is a good quartet $\langle z_1x^*y^*z_2\rangle$ such that
    \begin{description}[itemsep=0.2ex, topsep=0.2ex, parsep=0cm]
      \item[(a)] $\sigma(x^*)=\sigma(x)$, $\sigma(y^*)=\sigma(y)$, and
      $\sigma(z_1)=\sigma(z_2)=t$;
      \item[(b)] $x^*,z_1\in L(T(v_x))$ and $y^*,z_2\in L(T(v_y))$ with $v_x$
      and $v_y$ being the unique children in $\child_T (\lca_T (x, y))$ such
      that with $x \preceq_T v_x$ and $y \preceq_T v_y$.
    \end{description}
  \end{lemma}
  \begin{proof}
    Consider an edge $xy$ of $\G(T,\sigma)$ and a color $t\in\Scap(x,y)$.  By
    Cor.~\ref{cor:sigma-xy-notin-Scap}, $t\neq \sigma(x),\sigma(y)$.
    Lemma~\ref{lem:exEdge} ensures the existence of an edge $x^*z_1$ in $\G$
    for some leaves $x^*\in L(T(v_x))\cap L[\sigma(x)]$ and
    $z_1\in L(T(v_x))\cap L[t]$.  By the same arguments as in the proof of
    Cor.~\ref{cor:sigma-xy-notin-Scap}, we can conclude that $z_1y'$ is
    not an edge in $\G$ for all $y'\in L(T(v_y))\cap L[\sigma(y)]$.  However,
    $(z_1,y')\in E(\G)$ since the color of $y'$ is not present in
    $T(v_x)$. Likewise, there are leaves $y^*\in L(T(v_y))\cap L[\sigma(y)]$
    and $z_2\in L(T(v_y))\cap L[t]$ such that $y^*z_2$ forms an edge in
    $\G$. Reusing the arguments from $L(T(v_x))$, we find that $x'z_2$ is not
    an edge in $\G$ and $(z_2,x')\in E(\G)$ for any
    $x'\in L(T(v_x))\cap L[\sigma(x)]$.  Finally,
    $\sigma(x)\notin\sigma(L(T(v_y)))$ and $\sigma(y)\notin\sigma(L(T(v_x)))$
    implies that $x^*y^*$ forms an edge in $\G$.  Hence,
    $\langle z_1x^*y^*z_2\rangle$ is a good quartet.  
  \end{proof}
  The edge $x^*y^*$ in Lemma~\ref{lem:good_quartet_existence-2} is the middle
  edge of a good quartet.  For completeness, we also provide a result for the
  identification of \ufp edges using bad quartets:
  \begin{proposition}
    Let $\langle zxyz'\rangle$ be a bad quartet in a BMG $(\G,\sigma)$.
    Then, the edges $xz$ and $yz'$ are \ufp and every tree that explains
    $(\G,\sigma)$ is non-binary.
    \label{prop:bad}
  \end{proposition}
  \begin{proof}
    Let $(T,\sigma)$ be an arbitrary tree that explains $(\G,\sigma)$, set
    $u\coloneqq\lca_T(x,z)$ and let $v_x,v_z\in\child_T(u)$ be the two
    distinct children of $u$ such that $x\preceq_T v_x$ and $z\preceq_T
    v_z$. By symmetry, it suffices to show that $xz$ is \ufp.  Since
    $\langle zxyz'\rangle$ is a bad quartet, we have $(x,z),(x,z')\in E(\G)$
    and thus $\lca_T(x,z')=\lca_T(x,z)=u$.  Let $v_{z'}\in\child_T(u)$ be the
    child of $u$ such that $z'\preceq_T v_{z'}$. Since $\lca_T(x,z')=u$ we
    have $v_x\ne v_{z'}$.  Now, assume for contradiction that $v_z=v_{z'}$,
    and thus $z'\in L(T(v_z))$.  Since $\langle zxyz'\rangle$ is a bad
    quartet, we have $(z',x)\notin E(\G)$, which implies the existence of a
    vertex $x'$ with $\sigma(x)=\sigma(x')$ and
    $\lca_T(x', z')\prec_T\lca_T(x,z')=u$ and therefore, $x'\in L(T(v_z))$.
    However, this implies that
    $\lca_T(x',z)\preceq_{T}v_z\prec_T u=\lca_T(x,z)$, which together with
    $\sigma(x)=\sigma(x')$ contradicts the fact that $xz$ is an edge in $\G$.
    Hence, $v_z\ne v_{z'}$.  Therefore,
    $\sigma(z)=
    \sigma(z')\in\sigma(L(T(v_z)))\cap\sigma(L(T(v_{z'})))\ne\emptyset$ for
    distinct children $v_z,v_{z'}\in\child_T(u)$.  By
    Lemma~\ref{lem:T-fp-no-mu}, the edge $xz$ is $(T,\sigma)$-\fp and since
    $(T,\sigma)$ was chosen arbitrarily, the edge $xz$ is \ufp.  Moreover, we
    have shown that $v_x$, $v_z$ and $v_{z'}$ must be pairwise distinct and
    thus, $(T,\sigma)$ is non-binary.
  \end{proof}
  
  Fig.~\ref{fig:ugly_quartet} shows that \ufp edges $xy$ with
  $\Scap(x,y)\ne\emptyset$ exist that are neither middle edges of good
  quartets or first edges of bad quartets. Thus we next consider ugly
  quartets.
  \begin{proposition}
    \label{prop:ugly_quartet}
    If $\langle xyx'z\rangle$ is an ugly quartet in a BMG $(\G,\sigma)$, then
    the edges $xy$ and $yx'$ are \ufp.
  \end{proposition}
  \begin{proof}
    Consider an ugly quartet $\langle xyx'z\rangle$. Let $(T,\sigma)$ be an
    arbitrary tree explaining $(\G,\sigma)$, put $u\coloneqq\lca_T(x,y)$ and
    let $v_x,v_y\in\child_T(u)$ be the two distinct children of $u$ such that
    $x\preceq_T v_x$ and $y\preceq_T v_y$.
    
    Since $x'y$ and $xy$ are edges in $\G$ we have $\lca_T(x',y)\preceq_T u$.
    Moreover, Cor.~\ref{cor:sigma-xy-notin-Scap} implies
    $\sigma(x')=\sigma(x)\notin \sigma(L(T(v_y)))$ and thus
    $x'\notin L(T(v_y))$.  Therefore, $\lca_T(x',y)=\lca_T(x,y)=u$.
    
    Now consider an arbitrary reconciliation map $\mu$ from $(T,\sigma)$ to
    some species tree $S$. The existence of $\mu$ is guaranteed by 
    Lemma~\ref{lem:reconAll}.  If $x'\notin L(T(v_x))$, then there is a vertex
    $v_3\in\child_T(u)$, $v_3\ne v_x,v_y$ such that $x'\preceq_T v_3$ and
    $\sigma(x)=\sigma(x')\in\sigma(L(T(v_x)))\cap\sigma(L(T(v_3)))
    \ne\emptyset$, which by Lemma~\ref{lem:duplication_witness} implies
    $t_\mu(u)=\DUPL$.
    
    Now suppose $x'\in L(T(v_x))$ and recall that $x'z$ is an edge in $\G$ by
    assumption. Since $\lca_T(x',z)$ and $\lca_T(x,x')$ are both ancestors of
    $x'$ they are comparable. If $\lca_T(x',z)\succ_T \lca_T(x,x')$, then
    $\lca_T(x,z)=\lca_T(x',z)$.  Together with the fact that $x'z$ is an
    edge in $\G$ but not $xz$, this implies that there is a
    $z'\in L[\sigma(z)]$ such that $\lca_T(x,z')\prec_T\lca_T(x,z)$. This in
    turn implies $\lca_T(x',z')\prec_T\lca_T(x',z)$, which contradicts that
    $x'z$ is an edge in $\G$. Therefore, $x'\in L(T(v_x))$ implies
    $\lca_T(x',z)\preceq_T \lca_T(x,x')$ and $x,x',z\in L(T(v_x))$. Since
    $yz$ is not an edge in $\G$ by assumption and
    Cor.~\ref{cor:sigma-xy-notin-Scap} implies
    $\sigma(y)\notin \sigma(L(T(v_x))$, there is a leaf $z'$ with color
    $\sigma(z')=\sigma(z)$ such that $\lca_T(y,z')\prec_T \lca_T(y,z)$. This
    is only possible if $z'\in L(T(v_y))\cap L[\sigma(z)]$.  Therefore,
    $\sigma(z)\in\sigma(L(T(v_x)))\cap\sigma(L(T(v_y)))$ and
    Lemma~\ref{lem:duplication_witness} implies that $t_\mu(u)=\DUPL$.
    
    In summary, $\lca_T(x',y)=\lca_T(x,y)=u$ and $t_\mu(u)=\DUPL$ for
    every tree explaining $(\G,\sigma)$ and every possible reconciliation
    map $\mu$ from $(T,\sigma)$ to any species tree. Thus both $xy$
    and $x'y$ are \ufp.
  \end{proof}
  
  \begin{proposition}
    \label{prop:good_or_ugly}
    Let $(\G,\sigma)$ be a BMG and $xy$ an edge in $\G$ with
    $\Scap(x,y)\ne\emptyset$. Then $xy$ is either the middle edge of some
    good quartet $\langle zxyz'\rangle$ or the first edge in some ugly
    quartet $\langle xyx'z\rangle$ or $\langle yxy'z\rangle$.
  \end{proposition}
  \begin{proof}    
    Let $(T,\sigma)$ be a leaf-colored tree explaining the BMG $(\G,\sigma)$
    with symmetric part $(G,\sigma)$. Let
    $v_x, v_y \in \child_T(\lca_T(x,y))$ such that $x\preceq_T v_x$ and
    $y\preceq_T v_y$. Since $\Scap(x,y)\ne\emptyset$,
    Lemma~\ref{lem:good_quartet_existence-2} implies that there is a good
    quartet $\langle z_1x^*y^*z_2\rangle$ with $\sigma(x^*)=\sigma(x)$,
    $\sigma(y^*)=\sigma(y)$, $\sigma(z_1)=\sigma(z_2)=t\in\Scap(x,y)$,
    $x^*,z_1\in L(T(v_x))$ and $y^*,z_2\in L(T(v_y))$.
    
    If $x=x^*$ and $y=y^*$ we are done. By symmetry it suffices to
    consider the case $x\ne x^*$. Before we proceed, we consider the
    (non-)existence of certain edges in the RBMG $G(T,\sigma)$ and the BMG
    $\G(T,\sigma)$. By definition of good quartets, we have
    $x^*z_1,x^*y^*,y^*z_2\in E(G)$ and
    Cor.~\ref{cor:sigma-xy-notin-Scap} implies
    $\sigma(x),\sigma(y)\notin\Scap(x,y)$. Hence,
    $\sigma(x^*)=\sigma(x)\notin \sigma(L(T(v_y)))$ and
    $\sigma(y^*)=\sigma(y)\notin \sigma(L(T(v_x)))$, and thus $x^*y\in E(G)$
    and $xy^*\in E(G)$. Moreover, since $\lca_T(y,z_2)\prec_T\lca_T(y,z_1)$,
    we have $yz_1\notin E(G)$.  Similarly, $xz_2\notin E(G)$.  However,
    $\sigma(x)\notin\sigma(L(T(v_y)))$ implies that
    $\lca_T(z_2,x) = \lca_T(x,y)\preceq\lca_T(z_2,x')$ for all
    $x'\in L[\sigma(x)]$ and thus, $(z_2,x)\in E(\G)$. Similarly,
    $(z_1,y)\in E(\G)$. Furthermore, we note that neither $x$ and $x^*$ nor
    $y$ and $y^*$ can be adjacent in $G$ or $\G$ since
    $\sigma(x)=\sigma(x^*)$ and $\sigma(y)=\sigma(y^*)$.
    
    If $xz_1 \notin E(G)$, then $\langle xyx^*z_1\rangle$ forms an ugly
    quartet. Now suppose that $xz_1 \in E(G)$.  Assume that there is an edge
    $yz'\in E(G)$ with $z'\in L(T(v_y))\cap L[t]$. Then,
    $\lca(x,z_1)\prec_T\lca(x,z')$ implies $xz'\notin E(G)$. Moreover, since
    $\sigma(x)\notin\sigma(L(T(v_y)))$ we have, by similar arguments as
    above, that $(z',x)\in E(\G)$. Thus, $\langle z'yxz_1\rangle$ forms a
    good quartet. Finally, if there is no such edge $yz'\in E(G)$ then, in
    particular, $yz_2\notin E(G)$ and $y\neq y^*$. In this case,
    $\langle yxy^*z_2\rangle$ forms an ugly quartet. 
  \end{proof}
  The example Fig.~\ref{fig:Scap_empty_ugly_quartet} shows that the converse
  of Prop.~\ref{prop:good_or_ugly} is not true in general.
  
  \begin{figure}[t]
    \begin{center}
      \includegraphics[width=0.6\textwidth]{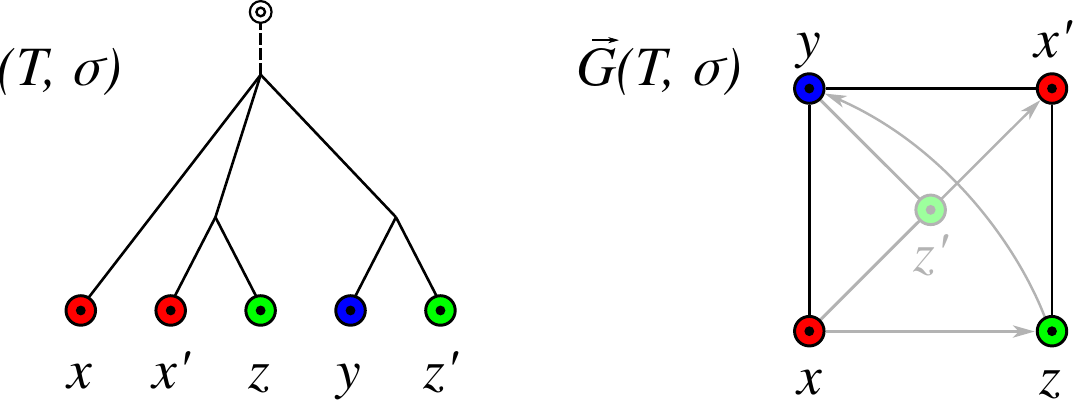}
    \end{center}
    \caption{The edge $xy$ is \ufp since it is the first edge of an ugly
      quartet. However, $\Scap(x,y)=\emptyset$ and thus, the converse of
      Prop.~\ref{prop:good_or_ugly} is not satisfied.}
    \label{fig:Scap_empty_ugly_quartet}
  \end{figure}
  
  \noindent We summarize the results of
  Props.~\ref{prop:color_intersection_dupl},
  \ref{prop:good_quartet_middle_edge}, \ref{prop:ugly_quartet} and
  \ref{prop:good_or_ugly} in the following
  \begin{corollary}\label{cor:ufp-quartets}
    Let $(\G,\sigma)$ be a BMG that contains the edge $xy$. Then,
    $\Scap(x,y)\ne\emptyset$ implies that $xy$ is either the middle edge of
    some good quartet or the first edge of some ugly quartet, which in turn
    implies that $xy$ is \ufp.
  \end{corollary}
  
  \subsection{$\Scap(x,y)=\emptyset$: hourglasses}
  \label{APP:ssect:hourglass}
  
  The case $\Scap(x,y)\neq\emptyset$ is sufficient to detect the edge $xy$ as
  \ufp. In this section we turn to the case $\Scap(x,y)=\emptyset$ and show
  how to identify further \ufp edges.
  \begin{definition}[Hourglass] \label{def:hourglass}
    An \emph{hourglass} in a proper vertex-colored graph $(\G,\sigma)$,
    denoted by $[xy \hourglass x'y']$, is a subgraph $(\G[Q],\sigma_{|Q})$
    induced by a set of four pairwise distinct vertices
    $Q=\{x, x', y, y'\}\subseteq V(\G)$ such that (i)
    $\sigma(x)=\sigma(x')\ne\sigma(y)=\sigma(y')$, (ii) $xy$ and $x'y'$ are
    edges in $\G$, (iii) $(x,y'),(y,x')\in E(\G)$, and (iv)
    $(y',x),(x',y)\notin E(\G)$.
  \end{definition}
  Note that Condition (i) rules out arcs between $x,x'$ and $y,y'$,
  respectively, i.e., the only arcs in an hourglass are the ones specified by
  Conditions (ii) and (iii).  An example is shown in
  Fig.~\ref{fig:hourglasses}(A).
  \begin{fact}\label{obs:hourbmg}
    Every hourglass is a BMG since it can be explained by a tree as shown in
    Fig.~\ref{fig:hourglasses}(B).
  \end{fact}
  
  We first show that hourglasses cannot appear in a BMG that can be explained
  by a binary tree.
  \begin{lemma}
    \label{lem:hourglass}
    If $(\G,\sigma)$ is a BMG containing the hourglass
    $[xy \hourglass x'y']$, then every tree $(T,\sigma)$ that explains
    $(\G,\sigma)$ contains a vertex $u\in V^0(T)$ with three distinct
    children $v_1$, $v_2$, and $v_3$ such that $x\preceq_T v_1$,
    $\lca_T(x',y')\preceq_T v_2$ and $y\preceq_T v_3$.
  \end{lemma}
  \begin{proof}
    By assumption, $xy$ and $x'y'$ are edges in $\G$,
    $(x,y'),(y,x')\in E(\G)$, and $(y',x),(x',y)\notin E(\G)$.  By 
    Lemma~\ref{lem:informative_triples}, the informative triples $x'y'|x$ and
    $x'y'|y$ thus must be displayed by every tree $(T,\sigma)$ that explains
    $(\G,\sigma)$. Thus
    $u_{x'y'}\coloneqq\lca_T(x',y') \prec_T u_x\coloneqq\lca_T(x,u_{x'y'})$
    and $u_{x'y'} \prec_T u_y\coloneqq\lca_T(y,u_{x'y'})$.  Furthermore,
    $u_x$ and $u_y$ are both ancestors of $u_{x'y'}$ and thus comparable
    w.r.t. $\preceq_T$.  If $u_x\prec_T u_y$, then
    $\lca_T(x,y')\prec_T\lca_T(x,y)$ which implies that $xy$ cannot form an
    edge in $\G$; a contradiction.  By similar arguments, $u_y\prec_T u_x$ is
    not possible and therefore, $u_x=u_y\eqqcolon u$.
    
    Since $u_{x'y'}\prec_T u$, there are two distinct children
    $v_1,v_2\in\child_T(u)$ of $u$ such that $x\preceq_T v_1$ and
    $u_{x'y'}\preceq_T v_2$.  Clearly, $y\notin L(T(v_2))$ since
    $\lca_T(y,u_{x'y'})=u\succ_T v_2$.  We also have $y\notin L(T(v_1))$
    since $y\in L(T(v_1))$ would imply
    $\lca_T(x,y)\preceq_T v_1\prec_T u=\lca_T(x,u_{x'y'})=\lca_T(x,y')$,
    contradicting $(x,y')\in E(\G)$.  Together with $y\in L(T(u))$, this
    implies the existence of a vertex $v_3\in\child(u)$ such that
    $v_3\notin\{v_1,v_2\}$ and $y\preceq_T v_3$.
  \end{proof}
  The result shows that hourglasses $[xy \hourglass x'y']$ can be used to
  identify false-positive edges $xy$ with $\Scap(x,y)=\emptyset$.
  \begin{proposition}\label{prop:singeHG-ufp}
    If a BMG $(\G,\sigma)$ contains an hourglass $[xy \hourglass x'y']$, then
    the edge $xy$ is \ufp.
  \end{proposition}
  \begin{proof}
    According to Lemma~\ref{lem:hourglass}, every tree $(T,\sigma)$ that
    explains $(\G,\sigma)$ contains a vertex $u\in V^0(T)$ with three
    distinct children $v_1$, $v_2$, and $v_3$ such that $x\preceq_T v_1$,
    $\lca_T(x',y')\preceq_T v_2$ and $y\preceq_T v_3$.  Thus, $u=\lca_T(x,y)$
    and $\sigma(x)\in \sigma(L(T(v_1)))\cap \sigma(L(T(v_2)))$.  Hence, we
    can apply Lemma~\ref{lem:T-fp-no-mu} to conclude that $xy$ is
    $(T,\sigma)$-\fp for every tree that explains $(\G,\sigma)$.  Therefore,
    the edge $xy$ is \ufp.
  \end{proof}
  Prop.~\ref{prop:singeHG-ufp} implies that there are \ufp edges that are not
  contained in a quartet, since an hourglass (see
  Fig.~\ref{fig:hourglasses}(A)) does not contain a $P_4$.  We next
  generalize the concept of hourglasses.
  \begin{definition}[Hourglass chain]
    \label{def:hc}
    An \emph{hourglass chain} $\mathfrak{H}$ in a graph $(\G,\sigma)$ is a
    sequence of $k\ge 1$ hourglasses
    $[x_1 y_1 \hourglass x'_1 y'_1],\dots,[x_k y_k \hourglass x'_k y'_k]$
    such that the following two conditions are satisfied for all
    $i\in\{1,\dots,k-1\}$:
    \begin{description}[itemsep=0.2ex, topsep=0.2ex, parsep=0cm]
      \item[\emph{(H1)}] $y_i=x'_{i+1}$ and $y'_i=x_{i+1}$, and
      \item[\emph{(H2)}] $x_i y'_j$ is an edge in $\G$ for all
      $j\in\{i+1,\dots,k\}$
    \end{description}
    A vertex $z$ is called a \emph{left} (resp., \emph{right}) \emph{tail} of
    the hourglass chain $\mathfrak{H}$ if it holds that $(z,x_1)\in E(\G)$
    and $(z,x'_1)\notin E(\G)$ (resp., $(z,y_k)\in E(\G)$ and
    $(z,y'_k)\notin E(\G)$).  We call $\mathfrak{H}$ \emph{tailed} if it has
    a left or right tail.
  \end{definition}
  
  Note that in contrast to good and bad quartets as well as individual
  hourglasses, an hourglass chain in $(\G,\sigma)$ is not necessarily an
  induced subgraph.
  \begin{fact}\label{fact:subchain}
    If
    $\mathfrak{H}=[x_1 y_1 \hourglass x'_1 y'_1], \dots, [x_k y_k \hourglass
    x'_k y'_k]$ be an hourglass chain in $(\G,\sigma)$, then
    $[x_i y_i \hourglass x'_i y'_i],\dots,[x_j y_j \hourglass x'_j y'_j]$ is
    an hourglass chain in $(\G,\sigma)$ for every $1\leq i < j \leq k$.
  \end{fact}
  Hourglass chains are composed of ``overlapping'' hourglasses.  The
  additional condition that $x_i y'_j\in E(G)$ for all $1\le i<j\le k$
  ensures that the two pairs $x'_k,y'_k$ and $x'_l,y'_l$ with $k\ne l$ cannot
  lie in the same subtree below the last common ancestor $u$ which is common
  to all hourglasses in the chain.
  \begin{lemma}
    \label{lem:hourglass_chain}
    Let
    $\mathfrak{H}=[x_1 y_1 \hourglass x'_1 y'_1],\dots,[x_k y_k \hourglass
    x'_k y'_k]$ be an hourglass chain in a BMG $(\G,\sigma)$.  Then, for
    every tree $(T,\sigma)$ that explains $(\G,\sigma)$ there is a vertex
    $u\in V^0(T)$ with pairwise distinct children $v_0,v_1,\dots,v_k,v_{k+1}$
    such that $x_1\in L(T(v_0))$, $y_k\in L(T(v_{k+1}))$, and, for
    all $1\le i\le k$, we have $x'_i,y'_i\in L(T(v_i))$.
  \end{lemma}
  \begin{proof}
    We prove the statement by induction on $k$.  For the base case $k=1$,
    observe that the hourglass $[x_1 y_1 \hourglass x'_1 y'_1]$ together with
    Lemma~\ref{lem:hourglass} implies that there is a vertex $u\in V^0(T)$
    with pairwise distinct children $v_0,v_1$ and $v_2$ such that
    $x_1\preceq_T v_0$, $\lca_T(x'_1,y'_1)\preceq_T v_1$ (thus
    $x'_1,y'_1\preceq_T v_1$) and $y_1\preceq_T v_2$.
    
    Now let $k>1$ and assume that the statement is true for all hourglass
    chains containing less than $k$ hourglasses.  Let
    $\mathfrak{H}=[x_1 y_1 \hourglass x'_1 y'_1],\dots,[x_k y_k \hourglass
    x'_k y'_k]$ be an hourglass chain.  By induction hypothesis, for every
    subsequence
    $\mathfrak{H}_{i|}\coloneqq [x_1 y_1 \hourglass x'_1 y'_1],\dots,[x_i y_i
    \hourglass x'_i y'_i]$ of $\mathfrak{H}$ with $1\le i< k$, which by
    Obs.~\ref{fact:subchain} is again an hourglass chain,
    the statement is true.
    
    Consider the subsequence $\mathfrak{H}_{i|}$ with $i=k-1$. By assumption,
    there is a vertex $u\in V^0(T)$ with pairwise distinct children
    $v_0,v_1,\dots,v_i,v_{i+1}$ such that it holds $x_1\in L(T(v_0))$,
    $y_i\in L(T(v_{i+1}))$, and, for all $1\le j\le i$, we have
    $x'_j,y'_j\in L(T(v_j))$. The hourglass
    $[x_{i+1} y_{i+1} \hourglass x'_{i+1} y'_{i+1}]$ and
    Lemma~\ref{lem:hourglass} imply the existence of a vertex $u'\in V^0(T)$
    with pairwise distinct children $v'_i,v'_{i+1}$ and $v'_{i+2}$ such that
    $x_{i+1}\preceq_T v'_i$, $\lca_T(x'_{i+1},y'_{i+1})\preceq_T v'_{i+1}$
    and $y_{i+1}\preceq_T v'_{i+2}$.  By the definition of hourglass chains,
    we have $y_i=x'_{i+1}$ and $y'_i=x_{i+1}$. Therefore,
    $u'=\lca_T(x'_{i+1},x_{i+1})=\lca_T(y_i,y'_i)=u$. Since $v_i$ and $v'_i$
    are both children of $u$, $y'_i=x_{i+1}$ and it holds both that
    $y'_i\preceq_T v_i$ and $x_{i+1}\preceq_T v'_i$, we conclude that
    $v_i=v'_i$. Similarly, it holds $v_{i+1}=v'_{i+1}$ since
    $v_{i+1},v'_{i+1}\in\child(u)$ and $y_i=x'_{i+1}$. In particular, we have
    $v'_{i+2}\ne v'_{i+1}=v_{i+1}$ and $v'_{i+2}\ne v'_{i}=v_{i}$. It remains
    to show that $v'_{i+2}\ne v_j$ for $0\le j<i$. Assume, for contradiction,
    that $v'_{i+2}=v_j$ for some fixed $j$ with $0\le j<i$. By assumption,
    $x_1\preceq_T v_j$ if $j=0$, and otherwise, $x_{j+1}=y'_j\preceq_T
    v_j$. Moreover, since $v'_{i+2}=v_j$, we have $y_{i+1}\preceq_T
    v_j$. Hence, $\lca_T(x_{j+1},y_{i+1})\preceq_T v_j$.  Furthermore, since
    $y'_{i+1}\preceq_T v_{i+1}\ne v_j$, it holds
    $\lca_T(x_{j+1},y'_{i+1})=u\succ_T v_j$. Since
    $\sigma(y_{i+1})=\sigma(y'_{i+1})$ by the definition of hourglasses, the
    latter two arguments contradict $x_{j+1}y'_{i+1}\in E(G)$, which must
    hold by the definition of hourglass chains. Hence, we can conclude that
    $v'_{i+2}\ne v_j$ for and $0\le j<i$ and we set
    $v_{i+2}\coloneqq v'_{i+2}$.  In summary, the statement holds for the
    hourglass chain $\mathfrak{H}_{i+1|} = \mathfrak{H}$.  
  \end{proof}
  It is straightforward to generalize the latter statement to tailed
  hourglass chains.
  \begin{lemma}
    \label{lem:hourglass_chain_tails}
    Let
    $\mathfrak{H}=[x_1 y_1 \hourglass x'_1 y'_1],\dots,[x_k y_k \hourglass
    x'_k y'_k]$ be an hourglass chain with left (resp.\ right) tail $z$ in a
    BMG $(\G,\sigma)$. Then, every tree $(T,\sigma)$ that explains
    $(\G,\sigma)$ contains a vertex $u\in V^0(T)$ with pairwise distinct
    children $v_0,v_1,\dots,v_k,v_{k+1}$ such that it holds
    $x_1\in L(T(v_0))$, $y_k\in L(T(v_{k+1}))$, and, for all $1\le i\le k$,
    we have $x'_i,y'_i\in L(T(v_i))$.  Furthermore, we have $z\preceq_T v_0$
    (resp.\ $z\preceq_T v_{k+1}$).
  \end{lemma}
  \begin{proof}
    By Lemma~\ref{lem:hourglass_chain}, there is a vertex $u\in V^0(T)$ with
    pairwise distinct children $v_0,v_1,\dots,v_k,v_{k+1}$ such that it holds
    $x_1\in L(T(v_0))$, $y_k\in L(T(v_{k+1}))$, and, for all $1\le i\le k$,
    we have $x'_i,y'_i\in L(T(v_i))$.
    
    Suppose that $z$ is a left tail of $\mathfrak{H}$. We need to show
    that $z\preceq_T v_0$. By definition, $(z,x_1)\in E(\G)$,
    $(z,x'_1)\notin E(\G)$, and $\sigma(x_1)=\sigma(x'_1)$.  Therefore,
    $zx_1|x'_1$ is an informative triple for $(\G,\sigma)$, and hence
    $\lca_{T}(z,x_1) \prec_T\lca_{T}(z,x'_1)=\lca_{T}(x_1,x'_1)=u$.  Since
    $v_0$ is the unique child of $u$ with $x_1\prec_T v_0$, we can conclude
    that $\lca_{T}(z,x_1)\preceq_{T} v_0$ and thus, $z\preceq_{T}v_0$.
    
    If $z$ is a right tail of $\mathfrak{H}$, a similar argument using the
    informative triple $z'y_k|y'_k$, which must be displayed by $T$ because
    $(z,y_k)\in E(\G)$ and $(z,y'_k)\notin E(\G)$, implies
    $z\preceq_T v_{k+1}$.
  \end{proof}
  
  We are now in the position to show that hourglass chains identify
  additional \ufp edges that are not contained in a single hourglass.
  \begin{lemma}
    \label{lem:hourglass_chain_dupl}
    Let
    $\mathfrak{H}=[x_1 y_1 \hourglass x'_1 y'_1],\dots,[x_k y_k \hourglass
    x'_k y'_k]$ be an hourglass chain in $(\G,\sigma)$, possibly with a left
    tail $z$ or a right tail $z'$. Then every edge
    $e\in\{x_1y_k, zy_k, x_1z', zz'\}\cap E(G)$ is \ufp, where $G$ denotes
    the symmetric part of $\G$.
  \end{lemma}
  \begin{proof}
    Let $(T,\sigma)$ be an arbitrary tree that explains $(\G,\sigma)$.  By
    the definition of hourglass chains, we have $k\ge 1$. Hence, the sequence
    contains at least the hourglass $[x_1 y_1 \hourglass x'_1 y'_1]$. Since
    $\mathfrak{H}=[x_1 y_1 \hourglass x'_1 y'_1],\dots,[x_k y_k \hourglass
    x'_k y'_k]$ in $\G(T,\sigma)$, Lemma~\ref{lem:hourglass_chain_tails}
    implies the existence of a vertex $u\in V^0(T)$ with pairwise distinct
    children $v_0,v_1,\dots,v_k,v_{k+1}$ such that it holds
    $x_1\in L(T(v_0))$, $y_k\in L(T(v_{k+1}))$, and, for all $1\le i\le k$,
    we have $x'_i,y'_i\in L(T(v_i))$. Furthermore, this lemma also implies 
    $z\preceq_T v_0$ if $z$ is a left tail of $\mathfrak{H}$, and
    $z'\preceq_T v_{k+1}$ if $z'$ is a right tail of $\mathfrak{H}$.  Note
    that $\lca_T(x_1,x'_1)=u$, and $x_1$ and $x'_1$ lie below distinct
    children of $u$. More precisely $x_1\preceq_T v_0$ and
    $x'_1\preceq_T v_1$. Since $\sigma(x_1)=\sigma(x'_1)$, we have
    $\sigma(L(T(v_0)))\cap\sigma(L(T(v_1)))\ne\emptyset$. Moreover,
    $\lca_{T}(a,b)=u$ for every edge $e=ab$ in $\G$ that coincides with one of
    $x_1y_k$, $zy_k$, $x_1z'$, and $zz'$. The latter two arguments together
    with Lemma~\ref{lem:T-fp-no-mu} imply that every such edge is
    $(T,\sigma)$-\fp. Since $(T,\sigma)$ was chosen arbitrarily, every such
    edge is also \ufp.
  \end{proof}
  
  It is important to note that the construction of hourglass chains does not
  imply that an edge $e\in\{x_1y_k, zy_k, x_1z', zz'\}$ must exist in
  $(\G,\sigma)$. Nevertheless, whenever such an edge occurs, it is \ufp. We
  will take a closer look at the properties of hourglass chains in
  Sec.~\ref{APP:sect:leftovers}.
  
  \section{Characterization of unambiguous false-positive edges}
  
  \subsection{Color-set intersection graphs}
  
  In this section, we take a closer look at the trees that explain a given
  BMG. In particular, we consider the color allocation to the subtrees below
  each vertex of a tree explaining a given BMG. This leads us to the idea of
  a color intersection graph.
  \begin{definition}
    The \emph{color-set intersection graph} $\CIG_T(u)$ of an inner vertex
    $u$ of a leaf-colored gene tree $(T,\sigma)$ is the undirected graph with
    vertex set $V\coloneqq\child_T(u)$ and edge set
    \begin{equation*}
    E\coloneqq\{ v_1v_2 \mid v_1,v_2\in V \textrm{, }v_1\ne v_2
    \textrm{ and } \sigma(L(T(v_1)))\cap\sigma(L(T(v_2)))\ne\emptyset \}.
    \end{equation*}
  \end{definition}
  
  Shortest paths in the color-set intersection graphs will play an important
  role in identifying many \ufp edges.
  \begin{lemma}
    \label{lem:shortest_path}
    Let $v_1$ and $v_k$ be two distinct vertices in the same connected
    component of the color-set intersection graph $\CIG_T(u)$ of a
    leaf-colored gene tree $(T,\sigma)$, and let
    $P(v_1,v_{k}) = (v_1, \dots, v_{k})$ be a shortest path in $\CIG_T(u)$
    connecting $v_1$ and $v_k$. Then
    $\sigma(L(T(v_i)))\cap\sigma(L(T(v_j)))=\emptyset$ for all $i$ and $j$
    satisfying $1\leq i<i+2 \leq j\leq k$.
  \end{lemma}
  \begin{proof}
    Assume, for contradiction, that
    $\sigma(L(T(v_i)))\cap\sigma(L(T(v_j)))\neq \emptyset$ for some $i,j$
    with $1\leq i<i+2 \leq j\leq k$. Then the edge $v_iv_j$ must be contained
    in $\CIG_T(u)$, contradicting the fact that $P(v_1,v_{k})$ is a
    \emph{shortest} path.
  \end{proof}
  
  The following lemma establishes a close connection between color-set
  intersection graphs and hourglass chains.
  \begin{lemma}
    \label{lem:hourglass_construction}
    Let $(\G,\sigma)$ be a BMG that is explained by $(T,\sigma)$ and suppose
    that $x,y\in L(T)$ are two distinct leaves with $u\coloneqq\lca_T(x,y)$
    and $v_x, v_y\in\child_T(u)$ such that (i) $x\preceq_T v_x$ and
    $y\preceq_T v_y$, and (ii) there is a shortest path
    $(v_x=v_0,v_1, \dots, v_{k}, v_{k+1}=v_y)$ of length at least two in
    $\CIG_T(u)$. Then there is an hourglass chain
    $\mathfrak{H}=[x_1 y_1 \hourglass x'_1 y'_1],\dots,[x_k y_k \hourglass
    x'_k y'_k]$ in $(\G,\sigma)$.  In particular, precisely one of the
    following conditions is satisfied:
    \begin{enumerate}[itemsep=0.2ex, topsep=0.2ex, parsep=0cm]
      \item $x_1=x$ and $y_k=y$;
      \item $y_k=y$ and $z\coloneqq x$ is a left tail of $\mathfrak{H}$;
      \item $x_1=x$ and $z'\coloneqq y$ is a right tail of $\mathfrak{H}$; or
      \item $z\coloneqq x$ is a left tail and $z'\coloneqq y$ is a right tail
      of $\mathfrak{H}$.
    \end{enumerate}
  \end{lemma}
  \begin{proof}
    Lemma~\ref{lem:shortest_path} implies
    $\Scap(x,y)=\sigma(L(T(v_x)))\cap\sigma(L(T(v_y))) =
    \sigma(L(T(v_0)))\cap\sigma(L(T(v_{k+1})))=\emptyset$. We proceed by
    showing that the BMG $\G(T,\sigma)$ contains an hourglass chain
    $\mathfrak{H}=[x_1 y_1 \hourglass x'_1 y'_1],\dots,[x_k y_k \hourglass
    x'_k y'_k]$ possibly with left tail $z$ and right tail $z'$ such that one
    of the Conditions 1--4 is satisfied.
    
    We first consider the two cases: either (A)
    $\sigma(x)\in\sigma(L(T(v_1)))$ or (B)
    $\sigma(x)\notin\sigma(L(T(v_1)))$.  In Case (A), we set $x_1\coloneqq x$
    and $c_0\coloneqq\sigma(x)$. In Case (B), we set $z\coloneqq x$, choose
    $c_0 \in \sigma(L(T(v_0)))\cap\sigma(L(T(v_1)))$ arbitrarily (note
    $v_0v_1$ forms an edge in $\mathfrak{C}_T(u)$ and thus, the latter
    intersection is non-empty) and we set $x_1 = v$ for some
    $v\in L(T(v_0))\cap L[c_0]$ such that
    $\lca(v,x)\preceq_T \lca_T(v',x)\preceq_T v_0$ for all
    $v'\in L(T(v_0))\cap L[c_0]$.  Clearly, such a vertex $v$
    exists. Moreover, $c_0\neq \sigma(x)$ and we obtain
    $(x, v) = (z,x_1)\in E(\G)$ as necessary requirement for left tails.  In
    summary, we have in Case (A) $x_1=x$ and in Case (B) $x$ plays the role
    of the left tail $z$ and $x_1$ is some other vertex.  Moreover, in both
    Cases (A) and (B), we have
    $\sigma(x_1)=c_0\in\sigma(L(T(v_0)))\cap\sigma(L(T(v_1)))$.
    
    We now consider the ``other end'' of the hourglass chain, that is, vertex
    $y_k$ and the possible right tail.  Again, we have two cases: either (A')
    $\sigma(y)\in\sigma(L(T(v_{k+1})))$ or (B')
    $\sigma(y)\notin\sigma(L(T(v_{k+1})))$.  In Case (A'), we set
    $y_k\coloneqq y$ and $c_k\coloneqq\sigma(y)$. In Case (B'), we set
    $z'\coloneqq y$, and , by similar arguments as in Case (A) and (B), we
    can choose $c_k \in \sigma(L(T(v_k)))\cap\sigma(L(T(v_{k+1})))$
    arbitrarily and set $y_k = w$ for some vertex
    $w\in L(T(v_{k+1}))\cap L[c_k]$ such that $(y,w) = (z',y_k)\in E(\G)$ as
    a necessary requirement for right tails.  Again, for both cases (A') and
    (B') we have
    $\sigma(y_k)=c_k\in\sigma(L(T(v_k)))\cap\sigma(L(T(v_{k+1})))$.
    
    We continue by picking an arbitrary color $c_i$ from
    $\sigma(L(T(v_i)))\cap\sigma(L(T(v_{i+1})))$ for each $1\le i < k$. This
    is possible because $v_i v_{i+1}\in E(\mathfrak{C}_T(u))$, and thus
    $\sigma(L(T(v_i)))\cap\sigma(L(T(v_{i+1})))\ne\emptyset$. Note that now
    $c_i\in\sigma(L(T(v_i)))\cap\sigma(L(T(v_{i+1})))$ holds for all
    $0\le i \le k$. In particular, the colors $c_0,c_1,\dots,c_k$ are
    pairwise distinct. To see this, assume, for contradiction, that $c_i=c_j$
    for some $i,j$ with $i<j$. Then $c_i\in\sigma(L(T(v_i)))$ and
    $c_i=c_j\in\sigma(L(T(v_{j+1})))$ which implies
    $c_i\in\sigma(L(T(v_i)))\cap\sigma(L(T(v_{j+1})))$. This contradicts
    Lemma~\ref{lem:shortest_path} for $j+1\ge i+2$.
    
    For each $1\le i\le k$, we have $c_{i-1}, c_i\in\sigma(L(T(v_i)))$. Thus
    Lemma~\ref{lem:exEdge} ensures the existence of vertices
    $x'_i\in L(T(v_i))\cap L[c_{i-1}]$ and $y'_i\in L(T(v_i))\cap L[c_{i}]$
    that form an edge $x'_i y'_i$ in $\G$. By assumption we have
    $x'_i y'_i\in E(G)$ for all $1\le i\le k$ since
    $[x_i y_i \hourglass x'_i y'_i]$ is an hourglass.  We already set $x_1$
    and $y_k$. We furthermore set $x_i\coloneqq y'_{i-1}$ for all $1<i\le k$,
    and $y_i\coloneqq x'_{i+1}$ for all $1\le i < k$. Thus ensures that (H1)
    in Def.~\ref{def:hc} is satisfied.  Moreover, since
    $\sigma(x_1)=c_0=\sigma(x'_1)$ and $\sigma(x_i)=\sigma(y'_{i-1})=c_{i-1}$
    for all $1<i\le k$, we have $\sigma(x_i)=c_{i-1}=\sigma(x'_i)$ for all
    $1\le i\le k$. Similar arguments imply $\sigma(y_i)=c_{i}=\sigma(y'_i)$
    for all $1\le i\le k$.
    
    We next show that the induced subgraph $\G[x_i,x'_i,y_i,y'_i]$ is an
    hourglass for $1\le i \le k$ and thus $x_i y'_j$ is an edge in $\G$ for
    all $i<j\le k$. We also know, by construction, that $x'_i y'_i$ is an
    edge in $\G$. 
    
    Independent of whether $x_1$ was constructed based on the cases (A) or
    (B), we have $x_i\preceq_T v_0$ if $i=1$ and
    $x_i=y'_{i-1}\preceq_T v_{i-1}$ otherwise. Thus $x_i\preceq_T
    v_{i-1}$. Likewise, independent of whether $y_k$ was constructed based on
    the cases (A') or (B'), we have $y_i\preceq_T v_{k+1}$ if $i=k$ and
    $y_i=x'_{i+1}\preceq_T v_{i+1}$ otherwise. Thus $y_i\preceq_T v_{i+1}$.
    In summary, we have $x_i\preceq_T v_{i-1}$; $x'_i,y'_i\preceq_T v_{i}$;
    and $y_i\preceq_T v_{i+1}$ for all $i\in \{1,\dots,k\}$.  This implies
    $\lca_T(x_i,y'_i)=\lca_T(x_i,y_i)=\lca_T(x'_i,y_i)=u$. Since
    $i+1 \ge (i-1)+2$ and $P(v_0,v_{k+1})$ is a shortest path,
    Lemma~\ref{lem:shortest_path} implies
    $\sigma(L(T(v_{i-1})))\cap\sigma(L(T(v_{i+1})))=\emptyset$.
    
    From $\sigma(x_i)\in\sigma(L(T(v_{i-1})))$ and
    $\sigma(y_i)\in\sigma(L(T(v_{i+1})))$ we obtain
    $\sigma(x_i)\notin\sigma(L(T(v_{i+1})))$ and
    $\sigma(y_i)\notin\sigma(L(T(v_{i-1})))$. Thus, there is no $\widetilde{y}$
    such that $\sigma(\widetilde{y})=\sigma(y'_i)=\sigma(y_i)$ and
    $\lca_T(x_i,\widetilde{y})\prec_T u =\lca_T(x_i,y'_i)=\lca_T(x_i,y_i)$, and
    no $\widetilde{x}$ such that 
    $\sigma(\widetilde{x})=\sigma(x'_i)=\sigma(x_i)$ 
    and
    $\lca_T(y_i,\widetilde{x})\prec_T u =\lca_T(y_i,x'_i)=\lca_T(y_i,x_i)$.
    Hence, $\G$ contains the arcs $(x_i,y'_i)$, $(x_i,y_i)$, $(y_i,x_i)$ and
    $(y_i,x'_i)$. Moreover, $x_i y_i$ is an edge in $\G$. However, since
    $\sigma(x'_i)=\sigma(x_i)$ and
    $\lca_T(x'_i,y'_i)\preceq_T v_i\prec_T u = \lca_T(x_i,y'_i)$ we conclude
    $(y'_i,x_i)\notin E(\G)$. Likewise, $\sigma(y'_i)=\sigma(y_i)$ and
    $\lca_T(x'_i,y'_i)\preceq_T v_i\prec_T u = \lca_T(x'_i,y_i)$ imply that
    $(x'_i,y_i)\notin E(\G)$. In summary,
    $\G[x_i,x'_i,y_i,y'_i]=[x_i y_i \hourglass x'_i y'_i]$ is an hourglass,
    for all $i\in \{1,\dots,k\}$, and $x_i\preceq_T v_{i-1}$ and
    $y'_j\preceq_T v_{j}$ for all $1\le i<j\le k$. 
    
    Since $j\ge (i-1)+2$ and $P(v_0,v_{k+1})$ is a shortest path,
    Lemma~\ref{lem:shortest_path} implies that
    $\sigma(L(T(v_{i-1})))\cap\sigma(L(T(v_{j})))=\emptyset$. Thus, there is
    no $\widetilde{y}$ such that $\sigma(\widetilde{y})=\sigma(y'_j)$ and
    $\lca_T(x_i,\widetilde{y})\prec_T u =\lca_T(x_i,y'_j)$, and no 
    $\widetilde{x}$
    such that $\sigma(\widetilde{x})=\sigma(x_i)$ and
    $\lca_T(y'_j,\widetilde{x})\prec_T u =\lca_T(y'_j,x_i)$. This implies that
    $(x_i,y'_j)\in E(\G)$ and $(y'_j,x_i)\in E(\G)$, respectively. Therefore
    $x_i y'_j$ is an edge in $\G$ for $1\le i<j\le k$.  In summary, (H2) of
    in Def.~\ref{def:hc} is always satisfied.
    
    Hence, if $x_1$ and $y_1$ are constructed based on Case (A) and (A'),
    respectively, we are done.
    
    It remains to show that $z$ and $z'$ are a left and a right tail, resp.,
    of the hourglass chain in Case (B) or (B').  First assume Case (B), and
    thus $z=x$.  We have $z,x_1\preceq_T v_0$ by construction and
    $(z, x_1)\in E(\G)$ as shown above.  Together with $x'_1\preceq_T v_1$,
    this implies that
    $\lca_T(z,x_1)\preceq_T v_0 \prec_T u = \lca_T(z,x'_1)$. Using
    $\sigma(x_1)=\sigma(x'_1)$ we therefore obtain $(z,x'_1)\notin E(\G)$.
    and hence $z$ is a left tail of the constructed hourglass chain.  Now
    assume Case (B'), and thus, $z'=y$.  We have $z',y_k\preceq_T v_{k+1}$
    and $(z', y_k)\in E(\G)$ by construction.  Together with
    $y'_k\preceq_T v_k$ this implies
    $\lca_T(z',y_k)\preceq_T v_{k+1} \prec_T u = \lca_T(z',y'_k)$. Using
    $\sigma(y_k)=\sigma(y'_k)$, we obtain $(z',y'_k)\notin E(\G)$ and hence
    $z'$ is a right tail of the constructed hourglass chain.
    
    In summary,
    $\mathfrak{H}=[x_1 y_1 \hourglass x'_1 y'_1],\dots,[x_k y_k \hourglass
    x'_k y'_k]$ is an hourglass chain, possibly with left tail $z$ and right
    tail $z'$.  Furthermore, precisely one of the Conditions 1--4 in the
    statement holds by construction.
  \end{proof}
  
  \subsection{Hug-edges and no-hug graphs}
  \label{APP:ssect:hug}
  
  \begin{definition}\label{def:hug-edge}
    An edge $xy$ in a vertex-colored graph $(\G,\sigma)$ is a
    \emph{hug-edge} if it satisfies at least one of the following
    conditions:
    \begin{description}[itemsep=0.2ex, topsep=0.2ex, parsep=0cm]
      \item[\emph{(C1)}] $xy$ is the middle edge of a good quartet in
      $(\G,\sigma)$;
      \item[\emph{(C2)}] $xy$ is the first edge of an ugly quartet in
      $(\G,\sigma)$; or
      \item[\emph{(C3)}] there is an hourglass chain
      $\mathfrak{H}=[x_1 y_1 \hourglass x'_1 y'_1],\dots,[x_k y_k \hourglass
      x'_k y'_k]$ in $(\G,\sigma)$, and one of the following cases holds:
      \begin{enumerate}[noitemsep, topsep=0.2ex, parsep=0cm, nolistsep]
        \item $x_1=x$ and $y_k=y$;
        \item $y_k=y$ and $z\coloneqq x$ is a left tail of $\mathfrak{H}$;
        \item $x_1=x$ and $z'\coloneqq y$ is a right tail of $\mathfrak{H}$; or
        \item $z\coloneqq x$ is a left tail and $z'\coloneqq y$ is a right tail
        of $\mathfrak{H}$.
      \end{enumerate}
    \end{description}
  \end{definition}  
  The term \textbf{hug}-edge refers to the fact $xy$ is a particular edge of
  an \textbf{h}ourglass-chain, an \textbf{u}gly quartet, or a \textbf{g}ood
  quartet.
  
  \begin{theorem}
    \label{thm:good_ugly_or_hourglass}
    An edge $xy$ in $\G(T,\sigma)$ with $u\coloneqq\lca_T(x,y)$,
    $v_x, v_y\in\child_T(u)$, $x\preceq_T v_x$, and $y\preceq_T v_y$ is a
    hug-edge if $v_x$ and $v_y$ belong to the same connected component of
    $\CIG_T(u)$.  Moreover, every hug-edge is \ufp.
  \end{theorem}
  \begin{proof}
    We show first that $xy$ satisfies one of the Conditions \AX{(C1)},
    \AX{(C2)}, or \AX{((C3)}, and hence is hug-edge. First, note that
    $v_x\ne v_y$. Moreover, Lemma~\ref{lem:edge-xy-lca} implies
    $\sigma(x)\notin\sigma(L(T(v_y)))$ and
    $\sigma(y)\notin\sigma(L(T(v_x)))$. Since by assumption $v_x,v_y$ belong
    to the same connected component, there is a shortest path
    $P\coloneqq (v_x=v_0,\dots,v_{k+1}=v_y)$ in $\CIG_T(u)$. For $k=0$,
    $v_x v_y\in E(\CIG_T(u))$. This implies 
    $\Scap(x,y)=\sigma(L(T(v_x)))\cap\sigma(L(T(v_y)))\ne\emptyset$. By
    Prop.~\ref{prop:good_or_ugly}, the edge $xy$ is either the middle edge
    of a good quartet or the first edge of an ugly quartets in $(\G,\sigma)$.
    Hence, Condition \AX{(C1)} or \AX{(C2)} is satisfied. If $k>0$,
    Lemma~\ref{lem:hourglass_construction} implies Condition \AX{(C3)}.
    
    For each of the three cases we have already shown  that $xy$ is \ufp: For
    \AX{(C1)} Prop.~\ref{prop:good_quartet_middle_edge} applies, for
    \AX{(C2)} Prop.~\ref{prop:ugly_quartet} provides the desired result, and
    for \AX{(C3)} we use Lemma~\ref{lem:hourglass_chain_dupl}.
  \end{proof}
  
  \begin{lemma}
    \label{lem:edited_RBMG_color_intersection}
    If the BMG $\G(T,\sigma)$ contains a hug-edge $xy$ in a BMG
    $\G(T,\sigma)$, then there are distinct vertices
    $v_1,v_2\in\child_{T}(\lca_T(x,y))$ such that
    $\sigma(L(T(v_1)))\cap\sigma(L(T(v_2)))\ne\emptyset$.
  \end{lemma}
  \begin{proof}
    Let $xy$ be a hug-edge in the BMG $(\G,\sigma) = \G(T,\sigma)$, i.e.\ one
    of \AX{(C1)}, \AX{(C2)}, or \AX{(C3)} applies.
    
    If $e=xy$ satisfies \AX{(C1)}, then $xy$ is the middle edge of a good
    quartet $\langle zxyz'\rangle$ in $(\G,\sigma)$.  By 
    \cite[Lemma~36]{Geiss:19b}, there is a vertex 
    $u\coloneqq\lca_{T}(x,y,z,z')$ 
    such that $x,z\preceq_{T}v_1$ and $y,z'\preceq_{T}$ for some distinct
    $v_1,v_2\in\child_{T}(u)$. Thus, $u=\lca_{T}(x,y)$. Moreover, since
    $\sigma(z)=\sigma(z')$, we have
    $\sigma(L(T(v_1)))\cap\sigma(L(T(v_2)))\ne\emptyset$ for two distinct
    vertices $v_1,v_2\in\child_{T}(u)$.
    
    If $e=xy$ satisfies \AX{(C2)}, then it is the first edge of some ugly
    quartet, which w.l.o.g.\ has the form $\langle xyx'z\rangle$. Re-using
    the arguments in the proof of Prop.~\ref{prop:ugly_quartet} shows that
    there must be two distinct children $v_1$ and $v_2$ of vertex
    $u=\lca_{T}(x,y)$ such that
    $\sigma(L(T(v_1)))\cap\sigma(L(T(v_2)))\ne\emptyset$.
    
    If $e=xy$ satisfies \AX{(C3)}, then there is a (tailed) hourglass chain
    $\mathfrak{H}=[x_1 y_1 \hourglass x'_1 y'_1],\dots,[x_k y_k \hourglass
    x'_k y'_k]$, $k\ge 1$, in $\G(T,\sigma)$, such that either $x=x_1$ or
    $z\coloneqq x$ is a left tail of $\mathfrak{H}$, and either $y=y_k$ or
    $z'\coloneqq y$ is a right tail of $\mathfrak{H}$.  In either case,
    Lemma~\ref{lem:hourglass_chain_tails} implies $x\preceq_{T} v_0$ and
    $y\preceq_{T} v_{k+1}$.  Since $x_1$ and $x'_1$ lie below distinct
    children $v_0$ and $v_1$ of vertex $\lca_T(x,y)$ and
    $\sigma(x_1)=\sigma(x'_1)$ by the definition of hourglasses, it holds
    that $\sigma(L(T(v_0)))\cap\sigma(L(T(v_1)))\ne\emptyset$.
    
    In each case, therefore, there are distinct vertices
    $v_1,v_2\in\child_{T}(\lca_T(x,y))$ such that
    $\sigma(L(T(v_1)))\cap\sigma(L(T(v_2)))\ne\emptyset$.
  \end{proof}
  
  The fact that all hug-edges are \ufp by
  Thm.~\ref{thm:good_ugly_or_hourglass} suggests to consider the subgraph of
  a BMG that is left after removing all these unambiguously recognizable
  false-positive orthology assignments.
  \begin{definition}
    \label{def:non-hug-graph}
    Let $(\G,\sigma)$ be a BMG with symmetric part $G$ and let $F$ be the set
    of its hug-edges. The \emph{no-hug} graph $\NH(\G,\sigma)$ is the
    subgraph of $G$ with vertex set $V(\G)$, coloring $\sigma$ and edge set
    $E(G)\setminus F$.
  \end{definition}
  The $\NH(\G,\sigma)$ is therefore the subgraph of the underlying RBMG of
  $\G$ that contains all edges that cannot be identified as \ufp by using
  only good quartets, ugly quartets and (tailed) hourglass chains as outlined
  in Thm.~\ref{thm:good_ugly_or_hourglass}.
  
  \begin{corollary}
    \label{cor:NH}
    Let $(T,\sigma)$ be a leaf-colored tree and $\mu$ a reconciliation map
    from $(T,\sigma)$ to some species tree $S$. Then,
    \begin{equation*}
    \Theta(T,t_{\mu}) \subseteq \Theta(T, \tT ) \subseteq
    \NH(\G(T,\sigma)) \subseteq \G(T,\sigma).
    \end{equation*}
  \end{corollary}
  \begin{proof}
    By Thm.~\ref{thm:extrem-ortho},
    $\Theta(T,t_{\mu}) \subseteq \Theta(T,\tT) \subseteq \G(T,\sigma)$; and
    by definition, we have $\NH(\G(T,\sigma)) \subseteq \G(T,\sigma)$.  Now,
    let $xy$ be an edge in $\Theta(T,\tT)$ and thus,
    $\tT(\lca_T(x,y))=\SPEC$.  By definition of $\tT$, we have
    $\sigma(L(T(v_1))) \cap \sigma(L(T(v_2)))=\emptyset$ for any two distinct
    $v_1,v_2\in\child_T(\lca_T(x,y))$. The contraposition of 
    Lemma~\ref{lem:edited_RBMG_color_intersection} implies that $xy$ is not a
    hug-edge and thus an edge of $\NH(\G(T,\sigma))$, which completes the
    proof.  
  \end{proof}
  
  The no-hug graph still may contain false-positive orthology assignments,
  i.e., $\NH(\G(T,\sigma))=\Theta(T,t_{\mu})$ does not hold in general. In
  the following section, we shall see that there are, however, no \ufp edges
  left in the no-hug graph.
  
  \subsection{Resolving least resolved trees}
  \label{APP:ssect:augtree} 
  
  Since every BMG $(\G,\sigma)$ at least implicitly contains all information
  needed to identify its \ufp edges, this is also true for its unique least
  resolved tree $(T^*,\sigma)$. It is not always possible, however, to assign
  an event labeling $t$ to $T^*$ such that $(T^*,t)$ is the cotree for the
  correct orthology relation.  Fig.~\ref{fig:messy_vertices-1} shows that
  $T^*$ may not be ``resolved enough''. To tackle this problem, we analyze
  the redundant edges of more resolved trees that explain
  $(\G,\sigma)$. Cor.~\ref{cor:edge_redundant} implies that all edges below a
  speciation vertex are redundant because, by
  Lemma~\ref{lem:duplication_witness}, the color sets of distinct subtrees
  below a speciation vertex do not overlap.  More precisely, we have
  \begin{fact}
    \label{obs:speciations_merged}
    Let $\mu$ be a reconciliation map from $(T,\sigma)$ to $S$ and assume
    that there is a vertex $u\in V^0(T)$ such that $\mu(u)\in V^0(S)$ and
    thus, $t_{\mu}(u)=\SPEC$.  Then every inner edge $uv$ of $T$ with
    $v\in\child_{T}(u)$ is redundant w.r.t.\ $\G(T,\sigma)$.  Moreover, if an
    inner edge $uv$ with $v\in\child_{T}(u)$ is non-redundant, then $u$ must
    have two children with overlapping color sets, and hence,
    $t_{\mu}(u)=\DUPL$.
  \end{fact}
  
  To identify the vertices in $(T^*,\sigma)$ that can be expanded to yield a
  tree that still explains $\G(T^*,\sigma)$, we introduce a particular way of
  ``augmenting'' a leaf-colored tree.
  \begin{definition}
    \label{def:augmenting}
    Let $(T,\sigma)$ be a leaf-colored tree, $u$ be an inner vertex of $T$,
    $\CIG_T(u)$ the corresponding color-set intersection graph, and
    $\mathcal{C}$ the set of connected components of $\CIG_T(u)$. Then
    the tree $T_u$ \emph{augmented at vertex $u$} is obtained by
    applying the following editing steps to $T$:
    \begin{itemize}[itemsep=0.2ex, topsep=0.2ex, parsep=0cm]
      \item If $\CIG_T(u)$ is connected, do nothing.
      \item Otherwise, for each $C\in\mathcal{C}$ with $|C|>1$
      \begin{itemize}[noitemsep,nolistsep]
        \item introduce a vertex $w$ and attach it as a child of $u$, i.e., add
        the edge $uw$,
        \item for every element $v_i\in C$, substitute the edge $uv_i$ by the
        edge $wv_i$.
      \end{itemize}
    \end{itemize}
    The augmentation step is \emph{trivial} if $T_u=T$, in which case
    we say that \emph{no edit step was performed}.
  \end{definition}
  An example of an augmentation is shown in
  Fig.~\ref{fig:augmenting_labeling_algo}.  It is easy to see that the tree
  $T_u$ obtained by an augmentation of a phylogenetic tree $T$ is again a
  phylogenetic tree.  The augmentation step at vertex $u$ of $T$ is trivial
  if and only if either $\CIG_T(u)$ is connected or all connected components
  $C\in\mathcal{C}$ are singletons, i.e., $|C|=1$.  If $(T_u,\sigma)$ is
  obtained by augmenting $(T,\sigma)$ at node $u$, we denote the set of newly
  introduced vertices by $V_{\neg T}\coloneqq V(T_u)\setminus V(T)$. Note
  that $V_{\neg T}=\emptyset$ whenever no edit step was performed.
  
  Since augmentation only inserts vertices between $u$ and its children, it
  affects neither $L(T(u))$ nor $L(T(v))$ for $v\in\child(u)$. As an
  immediate consequence we find
  \begin{fact}
    Let $(T,\sigma)$ be a leaf-colored tree, $u\ne v$ two inner vertices of
    $T$, $\CIG_T(u)$ the corresponding color-set intersection graph, and
    $(T_u,\sigma)$ the tree obtained by augmenting $T$ at $u$. Then
    $\CIG_{T_u}(v)=\CIG_{T}(v)$.
    \label{fact:aug-indep}
  \end{fact}

  \begin{lemma}
    \label{lem:augmenting_color_disjoint}
    Let $(T,\sigma)$ be a leaf-colored tree. Let $u\in V^0(T)$ and $T_u$ be
    the tree after augmenting $T$ at vertex $u$. If $\CIG_T(u)$ is
    disconnected, then $\sigma(L(T_u(w_1)))\cap\sigma(L(T_u(w_2)))=\emptyset$
    for any two distinct vertices $w_1,w_2\in\child_{T_u}(u)$.
  \end{lemma}
  \begin{proof}
    By construction, the vertex $w_i$ in $T_u$, $i=1,2$, is either a child of
    $u$ in $T$ or was inserted in the augmentation step. Therefore, the two
    connected components $C_1$ and $C_2$ of $\CIG_T(u)$ to which $w_1$ and
    $w_2$ belong are disjoint. Thus
    $\sigma(L(T(v_i)))\cap \sigma(L(T(v_j)))= \emptyset$ for all
    $v_i,v_j\in \child_T(u)$ with $v_i\in C_1$ and $v_j\in C_2$ because
    otherwise there would be an edge $v_iv_j$ in $\mathfrak{C}_T(u)$ and
    thus, $C_1=C_2$.  Since $w_i$ is either the single vertex in $C_i$ or
    $w_i$ has as children the vertices of $C_i$ in $T_u$, $i\in \{1,2\}$, we
    conclude that $\sigma(L(T_u(w_1)))\cap\sigma(L(T_u(w_2)))=\emptyset$.
  \end{proof}
  
  The following result shows that no further edit step can be performed at
  vertices that have been newly introduced by a previous augmentation
  step or have already undergone an augmentation.
  \begin{lemma}
    Let $(T,\sigma)$ be a leaf-colored tree, $u\in V^0(T)$, $(T_u,\sigma)$
    the tree obtained by augmenting $T$ at $u$, and denote by
    $(T_{uw},\sigma)$ the tree obtained by augmenting $T_u$ at $w$. Then
    $T_{uw}=T_u$ for $w=u$ as well as for all newly introduced vertices,
    i.e., for all $w\in V_{\neg T}\cup\{u\}$.
    \label{lem:augment-once}
  \end{lemma}
  \begin{proof}
    If $T_u=T$, then $V_{\neg T}=\emptyset$ and thus $T_{uu}=T_{u}=T$.  If
    $T_u\ne T$, then the definition of the augmentation step at $u$ implies
    that either $\CIG_{T_u}(u)$ is connected or all connected components of
    $\CIG_{T_u}(u)$ are singletons. In either case
    Lemma~\ref{lem:augmenting_color_disjoint} ensured that augmentation at
    $u$ leaves $T_{u}$ unchanged, i.e., $T_{uu}=T_u$.  By construction,
    $\CIG_{T_u}(w)$ is connected for $w\in V_{\neg T}\setminus\{u\}$ and thus,
    we have $T_{uw}=T_{u}$.  
  \end{proof}
  The tree obtained by augmenting a set of inner vertices of $(T,\sigma)$ is
  therefore independent of the order of the augmentation steps.
  \begin{definition}[Augmented tree]
    Let $(T,\sigma)$ be a leaf-colored tree. The \emph{augmented tree of 
      $(T,\sigma)$}, denoted by $(\aug(T),\sigma)$, is obtained by augmenting 
      all 
    inner vertices of $(T,\sigma)$.
    \label{def:aug-tree}
  \end{definition}
  \begin{lemma}
    For every leaf-colored tree $(T,\sigma)$ there is a unique tree
    $(\aug(T),\sigma)$ obtained from $(T,\sigma)$ by repeated application of
    augmentation steps until only trivial augmentation steps remain.  The
    tree $(\aug(T),\sigma)$ is computed by Alg.~\ref{alg:augment_extremal}.
    \label{lem:augment_extremal}
  \end{lemma}
  \begin{proof}
    Lemma~\ref{lem:augment-once} together with Obs.~\ref{fact:aug-indep} 
    implies 
    that (i) every vertex $u$ in $T$ can be
    non-trivially augmented at most once, (ii) the newly introduced vertices
    cannot be non-trivially augmented at all, and (iii) augmentation of two
    distinct inner vertices of $T$ yields the same result irrespective of the
    order of the augmentation steps. Thus, $(\aug(T),\sigma)$ is unique. The
    correctness of Alg.~\ref{alg:augment_extremal} now follows
    immediately.  
  \end{proof}
  
  \begin{algorithm}[t]
    \caption{Augmented tree}
    \label{alg:augment_extremal}
    \SetAlgoLined
    \DontPrintSemicolon
    \KwData{Leaf-colored phylogenetic tree $(T,\sigma)$}
    \KwResult{Augmented tree  $(\aug(T),\sigma)$}
    \ForEach{$u \in V^0(T)$ in pre-order}{
      Compute $\CIG_T(u)$.\; $\mathcal{C}\leftarrow$ set of connected
      components of $\CIG_T(u)$\;
      \uIf{$|\mathcal{C}|>1$}{
        \ForEach{$C\in\mathcal{C}$ such that $|C|>1$}{ Introduce a vertex $w$
          and the edge $uw$.\label{line:new_edge_start}\;
          \ForEach{$v_i\in C$}{ Remove the edge $uv_i$.\; Add the edge
            $w v_i$.\label{line:new_edge_end}\;
          }
        }
      }
    }
  \end{algorithm}
  
  \begin{lemma}
    Alg.~\ref{alg:augment_extremal} with input $T=(V,E)$ and $\sigma$ runs in
    $O(|V|^2|\mathscr{S}|)$ time and $O(|V|^2)$ space, where
    $\mathscr{S} = \sigma(L(T))$ is the set of species under consideration.
    \label{lem:augalg}
  \end{lemma}
  \begin{proof}
    Assigning the color set $L(T(u))$ to each $u$ requires
    $O(|V| |\mathscr{S}|)$ time, where $|\mathscr{S}|<|V|$.  The total effort
    to construct all $\CIG_T(u)$ is bounded by $O(|V|^2|\mathscr{S}|)$,
    corresponding to comparing the color sets of all pairs of vertices of
    $T$. The total size of all color-set intersection graphs in $O(|V|^2)$.
    Computation of the connected components is linear in the size of the
    graph, which also bounds the editing effort for each $u$, implying the
    claim.
  \end{proof}
  
  We finally show that augmentation does not affect the underlying BMG.
  \begin{proposition}\label{prop:aug-bmg}
    \label{prop:augmenting_algo_same_BMG}
    For every leaf-colored tree $(T,\sigma)$, it holds
    $\G(T,\sigma)=\G(\aug(T),\sigma)$.
  \end{proposition}
  \begin{proof}
    Let $u\in V^0(T)$ and $T_u$ be the tree after augmenting $T$ at
    vertex $u$. Put $A\coloneqq\{uw\mid w\in V_{\neg T}\}$ and note that all
    edges of $T_u$ in $A$ are inner edges. Now consider $e\in A$. Since
    $w\in V_{\neg T}$, an edit step was performed to obtain $w$ and thus,
    $|\mathcal{C}|>1$ in
    $\mathfrak{C}_T(u)$. Lemma~\ref{lem:augmenting_color_disjoint} and
    $|\mathcal{C}|>1$ imply that for any $v'\in\child_{T_u}(u)$ with $v'\ne w$
    we have $\sigma(L(T_u(v')))\cap\sigma(L(T_u(w)))=\emptyset$.  Thus, 
    Cor.~\ref{cor:edge_redundant} implies that the edge $uw$ is redundant in
    $(T_u,\sigma)$ w.r.t.\ $\G(T,\sigma)$.
    
    Denoting by $T_{u_A}$ the tree
    obtained from  $T_u$ by contraction of all edges in $A$, we obtain
    $(T,\sigma) = (T_{u_A},\sigma)$. Lemma~\ref{lem:redundant_commutative} now
    implies $\G(T_u,\sigma)=\G(T_{u_A},\sigma)=\G(T,\sigma)$ for every
    augmentation step. By Lemma~\ref{lem:augment_extremal}, we can repeat this
    argument for every augmentation in the arbitrary order in which
    $\G(\aug(T),\sigma)$ is obtained from  $\G(T,\sigma)$, and thus
    $\G(\aug(T),\sigma)=\G(T,\sigma)$.
  \end{proof}
  
  \subsection{Extremal labeling of augmented trees}
  \label{APP::ssec-sec:augmented-extremal-labeling}
  
  While the least resolved tree in general cannot support an event labeling
  that properly reflects the underlying true history of a gene family, we
  shall see here that the augmented tree $(\aug(T),\sigma)$ does feature
  sufficient resolution. To this end, we investigate the extremal
  event labeling of $(\aug(T),\sigma)$.
  \begin{lemma}
    \label{lem:aumented_tree_SPEC_DUPL}
    Let $\wt\coloneqq\tTp$ be the extremal event labeling of the augmented
    tree $(\aug(T),\sigma)$ obtained from $(T,\sigma)$ and let $u$ be some
    vertex of $\aug(T)$.  Then it holds $\wt(u)=\DUPL$ if and only if
    $\CIG_{\aug(T)}(u)$ is connected.
  \end{lemma}
  \begin{proof}
    By the definitions of the extremal event labeling and
    $\CIG_{\aug(T)}(u)$, the `if'-direction is clear.  Now suppose that
    $\wt(u)=\DUPL$.  There are two possibilities:
    \par\noindent (1)  $u\in V^0(T)$. If $\CIG_T(u)$ is connected, then
    $\CIG_{\aug(T)}(u)=\CIG_T(u)$. Otherwise, 
    Lemma~\ref{lem:augmenting_color_disjoint} implies that
    $\sigma(L(\aug(T)(w_1)))\cap\sigma(L(\aug(T)(w_2))) = \emptyset$ for all
    $w_1,w_2\in \child_{\aug(T)}(u)$, thus the definition of the extremal
    event labeling implies $\wt(u)\neq\DUPL$, a contradiction.
    \par\noindent (2) $u\in V_{\neg T}$, i.e., $u$ is newly created by
    augmenting some $u'\in V^0(T)$, hence $\CIG_{T}(u)$ is connected and, by
    Obs.~\ref{fact:aug-indep} and Lemma~\ref{lem:augment-once},
    $\CIG_{\aug(T)}(u)$ is connected.
  \end{proof}
  
  For later reference, we need the following
  \begin{lemma}
    \label{lem:augmented_tree_no_adj_spec}
    Let $(\G,\sigma)$ be a BMG, $(T^*,\sigma)$ its least resolved tree, and
    $\wt\coloneqq \tTps$ the extremal event labeling of the augmented tree
    $(\aug(T^*),\sigma)$.  Then, $(\aug(T^*),\wt,\sigma)$ does not contain
    adjacent speciation vertices, i.e., if $\wt(u)=\SPEC$ for a vertex $u$ of
    $\aug(T^*)$, then $\wt(v)=\DUPL$ for any of its non-leaf children
    $v\in\child_{\aug(T^*)}(u)\setminus L(\aug(T^*))$.
  \end{lemma}
  \begin{proof}
    Set $\aug\coloneqq\aug(T^*)$ and note that, by Prop.~\ref{prop:aug-bmg},
    $(\aug,\sigma)$ explains $(\G,\sigma)$.  Assume, for contradiction, that
    there is an inner edge $uv$ in $\aug$ with $v\prec_{\aug} u$ such that
    $\wt(u)=\wt(v)=\SPEC$.  By the definition of the extremal event labeling
    $\wt$, we have $\sigma(L(\aug(v)))\cap\sigma(L(\aug(v')))=\emptyset$ for
    any $v'\in\child_{\aug}(u)\setminus\{v\}$.  Together with
    Cor.~\ref{cor:edge_redundant} this implies that $uv$ is redundant for
    $(\G,\sigma)$, and hence, not an edge in the least resolved tree
    $(T^*,\sigma)$.  Now consider the augmentation in which the edge $uv$,
    and thus vertex $v$ was created; resulting in a tree $(T',\sigma)$.  By
    the definition of augmenting (Def.~\ref{def:augmenting}), it clearly
    holds that $\CIG_{T'}(v)$ is connected.  By Lemma~\ref{lem:augment-once},
    the edges adjacent to $v$ do not change in any subsequent augmentation.
    Thus $\CIG_{\aug}(v)$ must be connected as well.
    Lemma~\ref{lem:aumented_tree_SPEC_DUPL} now implies that $\wt(v)=\DUPL$;
    a contradiction.  
  \end{proof}
  
  \begin{lemma}
    \label{lem:aumented_subgraph_RBMG}
    Let $(\G,\sigma)$ be a BMG and $(T^*,\sigma)$ its unique least resolved
    tree.  Moreover, let $\wt\coloneqq \tTps$ be the extremal event labeling
    of the augmented tree $(\aug(T^*),\sigma)$.  Then,
    $\Theta(\aug(T^*),\wt) \subseteq \G$.
  \end{lemma}
  \begin{proof}
    Since $(T^*,\sigma)$ explains $(\G,\sigma)$, we have
    $(\G,\sigma) = \G(T^*,\sigma)$.  By 
    Prop.~\ref{prop:augmenting_algo_same_BMG}, we have
    $\G(T^*,\sigma)=\G(\aug(T^*),\sigma)$.  Let $xy$ be an edge in
    $\Theta(\aug(T^*),\wt)$.  By definition, $\wt(\lca_{\aug(T^*)}(u))=\SPEC$
    where $u\coloneqq \lca_{\aug(T^*)}(x,y)$.  By definition of the
    extremal event labeling,
    $\sigma(L(\aug(T^*)(v_1)))\cap\sigma(L(\aug(T^*)(v_2)))=\emptyset$ for
    all two distinct vertices $v_1,v_2\in\child_{\aug(T^*)}(u)$.  The latter
    is true, in particular, for the two children
    $v_x,v_y\in\child_{\aug(T^*)}(u)$ with $x\preceq_{\aug(T^*)} v_x$ and
    $y\preceq_{\aug(T^*)} v_y$.  Therefore,
    $\sigma(x)\notin\sigma(L(\aug(T^*)(v_y)))$ and
    $\sigma(y)\notin\sigma(L(\aug(T^*)(v_x)))$. We conclude that $x$ and $y$
    are reciprocal best matches in $\aug(T^*)$. Finally,
    $(\G,\sigma) =\G(\aug(T^*),\sigma)$ implies that $xy$ is an edge in $\G$.
  \end{proof}
  
  Now we are in the position to prove the main results of this contribution. 
  \begin{theorem}
    Let $(\G,\sigma)$ be a BMG, $(T^*,\sigma)$ its unique least resolved
    tree, and $\wt\coloneqq \tTps$ the extremal event labeling of the
    augmented tree $(\aug(T^*),\sigma)$. Then
    $(\Theta(\aug(T^*),\wt),\sigma) = \NH(\G,\sigma)$.
    \label{thm:MAIN} 
  \end{theorem}
  \begin{proof}
    Let $(G,\sigma)$ be the symmetric part of $(\G=(V,E),\sigma)$.  For
    simplicity, we write $G_{\Theta} \coloneqq \Theta(\aug(T^*),\wt)$ and
    $G_{\NH} \coloneqq (V, E(\NH(\G,\sigma)))$.  Recall that, by definition,
    $G_{\NH}\subseteq G$ and, by Lemma~\ref{lem:aumented_subgraph_RBMG},
    $G_{\Theta}\subseteq \G$. Finally, as $G$ contains only edges of $\G$, we
    have $G_{\Theta}\subseteq G$.  Let
    $F \coloneqq E(G) \setminus E(G_{\NH}) $ be the set of all edges of $G$
    that are hug-edges, and let $F' \coloneqq E(G)\setminus E(G_{\Theta})$ be
    the set of all edges in $G$ that do not form orthologous pairs. Since
    $G_{\NH},G_{\Theta}\subseteq G$ it suffices to verify that $F=F'$ in
    order to show that $(G_{\Theta},\sigma)=(G_{\NH},\sigma)$.
    
    Assume $e=xy\in F'$. Hence, $xy\notin E(G_{\Theta})$ and therefore,
    $\wt(u)=\DUPL$ where $u\coloneqq\lca_{\aug(T^*)}(x,y)$.  By
    Lemma~\ref{lem:aumented_tree_SPEC_DUPL}, $\CIG_{\aug(T^*)}(u)$ has
    exactly one connected component.  This together with 
    Thm.~\ref{thm:good_ugly_or_hourglass} implies that $xy$ is a hug-edge and
    thus, $xy\in F$, and hence $F'\subseteq F$.
    
    Assume $e=xy\in F$ is a hug-edge.  Assume, for contradiction, that
    $e\notin F'$ and thus, $\wt(u)=\SPEC$ where
    $u\coloneqq\lca_{\aug(T^*)}(x,y)$. By definition of the extremal
    event labeling, it must therefore hold that
    $\sigma(L(\aug(T^*)(v_1)))\cap\sigma(L(\aug(T^*)(v_2)))=\emptyset$ for
    any two distinct vertices $v_1,v_2\in\child_{\aug(T^*)}(u)$. By 
    Prop.~\ref{prop:aug-bmg}, $(\aug(T^*),\sigma)$ explains $(\G,\sigma)$. This
    together with Lemma~\ref{lem:edited_RBMG_color_intersection} implies that
    there are two distinct vertices $v_1,v_2\in\child_{\aug(T^*)}(u)$ such
    that
    $\sigma(L(\aug(T^*)(v_1)))\cap\sigma(L(\aug(T^*)(v_2)))\neq \emptyset$; a
    contradiction. Therefore, $e\in F'$, and hence $F\subseteq F'$.
  \end{proof}
  
  \begin{theorem}
    An edge $xy$ in a BMG $(\G,\sigma)$ is \ufp if and only if $xy$ is a
    hug-edge of $(\G,\sigma)$.
    \label{thm:ufp-iff-hug}
  \end{theorem}
  \begin{proof}
    Let $(\G,\sigma)$ be a BMG, $(T^*,\sigma)$ its unique least resolved
    tree, and $\wt\coloneqq \tTps$ the extremal event labeling of the
    augmented tree $(\aug(T^*),\sigma)$.  As shown in the proof of
    Thm.~\ref{thm:MAIN}, every edge $xy$ of of the symmetric part $G$ that is
    not a hug-edge satisfies $xy\in E(G_{\Theta})$ and therefore
    $\wt(u)=\SPEC$, where $u\coloneqq\lca_{\aug(T^*)}(x,y)$.  
    Lemma~\ref{lem:T-fp-no-mu} implies that $e$ is not $(\aug(T^*),\sigma)$-\fp 
    and thus, in particular, not \ufp.  That is, all edges in
    $(G_{\Theta},\sigma)=(G_{\NH},\sigma)$ are non-\ufp edges. Moreover,
    Thm.~\ref{thm:good_ugly_or_hourglass} implies that all hug-edges in
    $E(G) \setminus E(G_{\NH})$ are \ufp.  Since $(G_{\NH},\sigma)$ does not
    contain \ufp edges, all \ufp edges must also be hug-edges, which
    completes the proof.
  \end{proof}
  
  We next show that $\NH(\G,\sigma)$ can be computed in polynomial time. In
  fact, the effort is dominated by computing the least resolved tree
  $(T^*,\sigma)$ for a given BMG.
  \begin{theorem}
    For a given BMG $(\G,\sigma)$, the set of all \ufp edges can be computed
    in $O(|L|^3 |\mathscr{S}|)$ time, where $L=V(\G)$ and
    $\mathscr{S} = \sigma(L(T))$ is the set of species under consideration.
    \label{thm:time}
  \end{theorem}
  \begin{proof}
    Given a BMG $(\G,\sigma)$, its least resolved tree $(T^*,\sigma)$ can be
    computed in $O(|L|^3 |\mathscr{S}|)$ time (cf.\ Thm.~\ref{thm:LRT} and
    \cite[Sec.~5]{Geiss:19a}).  The augmented tree $(\aug(T^*),\sigma)$ can
    be obtained from $(T^*,\sigma)$ in $O(|L|^2 |\mathscr{S}|)$ time
    according to Lemma~\ref{lem:augalg}. The extremal event labeling $\wt$
    can be obtained from the connectivity information on the
    $\CIG_{\aug(T^*)}(u)$ in linear time.  Computing
    $(\Theta(\aug(T^*),\wt),\sigma) = \NH(\G,\sigma)$ then only requires
    evaluation of $\lca_{\aug(T^*)}(x,y)$, which can be achieved in
    polynomial time in $O(|L|^2)$ as described in \cite[Sec.~5]{Geiss:19a}).
  \end{proof}
  
  \subsection{Additional unidentified false-positives}
  \label{ssec:APP:further-fp}
  
  For an event-labeled, leaf-colored tree $(T,t,\sigma)$, we consider the
  triple set
  \begin{equation}
  \begin{split}
  \mathfrak{S}(T,t,\sigma) = \{\sigma(a)\sigma(b)|\sigma(c) \colon
  & ab|c\le T;\;
  t(\lca_T(a,b,c))=\SPEC; \\
  &\sigma(a),\sigma(b),\sigma(c) \text{ pairwise distinct} \}.
  \end{split}
  \end{equation}
  Moreover, we will need the following characterization of biologically 
  plausible event-labeled gene trees:
  \begin{theorem}{\cite{HernandezRosales:12a,Hellmuth:17}}
    There is a species tree $S$ together with a reconciliation map $\mu$ from
    $(T,t,\sigma)$ to $S$ such that $t_{\mu}=t$ if and only if
    $\mathfrak{S} (T,t,\sigma)$ is compatible. In this case, every species
    tree $S$ that displays $\mathfrak{S}(T,t,\sigma)$ can be reconciled with
    $(T,t,\sigma)$.  Moreover, there is a polynomial-time algorithm that
    determines whether a species tree for $(T,t, \sigma)$ exists, and if so,
    returns a species tree $S$ together with a reconciliation map
    $\mu:T\to S$.
    \label{thm:SpeciesTriplets}  
  \end{theorem}
  
  Throughout this section we are only concerned with the extremal event 
  labeling $\tTps$ of the augmented trees $(\aug(T^*),\sigma)$ of 
  least resolved trees $(T^*,\sigma)$. For brevity, we simply 
  write~$\wt$. For a BMG $(\G,\sigma)$, we consider the set of trees
  \begin{equation}
  \mathfrak{T} \coloneqq
  \left\{ (T,t,\sigma) \;|\; \NH(\G,\sigma) = (\Theta(T,t),\sigma)\right\}.
  \label{eq:trees}
  \end{equation}
  
  An orthology relation $\NH(\G,\sigma)$ obtained from a BMG $(\G,\sigma)$ by
  removing all of its \ufp edges is biologically feasible only if there is an
  event-labeled gene tree $(T,t,\sigma)\in \mathfrak{T}$ that can be
  reconciled with some species tree.  To show that this condition can be
  tested in polynomial time, we first need a technical result.
  
  \begin{lemma}
    Let $(\G,\sigma)$ be a BMG with LRT $(T^*,\sigma)$, and let
    $\mathfrak{T}$ be be given by Eq.~(\ref{eq:trees}).  If $ab|c$ is
    displayed by $\aug(T^*)$ and $\wt(\lca_{\aug(T^*)}(a,b,c)) = \SPEC$, then
    $ab|c$ is also displayed by every tree $(T,t,\sigma)\in\mathfrak{T}$ and
    $t(\lca_{T}(a,b,c))=\SPEC$.
    \label{lem:aug_tree_spec_triples}
  \end{lemma}
  \begin{proof}
    Suppose that $ab|c$ is displayed by $\aug(T^*)$ and
    $\wt(\lca_{\aug(T^*)}(a,b,c)) = \SPEC$.  Thm.~\ref{thm:MAIN} implies
    $(\Theta(\aug(T^*),\wt),\sigma)=\NH(\G,\sigma)$. Thus $\NH(\G,\sigma)$ is
    a cograph by Thm.~\ref{thm:ortho-cograph}.  Let $(T',t',\sigma)$ be a
    least resolved tree for the cograph $\NH(\G,\sigma)$. Clearly,
    $(T',t',\sigma)\in \mathfrak{T}$.  This tree is unique and any other tree
    in $\mathfrak{T}$ must be a refinement of $(T',t',\sigma)$
    \cite{CORNEIL:81,Boecker:98}.  We proceed with showing that (1)
    $t'(\lca_{T'}(a,b,c)) = \SPEC$ and (2) $ab|c$ is displayed by $T'$.
    
    In order to show~(1), assume for contradiction that
    $t'(\lca_{T'}(a,b,c)) = \DUPL$ and note that
    $(T',t',\sigma)\in \mathfrak{T}$ implies
    $\NH(\G,\sigma) = (\Theta(T',t'),\sigma)$. Since
    $\wt(\lca_{\aug(T^*)}(a,b,c)) = \SPEC$ and $ab|c \le \aug(T^*)$, the
    induced subgraph of $\NH(\G,\sigma)$ on $\{a,b,c\}$ contains at least the
    two edges $ac$ and $bc$. However, if $t'(\lca_{T'}(a,b,c)) = \DUPL$, then
    this induced subgraph can contain at most one edge; a
    contradiction. Hence, $t'(\lca_{T'}(a,b,c)) = \SPEC$.
    
    Next, we show~(2). Since $\aug(T^*)$ displays $ab|c$ and $T'$ is obtained
    from $\aug(T^*)$ by a series of edge contractions, $T'$ can neither
    display $ac|b$ nor $bc|a$, thus either $ab|c\le T'$ or
    $\lca_{T'}(a,b)=\lca_{T'}(a,b,c)$. By
    Lemma~\ref{lem:augmented_tree_no_adj_spec}, $(\aug(T^*),\wt)$ does not
    contain adjacent (consecutive) speciation vertices.  Therefore and since
    $\aug(T^*)$ displays $ab|c$, the path from $\lca_{\aug(T^*)}(a,b,c)$ to
    $\lca_{\aug(T^*)}(a,b)$ in $\aug(T^*)$ must contain at least one
    duplication vertex.  Since $T'$ can be obtained from $\aug(T^*)$ by
    contracting all edges $uv$ in $\aug(T^*)$ with $\wt(u)=\wt(v)$
    \cite{CORNEIL:81,Boecker:98}, the path from $\lca_{T'}(a,b,c)$ to
    $\lca_{T'}(a,b)$ in $T'$ must contain at least one duplication
    vertex. Together with $t'(\lca_{T'}(a,b,c)) = \SPEC$ this implies
    $\lca_{T'}(a,b)\ne\lca_{T'}(a,b,c)$, and hence, $ab|c$ is displayed by
    $T'$.
    
    Since every tree $(T,t,\sigma)\in \mathfrak{T}$ is a refinement of
    $(T',t',\sigma)$, the triple $ab|c$ is also displayed by~$T$.  Finally,
    since $\NH(\G,\sigma) = (\Theta(T,t),\sigma)$ for every tree
    $(T,t,\sigma)\in \mathfrak{T}$, we can re-use the arguments from the
    proof of Statement (1) to conclude that $t(\lca_{T}(a,b,c))=\SPEC$.  
  \end{proof}
  
  \begin{lemma}
    Let $(\G,\sigma)$ be a BMG with LRT $(T^*,\sigma)$ and let $\mathfrak{T}$
    be given by Eq.~(\ref{eq:trees}).  Then, the following statements are
    equivalent:
    \begin{itemize}
      \item[(1)] There is no reconciliation map $\mu$ from
      $(\aug(T^*),\wt,\sigma)$ to any species tree such that $t_{\mu}=\wt$.
      \item[(2)] For all trees $(T,t,\sigma)$ in $\mathfrak{T}$ there is no
      reconciliation map $\mu$ from $(T,t,\sigma)$ to any species tree such
      that $t_{\mu}=t$.
    \end{itemize}
    In particular, Condition (1) can be verified in polynomial time.
    \label{lem:reconc-aug-all}
  \end{lemma}
  \begin{proof}
    First note that $(\aug(T^*),\wt,\sigma)\in \mathfrak{T}$ since, by
    Thm.~\ref{thm:MAIN}, $(\Theta(\aug(T^*),\wt),\sigma) =\NH(\G,\sigma)$.
    Hence, Statement (2) implies (1).
    
    For the converse, let $ab|c$ be displayed by $\aug(T^*)$ where
    $\sigma(a)=A$, $\sigma(b)=B$, $\sigma(c)=C$ are pairwise distinct, and
    $\wt(\lca_{\aug(T^*)}(a,b,c)) = \SPEC$. By definition,
    $AB|C \in \mathfrak{S}(\aug(T^*),\wt,\sigma)$.
    Lemma~\ref{lem:aug_tree_spec_triples} implies that $ab|c$ is also
    displayed by every tree $(T,t,\sigma)\in\mathfrak{T}$ and
    $t(\lca_{T}(a,b,c))=\SPEC$. Therefore, we have
    $\mathfrak{S}(\aug(T^*),\wt,\sigma) \subseteq \mathfrak{S}(T,t,\sigma)$
    for all $(T,t,\sigma)\in\mathfrak{T}$.  Now suppose that Condition~(1)
    holds. Then, by Thm.~\ref{thm:SpeciesTriplets},
    $\mathfrak{S}(\aug(T^*),\wt,\sigma)$ is incompatible.  Thus,
    $\mathfrak{S}(T,t,\sigma)$ must be incompatible as well for every tree
    $(T,t,\sigma)\in\mathfrak{T}$.  Together with
    Thm.~\ref{thm:SpeciesTriplets}, this implies Condition~(2).
    
    Using the arguments in the proof of Thm.~\ref{thm:time} and
    Thm.~\ref{thm:SpeciesTriplets} we find that Condition~(1) can be
    verified in polynomial time by checking whether
    $\mathfrak{S}(\aug(T^*),\wt,\sigma)$ is incompatible.  
  \end{proof}
  It is possible, therefore to check in polynomial time whether the cograph
  $\NH(\G,\sigma)$ is a biologically feasible orthology relation for
  $(\G,\sigma)$ or whether $\NH(\G,\sigma)$ contains further false-positive
  edges.
  
  Now consider again a true evolutionary scenario
  $(\widetilde{T},\widetilde{t},\sigma)$. While $\widetilde{T}$ always
  displays the LRT $(T^*,\sigma)$ of the BMG $\G(\widetilde{T},\sigma)$, it
  does not necessarily display the augmented tree $\aug(T^*)$. As an example
  consider the scenario in Fig.~\ref{fig:messy_vertices-1}. Augmenting the
  only multifurcation in this case further resolves the root of $T^*$ and
  thus yields a tree that is not displayed by $\widetilde{T}$. It is
  interesting to ask, therefore, whether there are situations in which
  $\widetilde{T}$ does display $\aug(T^*)$.
  
  \begin{lemma}\label{lem:T-displays-aug-tree}
    Let $(T,t,\sigma)$ be an event-labeled tree explaining the BMG
    $(\G,\sigma)$, and let $(T^*,\sigma)$ be the least resolved tree of
    $(\G,\sigma)$.  If $(\Theta(T,t),\sigma) = \NH(\G,\sigma)$, then
    $\aug(T^*)$ is displayed by $T$.
  \end{lemma}
  \begin{proof}
    Let $\mathfrak{T}$ be the set of trees corresponding to $(\G,\sigma)$ as
    given by Eq.~(\ref{eq:trees}).  First note that
    $(T,t,\sigma)\in\mathfrak{T}$ and that $(T^*,\sigma)$ is displayed by
    $(T,\sigma)$ \cite[cf.][Thm.~8]{Geiss:19a}.  Now consider the set
    $r(\aug(T^*))$ of all triples displayed by $\aug(T^*)$.  For any triple
    $ab|c\in r(\aug(T^*))$, there are exactly two cases: (a) $\wt(u)=\SPEC$
    and (b) $\wt(u)=\DUPL$, where $u\coloneqq\lca_{\aug(T^*)}(a,b,c)$.
    
    In Case~(a), Lemma~\ref{lem:aug_tree_spec_triples} together with
    $(T,t,\sigma)\in\mathfrak{T}$ immediately implies that $ab|c$ is also
    displayed by~$T$.
    
    In Case~(b), we have $\wt(u)=\DUPL$.  Consider the child
    $v\in\child_{\aug(T^*)}(u)$ with $a,b\prec_{\aug(T^*)}v$.  Assume, for
    contradiction, that $v$ is not a vertex in $T^*$, i.e., it was newly
    created by augmenting a vertex $u'$. We have $u'=u$ by
    Lemma~\ref{lem:augment-once} since $u'$ cannot be (non-trivially)
    augmented any further.  Since $\aug(T^*)$ does not depend on the order of
    augmentation steps, we may assume w.l.o.g.\ that $v$ was created in the
    first augmentation step; resulting in the augmented tree $T_{u}$.
    Def.~\ref{def:augmenting} implies that $\CIG_{T}(u)$ is disconnected.
    Together with Lemma~\ref{lem:augmenting_color_disjoint}, this implies
    $\sigma(L(T_u(w_1)))\cap\sigma(L(T_u(w_2)))=\emptyset$ for any two
    distinct vertices $w_1,w_2\in\child_{T_u}(u)$.  This must still hold for
    $(\aug(T^*),\sigma)$ since the edges $uw$, where $w\in\child_{T_u}(u)$
    correspond to the vertices that have been newly introduced in the first
    augmentation step, do not change in any subsequent augmentation due to
    Lemma~\ref{lem:augment-once}. The definition of the extremal event
    labeling now implies $\wt(u)=\SPEC$; a contradiction.  Therefore, we
    conclude that $v$ is a vertex in $T^*$, and in particular,
    $a,b\in L(T^*(v))$ and $c\notin L(T^*(v))$, which in turn implies that
    $ab|c$ is displayed by $T^*$. From $T^*\le T$ we finally conclude that
    $T$ also displays $ab|c$. Denoting by $r(T)$ the set of all triples
    displayed by $T$ we therefore have $r(\aug(T^*))\subseteq r(T)$. Finally,
    we apply Thm.~1 of \citet{BS:95} to conclude that $\aug(T^*)$ is
    displayed by $T$.  
  \end{proof}
  
  \section{Quartets, hourglasses, and the structure of reciprocal best match
    graphs}
  \label{APP:sect:leftovers}
  
  \subsection{Hourglass-free BMGs}
  \label{APP:ssec:hourglass-free}
  \begin{definition}\label{def:hourglass-free}
    A BMG $(\G,\sigma)$ is \emph{hourglass-free} if it does not contain an
    hourglass as an induced subgraph.
  \end{definition}
  In particular, an hourglass-free BMG also does not contain an hourglass
  chain.  We will need the following technical result
  \begin{lemma}\label{lem:hourglass-color-sets}
    Let $(\G,\sigma)$ be a BMG explained by $(T,\sigma)$.
    Then $(\G,\sigma)$ has an hourglass $[xy \hourglass x'y']$ as an induced 
    subgraph if and only if there is a vertex $u\in V^0(T)$ with
    distinct children $v_1$, $v_2$, and $v_3$ and two distinct 
    colors $r$ and $s$ satisfying
    \begin{enumerate}
      \item $r\in\sigma(L(T(v_1)))$, 
      $r,s\in\sigma(L(T(v_2)))$, 
      and $s\in\sigma(L(T(v_3)))$, and
      \item $s\notin\sigma(L(T(v_1)))$, and 
      $r\notin\sigma(L(T(v_3)))$.
    \end{enumerate}
  \end{lemma}
  \begin{proof}
    First assume that $(\G,\sigma)$ contains the hourglass
    $[xy \hourglass x'y']$ as an induced subgraph. Then by
    Lemma~\ref{lem:hourglass}, $(T,\sigma)$ contains a vertex $u\in V^0(T)$
    with three distinct children $v_1$, $v_2$, and $v_3$ such that
    $x\preceq_T v_1$, $\lca_T(x',y')\preceq_T v_2$ and $y\preceq_T v_3$.
    Putting $r\coloneqq\sigma(x)=\sigma(x')$ and
    $s\coloneqq\sigma(y)=\sigma(y')$ immediately implies Condition~(1).  Now,
    assume for contradiction that Condition~(2) is violated and thus
    $s\in\sigma(L(T(v_1)))$ or $r\in\sigma(L(T(v_3)))$.  If
    $s \in\sigma(L(T(v_1)))$, then there is a leaf $y''\prec_T v_1$ with
    $\sigma(y'') =s$. In this case, however,
    $\lca(x,y'')\preceq_T v_1 \prec_T u = \lca_T(x,y')$ implies that $(x,y')$
    cannot be an arc in $(\G,\sigma)$; a contradiction to
    $[xy \hourglass x'y']$ being an hourglass.  By similar arguments,
    $r\in\sigma(L(T(v_3)))$ is not possible.  Therefore, Condition~(2) must
    be satisfied.
    
    Now assume that there is a vertex $u\in V^0(T)$ with pairwise distinct
    children $v_1$, $v_2$, and $v_3$ and two distinct colors $r$ and $s$
    satisfying Conditions~(1) and~(2).  It is now straightforward to see that
    $(\G,\sigma)$ contains an hourglass: Condition~(1) immediately implies
    the existence of vertices $x\in L[r]\cap L(T(v_1))$ and
    $y\in L[s]\cap L(T(v_3))$.  Moreover, $r,s\in\sigma(L(T(v_2)))$ together
    with Lemma~\ref{lem:exEdge} imply that there is an edge $x'y'$ in
    $(\G,\sigma)$ with $x'\in L[r]\cap L(T(v_2))$ and
    $y'\in L[s]\cap L(T(v_2))$.  Clearly, the vertices in $\{x,x',y,y'\}$ are
    pairwise distinct.  By Condition~(2) and the location of the four leaves,
    we obtain the arcs $(x,y')$, $(x,y)$, $(y,x')$, and $(y,x)$, and thus, in
    particular the edge $xy$. Since $T(v_2)$ contains both colors $r$ and
    $s$, we can furthermore conclude that $(x',y)$ and $(y',x)$ are not arcs
    in $(\G,\sigma)$.  In summary, the subgraph of $(\G,\sigma)$ induced by
    the set $\{x,x',y,y'\}$ is an hourglass $[xy \hourglass x'y']$.  
  \end{proof}
  
  In the following a tree $(T,\sigma)$ is called \emph{refinable} if there is a
  proper refinement $(T',\sigma)$ of $(T,\sigma)$, i.e., $T\le T'$ and $T\ne 
  T'$, 
  such that $\G(T',\sigma)=\G(T,\sigma)$. Otherwise, $(T,\sigma)$ is
  \emph{non-refinable}. An inner vertex of a tree is \emph{non-refinable} if
  it cannot be refined without changing the best match graph induced by the
  tree.
  
  Clearly, for every BMG $(\G,\sigma)$, there is a tree that has the
  maximum number of vertices among all trees that explain $(\G,\sigma)$ and
  thus, a tree that cannot be further resolved.  Hence, every BMG can be
  explained by a non-refinable tree.  We will need the following useful
  property of non-refinable vertices:
  
  \begin{lemma}\label{lem:non-refinable-vertex}
    Let $(\G,\sigma)$ be a BMG explained by a tree $(T,\sigma)$, and let
    $u\in V^0(T)$ be a non-refinable vertex of $(T,\sigma)$.  Then, for any
    proper subset $C\subsetneq\child_{T}(u)$ with $|C|\ge 2$, there are
    two distinct vertices $v,v'\in C$, a vertex
    $v''\in\child_{T}(u)\setminus C$, and two vertices $a\preceq_{T} v$ and
    $b\preceq_{T} v'$ such that $(a,b)\in E(\G)$ and
    $\sigma(b)\in\sigma(L(T(v'')))$.
  \end{lemma}
  \begin{proof}
    First note that the statement is trivially true if $u$ is binary, since
    then there is no proper subset $C\subsetneq\child_{T}(u)$ such that
    $|C|\ge 2$. Thus, assume $|\child_{T}(u)|\ge 3$ in the following.
    
    We refine $(T,\sigma)$ at vertex $u$ as follows: Take an arbitrary subset
    $C\subsetneq\child_{T}(u)$ such that $|C|\ge 2$ (which exists since
    $|\child_{T}(u)|\ge 3$) and place all vertices in $C$ as the children of
    a new vertex $w$, and connect $w$ as a child of $u$.  Since $u$ is a
    non-refinable vertex of $(T,\sigma)$, this refinement leads to a tree
    $(T',\sigma)$ that does not explain $(\G,\sigma)$, and therefore, the
    inner edge $uw$ must be non-redundant w.r.t.\ $\G(T',\sigma)$.  By
    Lemma~\ref{lem:redundant_edges}, there must be an arc $(a,b)$ in
    $\G(T',\sigma)$ such that $\lca_{T'}(a,b)=w$ and
    $\sigma(b)\in \sigma(L(T'(u))\setminus L(T'(w)))$.  In particular,
    $\lca_{T'}(a,b)=w$ implies that $a\preceq_{T} v$ and $b\preceq_{T} v'$
    for two distinct vertices $v,v'\in\child_{T'}(w)=C$.  Note that
    $(T,\sigma)$ can be obtained from $(T',\sigma)$ by contraction of the
    edge $uw$. Hence, we can apply Lemma~\ref{lem:contract-subgraph} to
    conclude that $\G(T',\sigma)\subseteq (\G,\sigma)$. Therefore,
    $(a,b)\in E(\G)$.  Taking the latter arguments together, for any subset
    $C\subsetneq\child_{T}(u)$ with $|C|\ge 2$, there are vertices
    $a\preceq_{T} v$ and $b\preceq_{T} v'$ with distinct $v,v'\in C$ such
    that $(a,b)\in E(\G)$ and $\sigma(b)\in\sigma(L(T(v'')))$ for some
    $v''\in\child_{T}(u)\setminus C$.  
  \end{proof}

  \begin{proposition}\label{prop:binary-iff-hourglass-free}
    A BMG $(\G,\sigma)$ can be explained by a binary tree if and only if it is 
    hourglass-free.
  \end{proposition}
  \begin{proof}
    If the BMG $(\G,\sigma)$ can be explained by a binary tree, it must be
    hourglass-free as a consequence of Lemma~\ref{lem:hourglass}. To prove the 
    converse, we assume, for contradiction, that $(\G,\sigma)$ is 
    hourglass-free 
    and cannot be explained by any binary tree. Then there is a non-refinable
    non-binary tree $(T,\sigma)$ that explains $(\G,\sigma)$. By 
    construction,
    furthermore, $T$ contains a non-binary vertex $u\in V^0(T)$, which by
    assumption is non-refinable.
    
    The key device for our proof are pairs $(\mathscr{M},\mathscr{N})$ where
    $\mathscr{M}\coloneqq\{v_1,\dots,v_k\}$ is an ordered set of $k\geq 2$
    pairwise distinct children of $u$ and
    $\mathscr{N}\coloneqq\{c_1,\dots,c_{k-1}\}$ is an ordered set of $k-1$
    pairwise distinct colors. We call $(\mathscr{M},\mathscr{N})$ an
    \emph{hourglass-free pair (hf-pair) of order $k$} for $u$ if the
    following conditions are satisfied:
    \begin{description}[nolistsep]
      \item[(i)] For all $c_i\in \mathscr{N}$ we have
      $c_i\in\sigma(L(T(v_j)))$, $i\le j\le k-1$,
      \item[(ii)] For all $c_i\in \mathscr{N}$ we have
      $c_i\notin\sigma(L(T(v_j)))$, $1\le j<i$, and
      \item[(iii)] $\mathscr{N} \subseteq \sigma(L(T(v_k)))$.
    \end{description}
    If $(\mathscr{M},\mathscr{N})$ is an hf-pair of order $k$, then Condition
    (i) implies by construction that
    $\mathscr{N} \subseteq\sigma(L(T(v_{k-1})))$.  Therefore,
    $(\mathscr{M}'=(v_1,\dots,v_{k},v_{k-1}),\mathscr{N})$ is also an hf-pair
    where $\mathscr{M}'$ is obtained from $\mathscr{M}$ by exchanging
    the positions of its last two elements. Hf-pairs and the following
    arguments are illustrated in Fig.~\ref{fig:hourglassfree_nonbinary}.
    In order to obtain the desired contradiction, we show by induction
    that the children of the non-binary, non-refinable vertex $u$ harbor
    hf-pairs of arbitrary large order $k$.
    
    \noindent\textbf{Base case.} There is
    an hf-pair $(\mathscr{M},\mathscr{N})$ of order $2$ for~$u$.
    \par\noindent\textit{Proof of Claim.}
    Consider an arbitrary subset
    $\{v,v'\}\subsetneq\child_{T}(u)$ consisting of two distinct children $v$
    and $v'$ of the non-binary vertex $u$.  By
    Lemma~\ref{lem:non-refinable-vertex} and since $u$ is non-refinable,
    there are vertices $a\preceq_{T} v$ and $b\preceq_{T} v'$ such that
    w.l.o.g.\ $(a,b)\in E(\G)$ and $\sigma(b)\in\sigma(L(T(v'')))$ for some
    $v''\in\child_{T}(u)\setminus \{v,v'\}$.  The latter implies that there
    is a vertex $b'\preceq_{T}v''$ of color $\sigma(b)$.  Clearly, $b$ and
    $b'$ are distinct and the color $\sigma(b)$ is also present in the
    subtree $T(v')$.  Thus we can set
    $\mathscr{M}\coloneqq(v_1\coloneqq v', v_2\coloneqq v'')$ and
    $\mathscr{N}\coloneqq(c_1\coloneqq \sigma(b))$.  It is an easy task to
    verify that $(\mathscr{M},\mathscr{N})$ satisfies
    Conditions~(i)--(iii).
    \hfill$\triangleleft$
    
    \noindent\textbf{Induction step.}
    The existence of an hf-pair of order $k$ implies the 
    existence of an hf-pair of order $k+1$ for~$u$.
    \par\noindent\textit{Proof of Claim.}
    Let $(\mathscr{M}= (v_1,\dots,v_k),\mathscr{N}= (c_1,\dots,c_{k-1}))$ be
    an hf-pair, and consider the set
    $\{v_{k-1},v_k\}\subsetneq\child_{T}(u)$.  By
    Lemma~\ref{lem:non-refinable-vertex} and since $u$ is non-refinable,
    there are again vertices $a\preceq_{T} v$ and $b\preceq_{T} v'$ for
    distinct $v,v'\in\{v_{k-1},v_k\}$ such that $(a,b)\in E(\G)$ and
    $\sigma(b)\in\sigma(L(T(v'')))$ for some
    $v''\in\child_{T}(u)\setminus \{v_{k-1},v_k\}$.  We can assume w.l.o.g.\
    that $a\preceq_{T} v=v_{k-1}$ and $b\preceq_{T} v'=v_k$ since otherwise
    we can simply swap $v_{k-1}$ and $v_{k}$ in the ordered set $\mathscr{M}$
    as argued above.  Since $(a,b)$ is an arc in $(\G,\sigma)$ and
    $\lca_{T}(a,b)=u$, the color $\sigma(b)$ cannot be present in the subtree
    $T(v_{k-1})$.  Since $\mathscr{N}\subseteq \sigma(L(T(v_{k-1})))$ and
    $\sigma(b)\notin \sigma(L(T(v_{k-1})))$, we conclude that
    $\sigma(b)\notin \mathscr{N}$.
    
    We continue to show that $v''$ is distinct from all elements in
    $\mathscr{M}$.  Clearly, in the case $k=2$, $v''$ is distinct from all
    elements in $\mathscr{M} = \{v_1,v_2\}=\{v,v'\}$ by
    construction. Now let $k>2$ and assume, for contradiction, that there
    is a vertex $v_j\in\{v_1,\dots,v_{k-2}\}$ such that
    $\sigma(b)\in\sigma(L(T(v_j)))$.  In this case, $j<k-1$ and
    Condition~(ii) imply that $c_{k-1}\notin\sigma(L(T(v_j)))$.  In addition,
    we have $c_{k-1}\in\sigma(L(T(v_{k-1})))$ and
    $c_{k-1}\in\sigma(L(T(v_{k})))$ by Conditions~(i) and~(iii),
    respectively. Recall that $v'=v_k$. In summary, we obtain three distinct 
    vertices
    $v_j,v_k,v_{k-1}$ and two distinct colors $\sigma(b)$ and $c_{k-1}$
    satisfying Conditions~(1) and~(2) in
    Lemma~\ref{lem:hourglass-color-sets}, which implies that
    $(\G,\sigma)$ contains an hourglass; a contradiction.  Hence,
    $\sigma(b)\notin\sigma(L(T(v_j)))$ for all $j\in\{1,\dots, k-2\}$.  This
    implies that $v''$ is distinct from $v_1,\dots,v_{k-2}$.  Moreover, by
    construction, $v''$ is distinct from $v_{k-1}$ and $v_k$.  In summary,
    $v''$ is therefore distinct from all elements in $\mathscr{M}$.
    
    Consider now the pair
    $(\mathscr{M}' \coloneqq(v_1,\dots,v_k,v_{k+1}\coloneqq v''),
    \mathscr{N}'\coloneqq (c_1,\dots,c_{k-1},c_k\coloneqq \sigma(b)))$.
    Since $(\mathscr{M},\mathscr{N})$ is an hf-pair, and since, by
    construction, $c_k=\sigma(b)\notin \sigma(L(T(v_{j})))$ for
    $1\le j\le k-1$ and $c_k=\sigma(b)\in \sigma(L(T(v_{k})))$, we can
    immediately conclude that Conditions~(i) and~(ii) are satisfied for
    $(\mathscr{M}',\mathscr{N}')$.  It remains to show that Condition~(iii)
    is satisfied as well, i.e., $c_i\in\sigma(L(T(v_{k+1})))$ for all
    $1\le i\le k$.  By construction, we have $c_k\in\sigma(L(T(v_{k+1})))$.
    Now assume that $c_i\notin\sigma(L(T(v_{k+1})))$ for some
    $1\le i \le k-1$.  We have $c_i\in\sigma(L(T(v_{k-1})))$ and
    $c_i,c_k\in\sigma(L(T(v_{k})))$ by Condition~(i), and
    $c_k\notin\sigma(L(T(v_{k-1})))$ by Condition~(ii).  Taken together, we
    obtain three distinct vertices $v_{k-1},v_{k},v_{k+1}$ and two distinct
    colors $c_i$ and $c_k$ satisfying Conditions~(1) and~(2) in
    Lemma~\ref{lem:hourglass-color-sets}, which implies that
    $(\G,\sigma)$ contains an hourglass; a contradiction.  Therefore,
    Condition~(iii) must be satisfied as well, and
    $(\mathscr{M}',\mathscr{N}')$ is an hf-pair of order $k+1$.
    \hfill$\triangleleft$
    
    Repeated application of the induction step implies that children of a
    non-refinable non-binary vertex $u$ in a non-refinable tree $(T,\sigma)$
    explaining an hourglass-free BMG harbor an hf-pair of arbitrary
    order. This is of course impossible since $G$ is finite, i.e, no such
    vertex $u$ can exist. Therefore, every hourglass-free BMG $(\G,\sigma)$
    can be explained by a binary tree.  
  \end{proof}
  
  \begin{figure}[t]
    \begin{center}
      \includegraphics[width=0.85\textwidth]{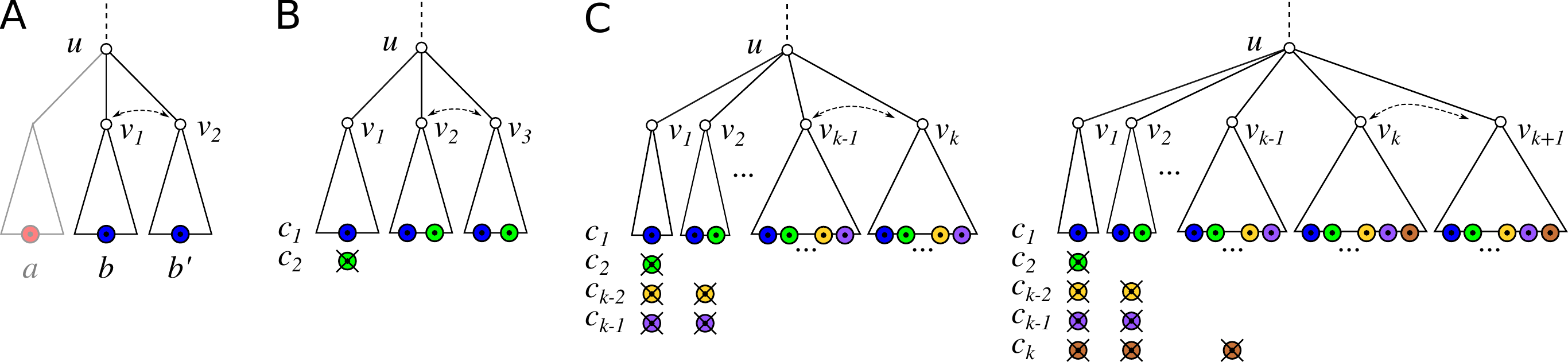}
    \end{center}
    \caption{Illustration of the induction argument in the proof of
      Prop.~\ref{prop:binary-iff-hourglass-free}. (A) Base case: an
      hourglass-free pair (hf-pair)
      $(\mathscr{M}=\{v_1,v_2\},\mathscr{N}=\{\sigma(b)\})$ of order
      $2$. Note that vertex $a$ is only required to show the existence of
      $(\mathscr{M},\mathscr{N})$.  (B) An hf-pair
      $(\mathscr{M}=\{v_1,v_2,v_3\},\mathscr{N}=\{c_1,c_2\})$ of order $3$.
      (C) Induction step: The existence of an hf-pair of order $k$ implies
      the existence of an hf-pair of order $k+1$, and thus, an infinite
      number of children of $u$. This gives the desired contradiction in the
      proof of Prop.~\ref{prop:binary-iff-hourglass-free}.  The dashed arrow
      indicates the last two elements in the ordered set $\mathscr{M}$ of an
      hf-pair $(\mathscr{M},\mathscr{N})$ are interchangeable.}
    \label{fig:hourglassfree_nonbinary}
  \end{figure}
  
  Prop.~\ref{prop:binary-iff-hourglass-free} gives rise to a procedure for
  determining whether a BMG $(\G,\sigma)$ can be explained by a binary tree.
  We simply need to check whether $(\G,\sigma)$ is hourglass-free, a task
  that can be done trivially in $O(|E(\G)|^2)$ time by checking, for all
  pairs of edges $ab$ and $a'b'$ (in constant time), whether or not they
  induce an hourglass $[ab \hourglass a'b']$ or $[a'b' \hourglass ab]$,
  respectively.  Hence, we obtain
  \begin{corollary}
    It can be decided in polynomial time whether a BMG $(\G,\sigma)$ can be
    explained by a binary tree.
    \label{cor:binary-polytime}
  \end{corollary}
  It remains open, however, whether such a tree can be constructed
  efficiently.
  
  \citet{Geiss:19b} found that a certain type of colored 6-cycles is an
  important characteristic of RBMGs with a ``complicated'' structure that can
  only be explained by multifurcating trees. Let us write
  $\langle x_1 x_2\dots x_k\rangle$ for an induced cycle $C_k$ with edges
  $x_i x_{i+1}$, $1 \leq i \leq k-1$, and $x_k x_1$ in the symmetric part $G$
  of $\G$. We say that $(\G,\sigma)$ contains a \emph{hexagon} if the
  corresponding RBMG $(G,\sigma)$ contains an induced
  $C_6 = \langle x_1 x_2\dots x_6\rangle$ such that any three consecutive
  vertices of $C_6$ have pairwise distinct colors, i.e.,
  $\sigma(x_i)=\sigma(x_i+3)$, $1\leq i\leq 3$. Since hexagons contain $P_4$s
  and, by \cite[Lemma~32]{Geiss:19b}, any $P_4$ is either a good or a bad
  quartet, there are exactly two possible induced subgraphs spanned by a
  hexagon $C_6 = \langle x_1 x_2\dots x_6\rangle$, which are shown in
  Fig.~\ref{fig:C6}.  A graph $(\G,\sigma)$ is \emph{hexagon-free} if it does
  not contain a hexagon.
  
  \begin{figure}[t]
    \begin{center}
      \includegraphics[width=0.85\textwidth]{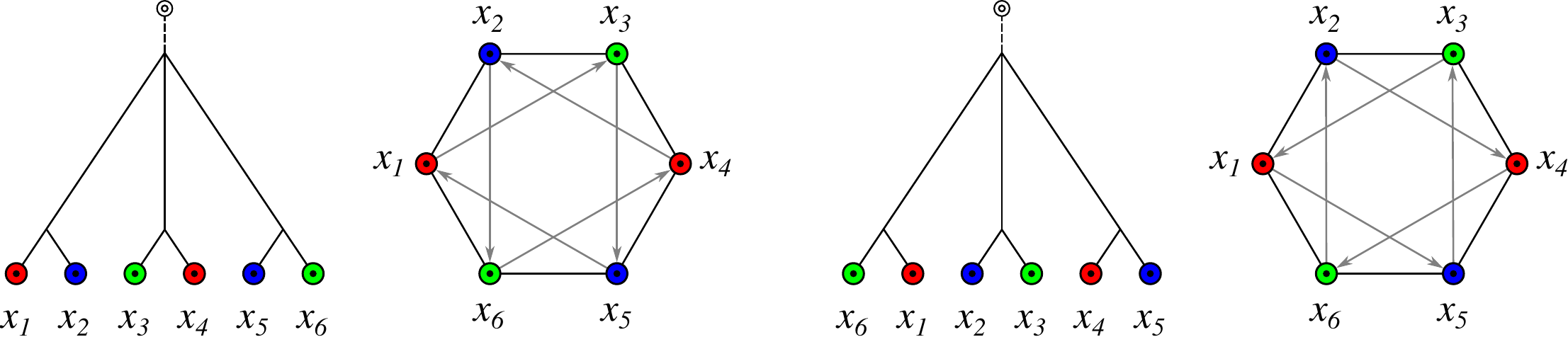}
    \end{center}
    \caption{Two examples of trees whose BMGs $\G(T,\sigma)$ contain a
      hexagon $\langle x_1x_2x_3x_4x_5x_6\rangle$. There are exactly two
      distinct possibilities for the placement of the non-symmetric arcs in
      the subgraph of the BMG induced by the hexagon.}
    \label{fig:C6}
  \end{figure}
  
  \begin{lemma}
    \label{lem:hour-hex}
    Every hourglass-free BMG $(\G,\sigma)$ is hexagon-free.
  \end{lemma}
  \begin{proof}
    By Prop.~\ref{prop:binary-iff-hourglass-free}, every hourglass-free BMG
    $(\G,\sigma)$ can be explained by a binary tree. Lemma~9 in
    \cite{Geiss:20a} implies that hexagons can only be explained by
    non-binary trees. Hence, $(\G,\sigma)$ must be hexagon-free.
  \end{proof}
  Clearly, the converse of Lemma~\ref{lem:hour-hex} is not always satisfied,
  since, by Obs.~\ref{obs:hourbmg}, an hourglass is a BMG without hexagons.
  
  A very useful observation in previous work is the fact that every 3-colored
  vertex induced subgraph of an RBMG $(G,\sigma)$ is again an RBMG
  \cite[Thm.~7]{Geiss:19b}. Furthermore, the connected components
  $(C,\sigma)$ of every 3-colored vertex induced subgraph of $(G,\sigma)$
  belong to precisely one of the three types \cite[Thm.~5]{Geiss:19b}:
  \begin{description}
    \item[\textbf{Type (A)}] $(C,\sigma)$ contains a $K_3$ on three colors
    but no induced $P_4$.
    \item[\textbf{Type (B)}] $(C,\sigma)$ contains an induced $P_4$ on three
    colors whose endpoints have the same color, but no induced 
    cycle $C_n$ on $n\geq 5$ vertices.
    \item[\textbf{Type (C)}] $(C,\sigma)$ contains a hexagon.
  \end{description}
  The graphs for which all such 3-colored connected components are of Type
  (A) are exactly the RBMGs that are cographs, or co-RBMGs for short
  \cite[Thm.~8 and Remark~2]{Geiss:19b}. Together with 
  Lemma~\ref{lem:hour-hex}, 
  this classification immediately implies
  \begin{corollary}\label{cor:hourglass-free}
    Let $(\G,\sigma)$ be an hourglass-free BMG.  Then its symmetric part
    $(G,\sigma)$ is either a co-RBMG or it contains an induced $P_4$ on three
    colors whose endpoints have the same color, but no induced cycle $C_n$ on
    $n\geq 5$ vertices.
  \end{corollary}
  
  Since all \ufp edges in an hourglass-free BMG are contained in quartets,
  we have
  \begin{corollary}\label{cor:hourglass-free-coRBMG}
    Let $(\G,\sigma)$ be an hourglass-free BMG. Then its symmetric part
    $(G,\sigma)$ is a co-RBMG if and only if there are no \ufp edges in
    $(\G,\sigma)$.
  \end{corollary}
  \begin{proof}
    Since $(G,\sigma)$ is a cograph, it contains no induced $P_4$s and thus,
    $(\G,\sigma)$ contains no good or ugly quartets.  By
    Thm.~\ref{thm:ufp-iff-hug}, all hug-edges are determined by hourglass
    chains and good or ugly quartets.  Since none of them is contained in
    $(\G,\sigma)$, it also does not contain \ufp edges. Conversely, suppose
    that $(\G,\sigma)$ contains no \ufp edges.  Then, by Thm.~\ref{thm:MAIN},
    $(G,\sigma) = \NH(\G,\sigma)$ is an orthology graph and thus, by
    Thm.~\ref{thm:ortho-cograph}, a cograph.  
  \end{proof}
  
  \subsection{\ufp edges in hourglass chains}
  \label{APP:ssec:hchain}
  
  The situation is much more complicated in the presence of hourglasses. We
  start by providing sufficient conditions for \ufp edges that are identified
  by hourglass chains.
  \begin{proposition}
    \label{prop:hourglass-ufp} 	
    Let
    $\mathfrak{H}=[x_1 y_1 \hourglass x'_1 y'_1],\dots,[x_k y_k \hourglass
    x'_k y'_k]$ be an hourglass chain in $(\G,\sigma)$, possibly with a left
    tail $z$ or a right tail $z'$.  Then, an edge in $\G$ is \ufp if it is
    contained in the set
    \begin{align*}
    F =
    & \{x_iy_j\mid 1\leq i \leq j \leq k\}
    \cup\{zz'\}
    \cup\{zy_{i}, x_iz', zy'_{i}, x'_{i}z' \mid 1 \leq i \leq k \}\\
    & \cup\{ x_{i}x_{j+1} \mid 1\le i < j < k \} 
    \cup \{ y_{i}y_{j+1} \mid 1\le i < j < k \} \\
    & \cup\{x'_1 y'_i, x'_1 y_i \mid 2 \leq i \leq k \}
    \cup\{x_i y'_k, x'_i y'_k \mid 1 \leq i \leq k-1 \} \\
    & \cup\{x'_1 z, x'_1 z', y'_k z, y'_k z'\}
    \end{align*} 
  \end{proposition}
  \begin{proof}
    Let $(T,\sigma)$ be an arbitrary tree that explains $(\G,\sigma)$.  By
    analogous arguments as in the proof of Lemma~\ref{lem:hourglass_chain_dupl} 
    and by Lemma~\ref{lem:hourglass_chain_tails}, there is a vertex $u\in 
    V^0(T)$
    with pairwise distinct children $v_0,v_1,\dots,v_k,v_{k+1}$ such that it
    holds $x_1\in L(T(v_0))$, $y_k\in L(T(v_{k+1}))$ and, for all
    $1\le i\le k$, we have $x'_i,y'_i\in L(T(v_i))$. Since $x_{i+1}=y'_i$ and
    $x'_{i+1}=y_i$ by definition of hourglass chains, it is an easy task to
    verify that for all edges $e=ab\in F$ the vertices $a$ and $b$ are
    located below distinct children of $u$ and thus, $\lca_T(a,b)=u$ for all
    such edges.  As argued in the proof of 
    Lemma~\ref{lem:hourglass_chain_dupl}, 
    we have $\sigma(L(T(v_0)))\cap\sigma(L(T(v_1)))\ne\emptyset$.  The latter
    arguments together with Lemma~\ref{lem:T-fp-no-mu} imply that every edge
    in $F$ is \ufp.  
  \end{proof}
  
  \begin{figure}[t]
    \begin{center}
      \includegraphics[width=0.85\textwidth]{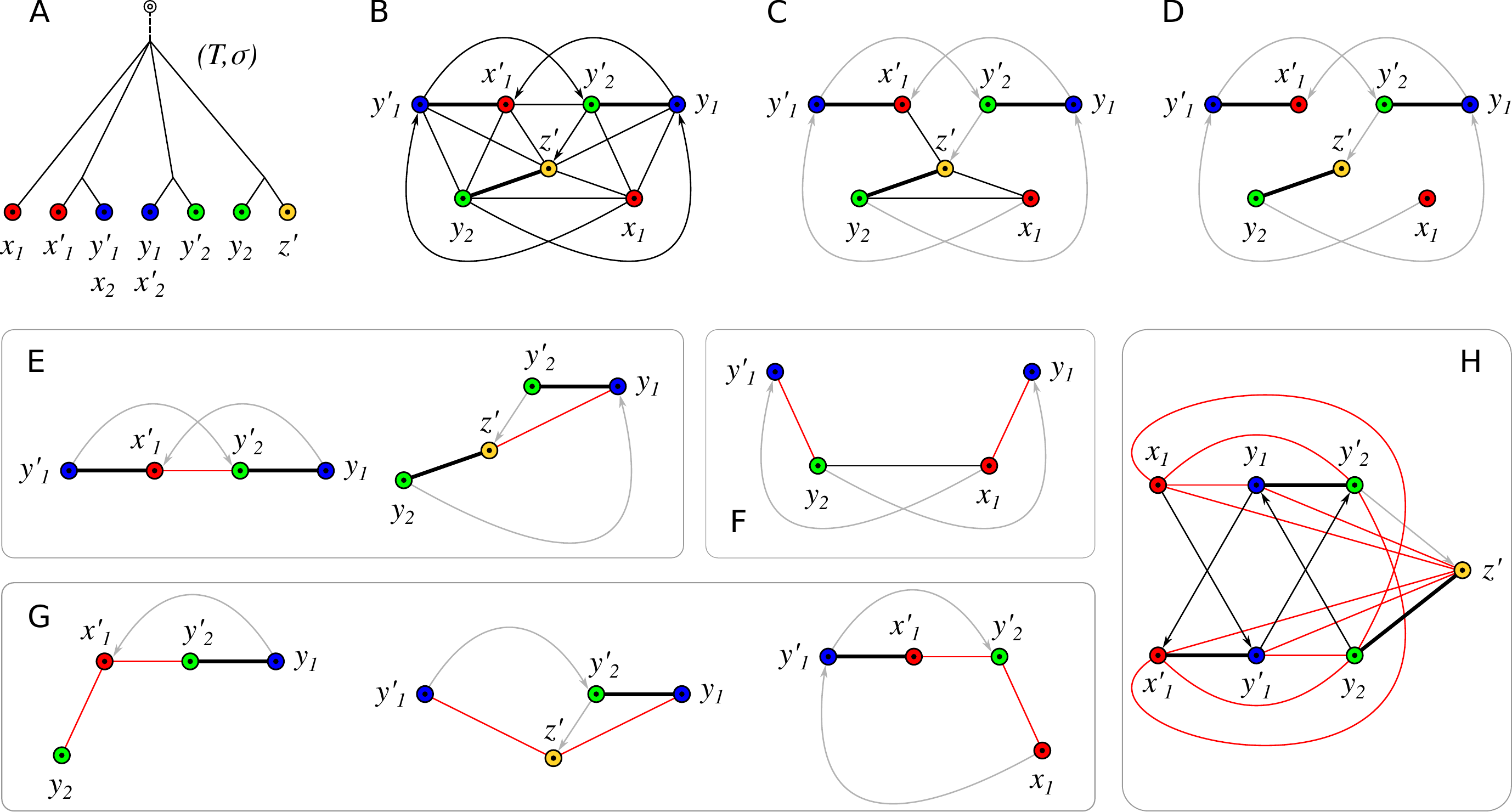}
    \end{center}
    \caption{The (non-binary) tree $(T,\sigma)$ in Panel (A) explains the BMG
      $(\G,\sigma)$ in Panel (B), which contains several induced $P_4$s and
      an hourglass chain of length $k=2$ with right tail $z'$. Edges that are
      not $(T,\sigma)$-\fp (and thus not \ufp) are shown as thick lines. Thin
      edges correspond to those that can be identified as \ufp by the
      subgraphs in (E--H), where they are highlighted in red.  (C) The graph
      after deletion of all edges that can be identified by good, bad and
      ugly quartets according to Props.~\ref{prop:good_quartet_middle_edge},
      \ref{prop:bad}, and~\ref{prop:ugly_quartet}. Note that it contains the
      induced $P_4$s $\langle y_1'x_1'z'y_2\rangle$ and
      $\langle y_1'x_1'z'x_1\rangle$, which were not induced subgraphs of the
      original BMG in (B).  Its symmetric part $(H,\sigma)$ differs from
      $\NH(\G,\sigma)$ (cf.\ Def.~\ref{def:non-hug-graph}) since it still
      contains \ufp edges.  (D) The BMG after deletion of all \ufp edges.
      Its symmetric part, comprising the thick edges, is $\NH(\G,\sigma)$.
      (E) The two good quartets.  (F) The single bad quartet.  (G) Examples
      for ugly quartets that cover the remaining \ufp edges that are
      identifiable via quartets.  Panel (H) shows the BMG $(\G,\sigma)$ in a
      different layout that highlights the hourglass chain with right tail
      $z'$. All edges that are \ufp according to
      Prop.~\ref{prop:hourglass-ufp} are in red.  To identify the \ufp edges
      in $(\G,\sigma)$, only the subgraphs in Panel (E), (G) and (H) are
      necessary (cf.\ Def.~\ref{def:hug-edge} and Thm.~\ref{thm:MAIN}).}
    \label{fig:evenhg}
  \end{figure}
  
  Figs.~\ref{fig:hourglasses} and~\ref{fig:evenhg} furthermore show that
  hourglass chains identify false-positive edges that are not associated with
  quartets in the BMG: The BMG in Fig.~\ref{fig:hourglasses}(A) has the \ufp
  edge $xy$, and the BMG in Fig.~\ref{fig:evenhg}(B) contains the \ufp edges
  $x_1y_2$, $x_1z'$ and $x'_1z'$. A careful investigation shows that these
  edges are either not even part of an induced $P_4$ (such as $xy$ in
  Fig.~\ref{fig:hourglasses} and $x'_1z'$ in Fig.~\ref{fig:evenhg}), or at
  least not identifiable as \ufp via good, bad or ugly quartets according to
  Props.~\ref{prop:good_quartet_middle_edge}, \ref{prop:bad} 
  and~\ref{prop:ugly_quartet}, as it is the case for $x_1y_2$ and $x_1z'$ in
  Fig.~\ref{fig:evenhg}.

  \subsection{Four-colored $P_4$s}
  \label{APP:ssec-4colP4}
  
  \citet[Thm.~8]{Geiss:19b} established that the RBMG $(G,\sigma)$ is a co-RBMG,
  i.e., a cograph, if and only if every subgraph induced on three colors is a
  cograph. Therefore, if $(G,\sigma)$ contains an induced 4-colored $P_4$, it
  also contains an induced 3-colored $P_4$. For hourglass-free BMGs
  $(\G,\sigma)$ it is clear that a 4-colored $P_4$ always overlaps with a
  3-colored $P_4$: In this case $\NH(\G,\sigma)$ is obtained by deleting
  middle edges of good quartets and first edges of ugly quartets. Since
  $\NH(\G,\sigma)$ is a cograph, there is no $P_4$ left, and thus at least
  one edge of any 4-colored $P_4$ was among the deleted edges. It is natural
  to ask whether this is true for BMGs in general.
  Fig.~\ref{fig:ex_P4_no_overlap} shows that good and ugly quartets are not
  sufficient on their own: there are 4-colored $P_4$s that do not overlap
  with the middle edge of a good quartet or the first edge of an ugly quartet.
  On the other hand, it is clear that at least one of its edges is \ufp. This
  does not imply, however, that the \ufp edges in a 4-colored $P_4$ are also
  edges of 3-colored $P_4$s.
  
  \begin{figure}[t]
    \begin{center}
      \includegraphics[width=0.6\textwidth]{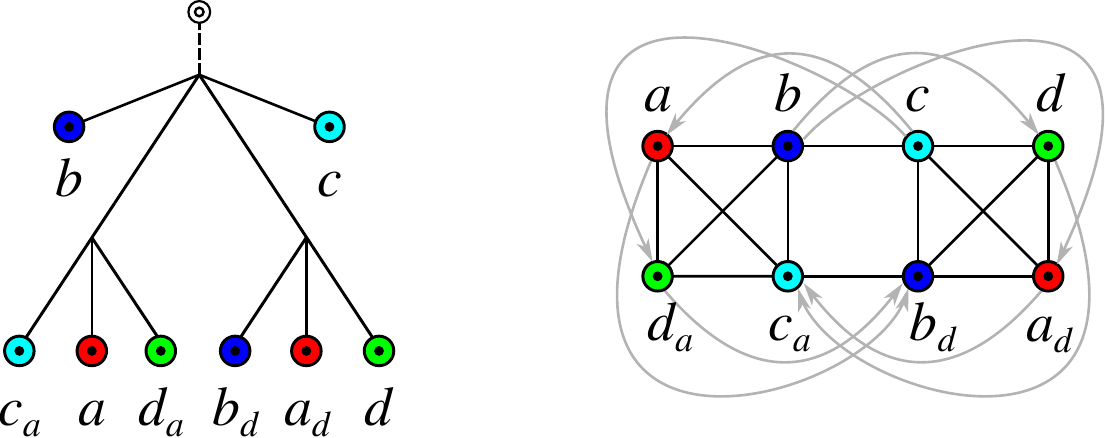}
    \end{center}
    \caption{
      The symmetric part of the BMG $(\G,\sigma)$ contains the 4-colored
      induced $P_4$ $\langle abcd \rangle$. None of its edges is the middle edge
      of a good quartet or the first edge of an ugly quartet. According to
      Lemma~\ref{lem:4col-P4}, there is the bad quartet
      $\langle abca_d \rangle$ that contains as first edge the edge $ab$.}	
    \label{fig:ex_P4_no_overlap}
  \end{figure}
  
  Still, in the context of cograph-editing approaches it is of interest
  whether the 3-colored $P_4$-s are sufficient. In the following we provide
  an affirmative answer.
  \begin{lemma}
    \label{lem:4col-P4}
    Let $(\G,\sigma)$ be a BMG and $\mathscr{P}$ a 4-colored induced
    $P_4$ in the symmetric part of $(\G,\sigma)$. Then at least one of
    the edges of $\mathscr{P}$ is either the middle edge of some good
    quartet or the first edge of a bad or ugly quartet in $(\G,\sigma)$.
  \end{lemma}
  \begin{proof}
    Let $(T,\sigma)$ be an arbitrary tree that explains $(\G,\sigma)$ and
    suppose that $\mathscr{P}\coloneqq\langle abcd \rangle$ is a 4-colored
    induced $P_4$ in the symmetric part $(G,\sigma)$.
    
    If one of the edges $ab$, $bc$, or $cd$ of $\mathscr{P}$ is the middle
    edge of some good quartet or the first edge of some ugly quartet, then we
    are done. Hence, we assume in the following that this is not the case and
    show that at least one of the edges of $\mathscr{P}$ is the first edge in
    a bad quartet.
    
    By contraposition of Prop.~\ref{prop:good_or_ugly}, we have
    $\Scap(a,b)=\emptyset$, $\Scap(b,c)=\emptyset$ and
    $\Scap(c,d)=\emptyset$. We set $v\coloneqq \lca_T(b,c)$ with children
    $v_b,v_c\in\child_{T}(v)$ such that $b\preceq_{T}v_b$ and
    $c\preceq_{T}v_c$, and $w\coloneqq \lca_T(a,b)$ with children
    $w_a, w_b\in \child_T(w)$ such that $a\preceq_{T}w_a$ and
    $b\preceq_{T}w_b$.  Note, that $v,v_b,w$, and $w_b$ are pairwise
    comparable, since they are all ancestors of $b$.
    
    We show that $w=v$. Assume, for contradiction, that (i) $w\prec_T v$ or
    (ii) $v \prec_T w$. In Case (i), we have $w_a\prec_Tw\preceq_T v_b$ and
    thus, $\sigma(a)\in \sigma(L(T(v_b)))$.  Hence, as
    $\Scap(b,c)=\emptyset$, it must hold that
    $\sigma(a)\notin \sigma(L(T(v_c)))$ and
    $\sigma(c)\notin \sigma(L(T(v_b)))$. Lemma~\ref{lem:edge-xy-lca} implies
    $ac\in E(G)$.  But then $\mathscr{P}$ is not an induced $P_4$; a
    contradiction. In Case (ii), we have $v_c\preceq_Tv\preceq w_b$ and thus,
    $\sigma(c)\in \sigma(L(T(w_b)))$. Since $\Scap(a,b)=\emptyset$ we thus
    have $\sigma(c)\notin \sigma(L(T(w_a)))$ and
    $\sigma(a)\notin \sigma(L(T(w_b)))$. By Lemma~\ref{lem:edge-xy-lca},
    $ac\in E(G)$; again a contradiction. Thus $w=v$. Analogous arguments can
    be used to establish $\lca_T(c,d)=v$. We therefore have
    $v=\lca_T(a,b)=\lca_T(b,c)=\lca_T(c,d)$. In the following $v_x$ denotes
    the child of $v$ with $x\preceq_Tv_x$ for $x\in \{a,b,c,d\}$. Note,
    $v_a\ne v_b$, $v_b\ne v_c$ and $v_c\ne v_d$.
    
    We next show that $v_a$, $v_b$, $v_c$, and $v_d$ are pairwise
    distinct. Fist, assume for contradiction that $v_a=v_c$. Together with
    $\Scap(c,d)=\emptyset$, this assumption implies that
    $\sigma(a)\notin \sigma(L(T(v_d)))$ and
    $\sigma(d)\notin \sigma(L(T(v_c)))$. By Lemma~\ref{lem:edge-xy-lca},
    $ad\in E(G)$, contradicting the assumption that $\mathscr{P}$ is an
    induced $P_4$. Hence, $v_a\ne v_c$. By symmetry of $\mathscr{P}$, we can
    use similar arguments to conclude that $v_b\ne v_d$.  Finally, assume for
    contradiction that $v_a= v_d$. Then, $\sigma(d)\in \sigma(L(T(v_a)))$.
    Hence, $\Scap(a,b)=\emptyset$ implies that
    $\sigma(d)\notin \sigma(L(T(v_b)))$ and
    $\sigma(b)\notin \sigma(L(T(v_d)))$.  Again Lemma~\ref{lem:edge-xy-lca}
    implies $bd\in E(G)$; a contradiction.  In summary, $v_a$, $v_b$, $v_c$,
    and $v_d$ must be pairwise distinct.
    
    We claim $\sigma(c)\in \sigma(L(T(v_a)))$. Since $ad\notin E(G)$ and
    $\lca_{T}(a,d)=v$, Lemma~\ref{lem:edge-xy-lca} implies that
    $\sigma(a)\in \sigma(L(T(v_d)))$ or $\sigma(d)\in \sigma(L(T(v_a)))$.  By
    symmetry of $\mathscr{P}$, we can w.l.o.g.\ assume that
    $\sigma(a)\in \sigma(L(T(v_d)))$ and thus, there is a vertex
    $a_d\in L(T(v_d))$ with $\sigma(a_d)=\sigma(a)$. In this case,
    $\Scap(c,d)=\emptyset$ implies that $\sigma(a)\notin
    \sigma(L(T(v_c)))$. This together with $ac\notin E(G)$ and
    Lemma~\ref{lem:edge-xy-lca} implies that
    $\sigma(c)\in \sigma(L(T(v_a)))$.
    
    We claim $\sigma(d)\in \sigma(L(T(v_a)))$.  We assume for contradiction
    that this is not the case and show that this implies the existence of an
    ugly quartet $\langle cdc'a'\rangle$ containing $cd$ as its first edge,
    which leads to a contradiction to our initial assumption that none of the
    edges in $\mathscr{P}$ is the first, resp., middle edge of an ugly,
    resp., good quartet.  To see this, note that
    $\sigma(a),\sigma(c)\in\sigma(L(T(v_a)))$ and Lemma~\ref{lem:exEdge}
    imply that there is an edge $a'c'$ for two vertices $a',c'\prec_T v_a$
    with $\sigma(a')=\sigma(a)$ and $\sigma(c')=\sigma(c)$. Since
    $\sigma(a)=\sigma(a')$ and
    $\lca_T(a',c')\preceq_Tv_a\prec_T v=\lca_T(a',c)$, we have
    $a'c\notin E(G)$. Since $\sigma(a_d)=\sigma(a')$ and
    $\lca_T(a_d,d)\preceq_Tv_d\prec_T v=\lca_T(a',d)$, we have
    $a'd\notin E(G)$. Now, $\Scap(c,d)$ implies that
    $\sigma(c)\notin \sigma(L(T(v_d)))$.  This and
    $\sigma(d)\notin \sigma(L(T(v_a)))$ together with
    Lemma~\ref{lem:edge-xy-lca} implies that there is an edge $c'd\in
    E(G)$. Thus, we obtain the ugly quartet $\langle cdc'a'\rangle$ and
    hence, the desired contradiction.  Therefore,
    $\sigma(d)\in \sigma(L(T(v_a)))$. Because of $\Scap(a,b)=\emptyset$ we
    also have $\sigma(d)\notin \sigma(L(T(v_b)))$.
    
    Since $\sigma(d)\in \sigma(L(T(v_a)))$, there is a vertex
    $d_a\preceq v_a$ with $\sigma(d_a)=\sigma(d)$. Moreover,
    $\sigma(b)\notin \sigma(L(T(v_a))$ and
    $\sigma(d)\notin \sigma(L(T(v_b)))$ together with
    Lemma~\ref{lem:edge-xy-lca} implies that $bd_a\in E(G)$. Furthermore,
    $\sigma(c)\in \sigma(L(T(v_a)))$ and Lemma~\ref{lem:edge-xy-lca} imply
    that $cd_a\notin E(G)$. Now, $\Scap(c,d)=\emptyset$ implies
    $\sigma(d)\notin \sigma(L(T(v_c)))$ and therefore,
    $\lca_T(c,d_a)=v\preceq \lca_T(c,d')$ for all $d'\in
    L[\sigma(d)]$. Hence, $(c,d_a)\in E(\G)$.
    
    In summary, $\langle dcbd_a \rangle$ is an induced $P_4$ in $G$. By
    \cite[Lemma~32]{Geiss:19b}, every such induced $P_4$ forms either a good,
    bad, or ugly quartet in $(\G,\sigma)$ and, since $(c,d_a)\in E(\G)$, we
    can conclude that $\langle dcbd_a \rangle$ is a bad quartet with first
    edge $cd$, which completes the proof.
  \end{proof}
  
  \begin{corollary}{\cite[Thm.~8]{Geiss:19b}}
    Let $(G,\sigma)$ be an RBMG. Then, $(G,\sigma)$ is a cograph if and only
    if all subgraphs induced by three colors are cographs.
  \end{corollary}
  \begin{proof}
    If $(G,\sigma)$ is a cograph, then all its induced subgraphs are also
    cographs \cite{CORNEIL:81}. Conversely, if $(G,\sigma)$ is not a cograph,
    then it contains at least one induced $P_4$. By Lemma~\ref{lem:4col-P4},
    $(G,\sigma)$ cannot contain only 4-colored $P_4$s and therefore the
    restriction to at least one combination of three colors contains a $P_4$
    and is thus not a cograph.
  \end{proof}

\end{appendix}

\bibliographystyle{spbasic.bst}
\bibliography{references}

\end{document}